\documentclass[aps,showpacs,pra,twocolumn]{revtex4-1}

\usepackage{exscale}
\usepackage{graphicx}
\usepackage{amsmath}
\usepackage{latexsym}
\usepackage[caption=false]{subfig}% justication of captions
\usepackage{amsfonts}
\usepackage{amssymb}
\usepackage{bbm}
\usepackage{times}
\usepackage[T1]{fontenc}
\usepackage{lipsum}
\usepackage{amsthm}
\usepackage{fancyhdr,txfonts,bbm}

\usepackage{graphics}
\usepackage{bbold}
\usepackage{bm}
\usepackage{subfig}
\usepackage{color}

\usepackage{epsfig,pstricks}
\usepackage{changes}
\definechangesauthor[name={Markus}, color=red]{M}
\definechangesauthor[name={Dario}, color=magenta]{D}

\newtheorem{theorem}{Theorem}

\newtheorem{definition}{Definition}
\newtheorem{proposition}{Proposition}

\newcommand{\tr}[1]{\mbox{Tr}\left[ #1\right]}

\newcommand{{\Cd}}{{\mathbb{C}^d}}
\newcommand{{\C}}{{\mathbb{C}}}

\newcommand{\sign}{\text{sign}}

\newcommand{\Int}{\mathrm{int}}

\DeclareMathOperator{\Tr}{Tr}
\DeclareMathOperator{\atanh}{atanh}

\newcommand{\I}{\mbox{Im}}
\newcommand{\bra}[1]{\langle #1|}
\newcommand{\ket}[1]{|#1\rangle}
\newcommand{\braket}[2]{\langle #1|#2\rangle}
\newcommand{\ketbra}[2]{| #1\rangle\!\langle #2|}

\newcommand{\trA}[1]{\mbox{Tr}_A\left[ #1\right]}
\newcommand{\trB}[1]{\mbox{Tr}_B\left[ #1\right]}
\newcommand{\trS}[1]{\mbox{Tr}_S \left[ #1\right]}
\newcommand{\expp}[1]{e^{#1}}
\newcommand{\exppp}[1]{\mbox{exp} \left[ #1 \right]}

\begin{document}

\title{Witnessing non-Markovian dynamics through correlations}

\author{Dario De Santis$^1$, Markus Johansson$^1$, Bogna Bylicka$^1$, Nadja K. Bernardes$^{2,3}$, and Antonio Ac\'\i n$^{1,4}$}
\affiliation{$^1$ICFO-Institut de Ciencies Fotoniques, The Barcelona Institute of Science and Technology,
08860 Castelldefels (Barcelona), Spain\\
$^2$ Departamento  de  F\'\i sica,  Universidade  Federal  de  Minas  Gerais,
Belo  Horizonte,  Caixa  Postal  702, 30161-970, Brazil\\
$^3$ Departamento de F\'\i sica, Universidade Federal de Pernambuco, 50670-901 Recife-PE, Brazil\\ 
$^4$ICREA--Institucio Catalana de Recerca i
Estudis Avan\c{c}ats, Lluis Companys 23, 08010 Barcelona, Spain}

\date{\today}
\begin{abstract}
Non-Markovian effects in an open-system dynamics are usually associated to information backflows from the environment to the system. However, the way these backflows manifest and how to detect them is unclear. A natural approach is to study the backflow in terms of the correlations the evolving system displays with another unperturbed system during the dynamics. In this work, we study the power of this approach to witness non-Markovian dynamics using different correlation measures. We identify simple dynamics where the failure of completely-positive divisibility is in one-to-one correspondence with a correlation backflow. We then focus on specific correlation measures, such as those based on entanglement and the mutual information, and identify their strengths and limitations. We conclude with a study of a recently introduced correlation measures based on state distinguishability and see how, for these measures, adding an extra auxiliary system enlarges the set of detectable non-Markovian dynamics.
\end{abstract}

\maketitle

\section{Introduction}\label{I0}

The dynamics of open quantum systems \cite{book_B&P, book_Weiss, book_R&H} has been investigated extensively in recent years for both fundamental and applicative reasons.
In particular 
the problem of understanding and characterizing memoryless dynamics, the so-called Markovian regime, and dynamics exhibiting memory effects, the non-Markovian regime, have been considered in a wide range of different ways (for extended reviews see \cite{rev_RHP, rev_BLP}). While a unique agreed upon concept of quantum Markovianity does not exist it is frequently identified with the property of { \it Completely Positive divisibility} (CP-divisibility).
An evolution is CP-divisible if between every two points in time it can be described by a CP-map.  
This idea generalises the semigroup property \cite{GKS, L} of classical Markovian processes.

Intuitively, one expects that non-Markovian effects are associated to a backflow of information from the environment to the system. Several approaches have been pursued to put this intuition in rigorous terms. A standard procedure, see for example~\cite{BLP0,Luo,Buscemi&Datta,bogna}, consists of considering operational quantities that are monotonically non-increasing under CP maps.
An increase of any such quantities implies that the evolution is non-Markovian, although the converse is not true in general. 
One of the quantities considered in the context of non-Markovian characterization is correlations, a fundamental concept for our understanding of quantum theory and also a resource for many quantum information protocols. The general idea of this approach consists of monitoring the evolution of the correlations between a system subjected to a dynamics and an additional particle that does not take part in the evolution. If at some point a correlation backflow, that is, an increase in the correlations quantified by a given correlation measure, is observed, then the dynamics must be {non-Markovian}. 

The main goal of this work is to understand the power of correlations to witness non-Markovian evolutions. In particular, we study the strengths and weaknesses of several well-known correlation measures for this task. We derive several results that improve our understanding of this question.
First, it is shown how for a class of differentiable evolutions termed single parameter, which includes relevant examples such as depolarization, dephasing, and amplitude damping, any continuously differentiable correlation measure increases during non-Markovian dynamics unless it is time independent on the whole image of the preceding evolution.
Secondly, we focus on two fundamental quantum correlation measures: entanglement measures and (quantum) mutual information. For the first, we provide a simple argument explaining how it fails to witness non-Markovianity in many situations. For the second, we study its behaviour in different scenarios. 
We first show that the mutual information witnesses the non-Markovianity of any bijective unital and non-P-divisible dynamics on a qubit.
We then provide several examples of non-Markovian dynamics where no mutual information backflow is observed when using maximally entangled states. For some of these examples, we demonstrate that a backflow in the mutual information does appear when using non-maximally entangled states, in some cases even arbitrarily weakly entangled pure states. This highlight how a high degree of initial correlations is not necessarily beneficial for the detection of non-Markovianity when using the mutual information as a witness. 
Lastly, based on the results {in~\cite{previous}}, we construct examples of non-Markovian dynamics that do not display any backflow of the mutual information.
The last part of our work is devoted to a detailed analysis of the correlation measure introduced in {Ref.~\cite{previous}} that detects almost all non-Markovian dynamics, more precisely all but a zero-measure set consisting of all those dynamics that are not bijective in finite time intervals. Interestingly, for all these cases, an initial state arbitrarily close to a product state suffices for the non-Markovianity detection.
We also discuss how, in order to fully exploit the power of this correlation measure, one has to consider an enlarged scenario involving a second ancilla $A'$ that is not subjected to the evolution either, and look for correlation backflows along the bipartition $SA'$ versus $A$. 

The outline of the manuscript is as follows. In Section \ref{II0nmdyn} we define the notation used in the article, we introduce the key concepts of open-system dynamics and present two classes of evolutions that are very useful for the derivation of our results: the single parameter and random unitary dynamics. Section \ref{III'0} provides the definition of a correlation measure, describes the two settings considered in this work and discusses a family of dynamics for which any correlation measure increases in the presence of non-Markovian effects. Section \ref{IV'0} contains the results regarding entanglement measures, while in Section \ref{V'0} we study in detail the case of the mutual information. In Section \ref{V0dario} the new class of correlation measures based on state distinguishability is analyzed. The paper concludes with a discussion of the results and open questions.

{
Note added: while completing the work we became aware of the results in \cite{Janek} where it is shown how the negativity, a computable entanglement measure, is able to detect all the bijective non-Markovian dynamics, and even non-bijective in the case of qubits, in a tripartite configuration. We come back to this point below and discuss it in the context of our results.}

\section{Non-Markovian dynamics}\label{II0nmdyn}
In this section we describe the mathematical representation of open system dynamics and the notation used in what follows.  In Section~\ref{IIA} we focus on the description of the time evolution of open quantum systems and in Section~\ref{IIIB} we introduce what we call ``single parameter'' evolutions, a family of differentiable evolutions. Finally, in Section~\ref{RUD} we present the random unitary evolutions, {an} example of dynamics defined by the random application of unitary operations.

We consider a quantum system $S$ with an associated finite dimensional Hilbert space $\mathcal{H}_S$ of dimension $d$. The set of bounded linear operators acting on $\mathcal{H}_S$ is denoted by $B(\mathcal{H}_S)$ and the subset of states is denoted by $S(\mathcal{H}_S)$. An ancillary system $A$ with Hilbert space $\mathcal{H}_A$ is introduced and the corresponding set of bounded linear operators and set of states are denoted by $B(\mathcal{H}_A)$ and $S(\mathcal{H}_A)$, respectively. The set of bounded operators on the combined Hilbert space $\mathcal{H}_A\otimes \mathcal{H}_S$ is denoted by $B(\mathcal{H}_A\otimes \mathcal{H}_S)$ and the set of states is denoted by $S(\mathcal{H}_A\otimes \mathcal{H}_S)$.

The evolution of $S$ from initial time $t=0$ to a later time $t$ is described by a dynamical map, i.e., a linear operator $\Lambda_t:B(\mathcal{H}_S)\rightarrow B(\mathcal{H}_S)$ that is completely positive and trace preserving (CPTP). The dynamics of the system is thus described by the family of maps $\{\Lambda_t\}_t$ parametrized by $t$.
An important concept for the study of non-Markovian effects is that of divisibility of the dynamical map $\Lambda_t$, as well as Positive (P) and Completely Positive (CP)-divisibility in terms of intermediate maps $V_{s,t}$.

\begin{definition}\label{defCPdiv}
A dynamical map $\Lambda_s$ is called divisible (P/CP-divisible) if it can
be expressed as a sequence of linear trace preserving (P/CP) maps $\Lambda_s=V_{s,t}\Lambda_t$, where $V_{s,t}$ is a linear trace preserving (P/CP) map, for any $0\leq t\leq s$.
\end{definition}

CP-divisibility of the family of dynamical maps has been taken by many authors as the definition of Markovian evolution \cite{book_R&H}.  It is in fact the definition we adopt in this work and, in what follows, whenever it is said that an evolution is non-Markovian, it means that it is not CP-divisible. {The idea behind this is that in a CP-divisible process the environment $E$ of the system can be disregarded and an intermediate map $V_{s,t}$ can be physically realized just by acting locally on the system and a second environment $E'$ in a product state with $S$ at time $t$ (see Fig. \ref{figenv}). If $V_{s,t}$ is not CP it cannot result from a local process involving only $S$ and $E'$. 
Hence, the state of $S-E$ at time $t$ holds information that is crucial for the subsequent evolution of $S$ up to time $s$. } This definition of Markovian evolution is the analogue of the classical Markov process where the evolution at each time interval only depends on the state of the system at the beginning of the interval. In other words both definitions capture the idea that no memory of previous events is needed to describe the process.

\begin{figure}
\includegraphics[width=0.3\textwidth]{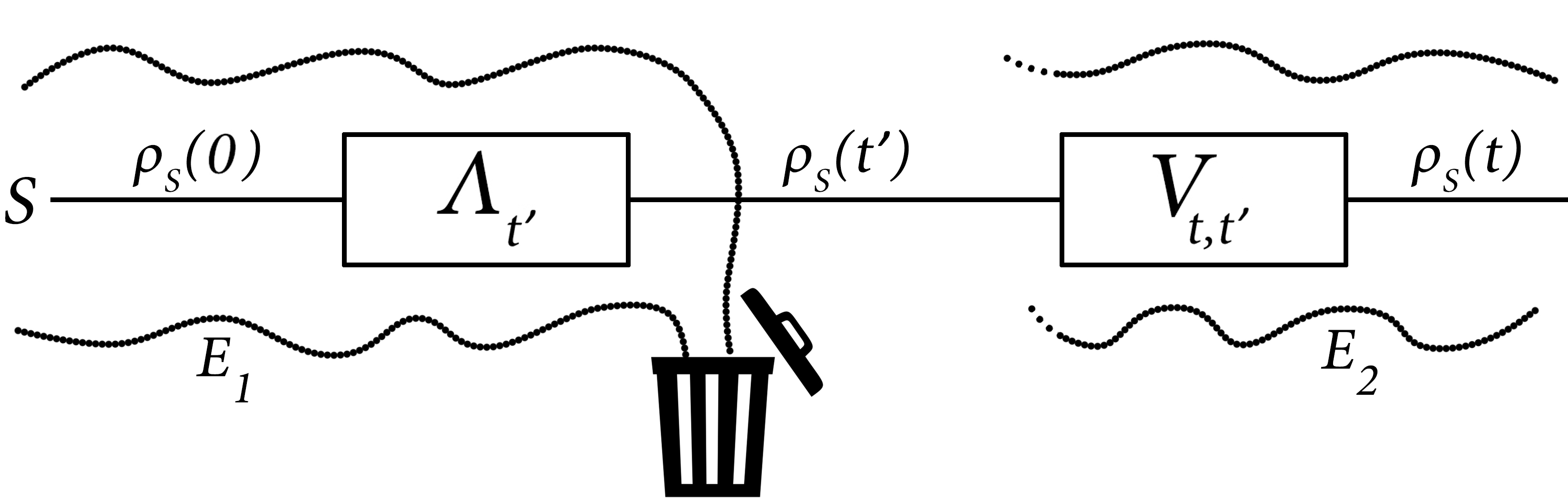}\vspace{0.2cm}
\includegraphics[width=0.45\textwidth]{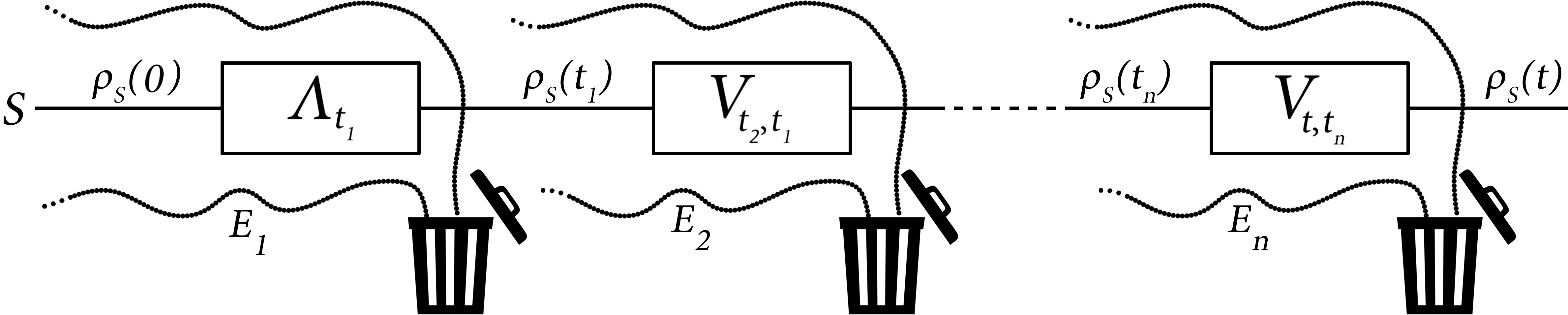}
\caption{ {Top: The evolution given by $\Lambda_t =V_{t,t'} \Lambda_{t'}$, where $V_{t,t'}$ is CP, can be represented by the interaction of $S$ with a first environment $E_1$, in a product state with $S$ at time 0, that after time $t'$ is discarded and replaced with a second environment $E_2$ in a product state with $S$ at time $t'$. Thus, any information discarded at time $t'$ cannot be recovered afterwards.  Bottom: We divide the evolution of a Markovian dynamical map $\Lambda_t$ in $n$ time intervals $(t_{i-1},t_i)$, where $i=1,\dots, n$, $t_0=0$ and $t_n=t$. The evolution during the $i$-th interval $(t_{i-1},t_i)$ is given by a CP intermediate map $V_{t_i,t_{i-1}}$ and it can be represented by the interaction of $S$ with an environment $E_i$ uncorrelated with $S$ at time $t_{i-1}$, while the previous environment $E_{i-1}$ is discarded.} }\label{figenv}

\end{figure}

\subsection{The generator of the evolution and the intermediate map}\label{IIA}

For a differentiable evolution, any dynamical map $\Lambda_t$ and any intermediate map $V_{s,t}$ can be expressed as time ordered exponentials
\begin{eqnarray}
\label{basics}
\Lambda_t(\rho)=\mathcal{T}e^{\int_{0}^{t}\mathcal{L}_{\tau}d\tau},\phantom{uu} V_{s,t}=\mathcal{T}e^{\int_{t}^{s}\mathcal{L}_{\tau}d\tau},
\end{eqnarray}
where $\mathcal{L}_{t}$ is the generator of the evolution as first defined by 
Gorini et. al. \cite{GKS} and Lindblad \cite{L} for dynamical semigroups and later extended to general divisible maps,

\begin{eqnarray}\label{geno}
\mathcal{L}_t(\rho)\equiv i[H(t),\rho]+\sum_k\gamma_{k}(t)\left( G_k(t)  \rho G_k^{\dagger}(t)-\frac{1}{2}\left\{G_k^{\dagger}(t) G_k(t),\rho\right\}\right).\nonumber\\
\end{eqnarray}
Here $\gamma_k(t)$ are real time dependent functions, $G_k(t)$ are time dependent operators and $H(t)$ is the Hamiltonian of the system, thus also a Hermitian possibly time-dependent operator. The Hamiltonian term of the generator describes the unitary part of the dynamics generated by $H(t)$ and the second term describes the dissipative part of the dynamics generated by the $G_k(t)$.
The generator $\mathcal{L}_t$ gives rise to a Markovian evolution if and only if it can be written in a form where $\gamma_k(t)\geq 0$ for all $k$. See e.g. Theorem 5.1 of Ref. \cite{book_R&H} for a proof.

The generator $\mathcal{L}_t$ can be defined in terms of the intermediate map as  $\mathcal{L}_t\equiv\frac{d V_{s,t}}{ds}\big|_{s=t}$.
Often, it is convenient to describe the intermediate map $\mathcal{I}_A\otimes V_{s,t}$ and the generator $\mathcal{L}_t$ by how they act on a basis of $B(\mathcal{H}_A\otimes \mathcal{H}_S)$. 
Such a basis can be constructed from operators of the form $\chi_A\otimes{\chi_S}$ where $\chi_A$ is an operator on $\mathcal{H}_A$ and $\chi_S$ is an operator on $\mathcal{H}_S$.
Let the dimension of $\mathcal{H}_S$ be $d_S$ and let $\chi_{Sk}$ be $d_S^2-1$ traceless Hermitian operators such that $\Tr[\chi_{Sk}\chi_{Sl}]=\delta_{kl}d_S$. Likewise let $d_A$ be the dimension of $\mathcal{H}_A$ and let $\chi_{Ak}$ be $d_A^2-1$ traceless Hermitian operators such that $\Tr[\chi_{Ak}\chi_{Al}]=\delta_{kl}d_A$.
Then one can choose an orthonormal basis $\{e_i\}_{i=0}^{d_A^2d_S^2-1}$ for $B(\mathcal{H}_A\otimes \mathcal{H}_S)$ by constructing the basis vectors $e_i$ as the tensor products $\chi_{Ai}\otimes{\chi_{Sj}}$, $\chi_{Ai}\otimes{\mathbbm{1}_{S}}$, $\mathbbm{1}_A\otimes{\chi_{Sj}}$, and $\mathbbm{1}_A\otimes{\mathbbm{1}_S}$ for all $i,j$. We denote $e_0\equiv \mathbbm{1}_A\otimes{\mathbbm{1}_S}$. 
A state $\rho\in S(\mathcal{H}_A\otimes \mathcal{H}_S)$ can be represented in this basis by coordinates given by the real numbers $a_i=\frac{1}{d_Ad_S}\Tr(\rho e_i)$, i.e., 
\begin{eqnarray}\label{stateAS}
\rho=\frac{1}{d_Ad_S} \mathbbm{1}_A\otimes{\mathbbm{1}_S}+\sum_{i=1}^{d_A^2d_S^2-1} a_i e_i.
\end{eqnarray}

We can now describe the intermediate map $\mathcal{I}_A\otimes V_{s,t}$ by how it acts on each basis elements $e_i$. For this purpose we define 
\begin{eqnarray}\label{comp}V_{ij}(s,t)\equiv \tr{e_i\mathcal{I}_A\otimes V_{s,t}(e_j)}.\end{eqnarray}
Note that $V_{ij}(s,t)$ is real for all $i,j$ since $V_{s,t}$ is hermiticity preserving.
Moreover, since the map $V_{s,t}$ is trace preserving it follows that $V_{00}(s,t)=1$ and $V_{0j}(s,t)=0$ for $j\neq{0}$.
Let the  coordinates {$\bar{a}\equiv\{a_i\}$}, where $a_0=1/d_Ad_S$, describe a state at time $t$. This state is mapped by $\mathcal{I}_A\otimes V_{s,t}$ to coordinates $\bar{a}(s)=\{a_i(s)\equiv\sum_{j}V_{ij}(s,t)a_j \}$.
Analogously to Eq. (\ref{comp}), we define the time derivatives of the components $V_{ij}(s,t)$ as
\begin{eqnarray}\label{kliy}\frac{d V_{ij}(s,t)}{ds}\Big|_{s=t}\equiv \tr{e_i\mathcal{I}_A\otimes \frac{d V_{s,t}}{ds}(e_j)}\Big|_{s=t}  \, .\end{eqnarray}

\subsection{Single parameter evolutions}\label{IIIB}

A simple family of open-system dynamics is given by those differentiable evolutions $\{\Lambda_t\}_t$ such that for all $t$ the generator $\mathcal{L}_t$ can be expressed as 

\begin{eqnarray}
\mathcal{L}_t(\rho)=i[H(t),\rho]+\gamma(t)\sum_k\left( G_k  \rho G_k^{\dagger}-\frac{1}{2}\left\{G_k^{\dagger} G_k,\rho\right\}\right),\nonumber\\
\end{eqnarray}
where $G_k$ are time-independent operators and $\gamma(t)$ is a continuous function of time. We call those evolutions with this property {\it single parameter} since $\gamma(t)$ alone describes the time dependence of the dissipative part. Paradigmatic examples of single parameter evolutions are depolarization, as well as dephasing and amplitude damping in a time independent basis.
Using the representation defined in Eq. (\ref{kliy}) we can state the single parameter property as
\begin{eqnarray}
\frac{d V_{ij}(s,t)}{ds}\Big{|}_{s=t}=g_{ij}\gamma(t)+h_{ij}(t),
\end{eqnarray}
where $h_{ij}(t)$ is corresponds to the action of $H(t)$ and  the time independent parameters $g_{ij}$ to $\sum_kg_k\left( G_k  \rho G_k^{\dagger}-\frac{1}{2}\left\{G_k^{\dagger} G_k,\rho\right\}\right)$.

An important property of single-parameter evolutions is that  
they can be divided into CP-divisible and not P-divisible time intervals, that is, they never present intermediate maps that are not CP but P.
%DARIO: We provide a proof for proposition 1 for infinitesimal time intervals. Therefore, can we really say that "they never present intermediate maps that are not CP but P"? Is it true even if we take a long enough time interval?
\begin{proposition}\label{prop2}
Let $\{\Lambda_t\}_t$ be single parameter. 
Then the intermediate map $V_{t'+\epsilon',t'}$ is either CP or not P for any $t'$ and sufficiently small $\epsilon'>0$.
\end{proposition}
\begin{proof}
An intermediate map $V_{s,t}$ is CP if the Choi matrix $C_{V_{s,t}}$ has non-negative eigenvalues \cite{choi,jammy}. Moreover, $V_{s,t}$ is P if $V_{s,t}(\rho)$ is positive semidefinite for all
positive semidefinite Hermitian trace one matrices $\rho\in B(\mathcal{H}_S)$. In particular, if $V_{s,t}$ is P, $V_{s,t}(\rho)$ is positive semidefinite for any positive semidefinite rank one $\rho$.

The eigenvalues of $V_{s,t}(\rho)$ for a rank-one pure state $\rho$ and the eigenvalues of $C_{V_{s,t}}$ are functions of the parameters $\bar{a}$ and $V_{ij}(s,t)$. Since we assumed that $\gamma(t)$ is continuous it follows that $V_{s,t}(\rho)$ and $C_{V_{s,t}}$ are both continuously differentiable. This in turn implies that the eigenvalues of $V_{s,t}(\rho)$ and $C_{V_{s,t}}$ can be described by continuously differentiable functions \cite{rellich}. 

If $\gamma(t)\neq 0$, so that $V_{s,t}$ is not unitary, there always exists at least one eigenvalue {$\lambda_{C}(V_{s,t})$ } of the Choi matrix that is zero and has a nonzero time derivative at $s=t$. There also exists at least one rank-one $\rho$ with eigenvalue { $\lambda(\rho)$ that is zero and its evolution $\lambda(\rho(s))$, i.e., given by $\rho(s)=\mathcal{I}_A\otimes V_{s,t}(\rho)$, has a nonzero time derivative at $s=t$  (see Appendix \ref{dermu}).
}

%If $\gamma(t)\neq 0$, so that $V_{s,t}$ is not unitary, there always exists at least one eigenvalue $\lambda_{C}(V_{ij}(s,t))$ of the Choi matrix that is zero and has a nonzero time derivative at $s=t$. There also exists at least one rank-one $\rho$ with eigenvalue \added{ $\lambda_{\rho}(\rho,V_{ij}(s,t))$} that is zero and has a nonzero time derivative at $s=t$  (see Appendix \ref{dermu}).

 {
We start by considering the time derivatives  $\frac{d \lambda_{C}(V_{s,t})}{ds}\big{|}_{s=t}=\sum_{ij}\frac{d \lambda_{C}(V_{s,t})}{dV_{ij}(s,t)}\big{|}_{s=t} g_{ij}\gamma(t)$ and $\frac{d \lambda (\rho(s)) }{ds}\big{|}_{s=t}=\sum_{ij}\frac{d \lambda (\rho(s) ) }{dV_{ij}(s,t)}\big{|}_{s=t} g_{ij}\gamma(t)$. Here we used that the eigenvalues are invariant under unitary evolution and thus $\sum_{ij}\frac{d \lambda_{C}(V_{s,t})}{dV_{ij}(s,t)} \big{|}_{s=t} h_{ij}(t)=0$ and $\sum_{ij}\frac{ d \lambda (\rho(s) )}{dV_{ij}(s,t)}\big{|}_{s=t} h_{ij}(t)=0$. If the time derivatives are non-zero they are proportional to $\gamma(t)$. If $V_{t+\epsilon,t}$ is CP for any sufficiently small $\epsilon$ it is clear that $\frac{d \lambda_{C}(V_{s,t})}{ds}\big{|}_{s=t}>0$ and $\frac{d \lambda (\rho(s) ) }{ds}\big{|}_{s=t}>0$. Then if a $t'$ exist such that $\sign[\gamma(t')]=-\sign[\gamma(t)]$, there exists an $\epsilon'$ such that $V_{t'+\epsilon',t'}$ is neither CP or P since $\frac{d \lambda_{C}(V_{s',t'} )}{ds'}\big{|}_{s'=t'}<0$ and $\frac{d \lambda (\rho(s'))}{ds'}\big{|}_{s'=t'}<0$. From this follows that it is impossible for $V_{t'+\epsilon',t'}$ to be P but not CP for sufficiently small $\epsilon'$.
}

\end{proof}

Thus, if the assumptions of Proposition \ref{prop2} hold we can conclude that the dynamics can be divided into closed intervals where $\gamma(t)\geq 0$ and open time intervals for which $\gamma(t)< 0$. In the closed intervals  where $\gamma(t)\geq 0$ the dynamics is CP-divisible and in the open intervals where $\gamma(t)<0$ the dynamics is not P-divisible.

Notice that, if $\Lambda_t$ is CP-divisible in a time interval, we do not necessarily mean that it is Markovian. Indeed, in order for $\Lambda_t$ to be Markovian, the corresponding intermediate map {$V_{s,t}$ has to be CP for any $s$ and $t$ such that $0\leq t \leq s$}  (see Definition \ref{defCPdiv}).

\subsection{Random unitary evolutions}\label{RUD}

Unital dynamical maps $\Lambda_t$ on a finite dimensional Hilbert space $\mathcal{H}_S$ are defined by the property of preserving the identity operator, i.e., $\Lambda_t(\mathbbm{1}_S)=\mathbbm{1}_S$ for any $t\geq 0$. In this section we introduce a useful class of unital evolutions that we repeatedly use in the following: the random unitary dynamics.

We consider a generic finite-dimensional Hilbert space $\mathcal{H}_S$, where dim$(\mathcal{H}_S)=d_S$.
The CPTP map $\Lambda: B(\mathcal{H}_S) \rightarrow B(\mathcal{H}_S)$ represents a random unitary channel on $S$ if its action can be written as
\begin{equation}\label{RUgen}
\Lambda (\rho_S) = \sum_{k} p_k \,  U_k \, \rho_S  \, U_k^\dagger \, ,
\end{equation}
where $\{p_k\}_k$ is a probability distribution and $\{U_k\}_k$ is a set of unitary transformations \cite{Wuda}. If $d_S=2$, i.e., for qubit evolutions, any unital channel is random unitary \cite{LaSt}. However, when $d_S\geq 3$, the set of random unitary channels is strictly included in the set of unital channels.

In the following, we consider random unitary dynamical maps $ \Lambda_t :  B(\mathcal{H}_S) \rightarrow B(\mathcal{H}_S)$ of the form
\begin{equation}\label{RU}
\Lambda_t (\rho_S(0)) = \sum_{k=0}^{N_S} p_k(t)\,  \sigma_k \, \rho_S (0) \, \sigma_k \, ,
\end{equation}
where the unitary operators $\sigma_k$, for $k=1,\dots,N_S$ and $N_S=d_S^2-1$, constitute a basis of elements for the Hermitian and traceless operators in $B(\mathcal{H}_S)$ such that $\sigma_0=\mathbbm{1}_S$ and $\tr{\sigma_i\sigma_j^\dagger}= d_S \delta_{ij}$ for $i,j=0,1,\dots,d_S^2-1$. Since $\{ p_k(t)\}_k$ is a probability distribution, we have that $\Sigma_{k=0}^{N_S} p_k(t)=1$ and $p_k(t)\geq 0$ for any $k=0,\dots,N_S$. Moreover, from the condition $\Lambda_0=\mathcal{I}_S$ it follows that $p_0(0)=1$. For the sake of simplicity, when we refer to a random unitary dynamics we consider dynamical maps of the form given in Eq. (\ref{RU}).

\subsubsection{Random unitary evolutions for qubits}

In the case that the evolution is defined on a qubit, a random unitary dynamics is given by the action of the corresponding dynamical map $\Lambda_t$ on the Pauli operators $\sigma_x, \sigma_y$ and $\sigma_z$
\begin{eqnarray}\label{randomunit}
\Lambda_t(\sigma_x)&&=\exppp{-\int_{0}^{t}( \gamma_z(\tau)+\gamma_y(\tau))d\tau}\sigma_x,\nonumber\\
\Lambda_t(\sigma_y)&&=\exppp{-\int_{0}^{t}( \gamma_z(\tau)+\gamma_x(\tau))d\tau}\sigma_y,\nonumber\\
\Lambda_t(\sigma_z)&&=\exppp{-\int_{0}^{t}( \gamma_x(\tau)+\gamma_y(\tau))d\tau}\sigma_z,\nonumber\\
\Lambda_t(\mathbbm{1})&&=\mathbbm{1},
\end{eqnarray}
where $\gamma_k(t)$ are real valued functions of $t$. Any random unitary dynamical map is bijective for all $t$ and the intermediate {maps are} given by
\begin{eqnarray}\label{ranu}
V_{s,t}(\sigma_x)&&=\exppp{-\int_{t}^{s}(\gamma_z(\tau)+\gamma_y(\tau))d\tau}\sigma_x,\nonumber\\
V_{s,t}(\sigma_y)&&=\exppp{-\int_{t}^{s}(\gamma_z(\tau)+\gamma_x(\tau))d\tau}\sigma_y,\nonumber\\
V_{s,t}(\sigma_z)&&=\exppp{-\int_{t}^{s}(\gamma_x(\tau)+\gamma_y(\tau))d\tau}\sigma_z,\nonumber\\
V_{s,t}(\mathbbm{1})&&=\mathbbm{1}.
\end{eqnarray}

The corresponding generator of the evolution (see Eqs. (\ref{basics}) and (\ref{geno})) is
\begin{eqnarray}\label{genoru}
\mathcal{L}_t(\rho)=\sum_{k=x,y,z}\gamma_k(t)(\sigma_k\rho\sigma_k-\rho).
\end{eqnarray}
We define 
\begin{equation}\label{Aij}
A_{ij} (t) \equiv \exppp{-2 \int_{0}^t (\gamma_i(\tau) + \gamma_j(\tau)) \, d\tau} \geq 0 \, ,
\end{equation}
for $i\neq j$ and $i,j\in{x,y,z}$. The map defined by Eqs. (\ref{randomunit}) is CPTP if and only if 
\begin{equation}\label{CPmap}
B_{ijk}(t)  \equiv 1 + A_{ij}(t) - A_{jk}(t) - A_{ki}(t) \geq 0 \, ,
\end{equation}
for any cyclic permutation of $i,j$ and $k$, i.e., for $(i,j,k)=(x,y,z),(y,z,x),(z,x,y)$ \cite{Hall}.
The dynamics is CP-divisible if and only if $\gamma_k(t)\geq 0$ for all $k$ and $t\geq 0$. Moreover, the dynamics is P-divisible if and only if the conditions
\begin{eqnarray}\label{rattes}
\gamma_x(t)+\gamma_y(t)\geq{0},\nonumber\\
\gamma_y(t)+\gamma_z(t)\geq{0},\nonumber\\ 
\gamma_z(t)+\gamma_x(t)\geq{0},
\end{eqnarray}
are satisfied since the intermediate maps are then contractive in the trace norm~\cite{koss1,koss2,ruskai,darekk}. 
The similarity between Eqs. (\ref{randomunit}) and (\ref{Aij}) allows us to use the following equivalent definition for the random unitary qubit dynamics $\Lambda_t$
\begin{equation}\label{lambdaAij}
\Lambda_t(\sigma_i)=\lambda_i(t)\sigma_i \, \, , \,\,\, \mbox{ where }\lambda_i(t)=\sqrt{A_{jk}(t)} \, ,
\end{equation}
where $(i,j,k)=(x,y,z),(y,z,x),(z,x,y)$ and $\Lambda_t(\mathbbm{1})=\mathbbm{1}$.

Any random unitary dynamical map for qubits can be defined either by the time-dependent probability distribution $\{p_k(t)\}_k$ that appears in Eq. (\ref{RU}) or by the time-dependent set of rates $\{\gamma_k(t)\}_k$ that defines the generator $\mathcal{L}_t$ and the dynamical map $\Lambda_t$ itslef (see Eqs. (\ref{randomunit}) and (\ref{genoru})). The relations that link these two sets are \cite{darekk} 
$$
p_0(t) = \frac{1}{4} \left(1 + A_{xy}(t) + A_{xz}(t) + A_{yz} (t) \right) \, ,
$$
$$
p_x(t) = \frac{1}{4} \left(1 - A_{xy}(t) - A_{xz}(t) + A_{yz} (t) \right) \, ,
$$
$$
p_y(t) = \frac{1}{4} \left(1 - A_{xy}(t) + A_{xz}(t) - A_{yz} (t) \right) \, ,
$$
\begin{equation}\label{pkt}
p_z(t) = \frac{1}{4} \left(1 + A_{xy}(t) - A_{xz}(t)- A_{yz} (t) \right) \, .
\end{equation}

\subsubsection{Example: the quasi-eternal non-Markovian model}\label{quasieternal}

We present a class of non-Markovian random unitary evolutions inspired by the `eternal' non-Markovian model in Refs.\cite{eternal0,eternal}.
We analyze the dynamical maps, defined by Eqs. (\ref{randomunit}), with parametrized time-dependent rates
\begin{equation}\label{gammaet}
\left\{\gamma_x(t),\gamma_y(t),\gamma_z(t)\right\} = \frac{\alpha}{2}  \left\{ 1 , 1 , -\mbox{tanh}(t-t_0)\right\} \, ,
\end{equation}
where $\alpha>0$ and $t_0\geq 0$. 

First, we notice that: $\gamma_x(t)=\gamma_y(t)=\alpha/2>0$ for any $t$, $\gamma_z(t)\geq 0$ for $t\leq t_0$ but $\gamma_z(t)< 0$ for any $t>t_0$.
We point out that, depending on the value of $\alpha>0$ and $t_0\geq 0$, the corresponding dynamical map $\Lambda_t^{(t_0,\alpha)}$  generated by Eqs. (\ref{randomunit}) may be physical, i.e., CPTP, or not. While the non-positivity of the rates implies the non-CP-divisibility of the dynamics, the P-divisibility conditions given by Eqs. (\ref{rattes}) are satisfied for any $\alpha>0$, $t_0$ and $t$. Therefore, in case $\alpha$ and $t_0$ defines a CPTP dynamical map, it is non-Markovian and P-divisible: the intermediate maps $V_{t_1,t_2}^{(t_0,\alpha)}$ are P (but not CP) for any time interval $(t_1,t_2)$ such that $t_1 > t_0$.  

\begin{itemize}
\item
{ $t_0=0$ (eternal model): }The trace-preserving map  $\Lambda_t^{(t_0,\alpha)}$ generated by the rates in Eq. (\ref{gammaet}) is CP for $\alpha \geq 1$ and only P for $0<\alpha <1$ \cite{eternal}.

\item
{ $t_0>0$ (quasi-eternal model):} If $\alpha\geq 1$, the map $\Lambda_t^{(t_0,\alpha)}$ is CPTP for any $t\geq 0$ and $t_0\geq 0$. Instead, if $0<\alpha<1$, we need to wisely choose the value of $t_0 \geq 0$ in order to have a $\Lambda_t^{(t_0,\alpha)}$ that is CPTP for every $t\geq 0$. First of all, we calculate the quantities $A_{ij}^{(\alpha,t_0)}(t)$  that completely define the dynamical map $\Lambda_t^{(\alpha,t_0)}$ (see Eq. (\ref{lambdaAij}))
\begin{equation}\label{Axy}
A^{(\alpha,t_0)}_{xy}(t)=\expp{-2 \alpha t} \, ,
\end{equation}
\begin{equation}\label{Ayz}
  A^{(\alpha,t_0)}_{yz}(t)=A^{(\alpha,t_0)}_{zx}(t)=\left( \expp{-t} \frac{\cosh(t-t_0)}{\cosh(t_0)}\right)^{\alpha} \, .
\end{equation}
Now, we derive the ranges of physicality of the parameters $\alpha>0$ and $t_0\geq 0$. 
In order to satisfy the CPTP conditions given in Eq. (\ref{CPmap}), we notice that $B^{(\alpha,t_0)}_{yzx}(t)=B^{(\alpha,t_0)}_{zxy}(t)=1-A^{(\alpha,t_0)}_{xy}(t)=1-\expp{-2 \alpha t}\geq 0$ for any $\alpha>0$ and $t\geq 0$. It is straightforward to verify that the last condition $B^{(\alpha,t_0)}_{xyz}(t)=1+A^{(\alpha,t_0)}_{xy}(t)-2 A^{(\alpha,t_0)}_{yz}(t) \geq 0$ is satisfied for any $t \geq 0$ if and only if  $\lim_{t\rightarrow\infty} B^{(\alpha,t_0)}_{xyz}(t) \geq 0$. 
Therefore, the relation between $\alpha>0$ and $t_0\geq 0$ that implies that $\Lambda_t^{(t_0,\alpha)}$ is CPTP for any $t\geq 0$ is $1-2( {\expp{2 t_0}+ 1 } )^{-\alpha}\geq 0$ or, equivalently
\begin{equation}\label{alphat0}
 t_0 \geq  T^{(\alpha)} \equiv \frac{1}{2}\log\left( 2^{1/\alpha} - 1 \right) \, .
\end{equation} 
\end{itemize}

\section{Correlation measures as witnesses of non-Markovianity}\label{III'0}

In this section we introduce the general notion of correlation measures and see how they can be used for the detection of non-Markovian dynamics. We then show that for single-parameter evolutions, any non-Markovian effect leads to an increase in any continuously differentiable correlation measure that is non-zero during the preceding evolution. %\add{In Section \ref{III'C}, since we are interested to evaluate correlations between two systems, we define two possible bipartition settings we are interested to focus on. Finally, in Section \ref{III'C2}, we consider time derivatives of continuously differentiable correlation measures for ``single parameter'' evolutions.}

\subsection{Correlation measures}\label{III0corrme}

A correlation measure is a function that quantifies correlations between subsystems. Many different measures have been constructed that quantify correlations in different ways or capture qualitatively different types of correlations. 
For a function $M$ to be considered an operationally meaningful correlation measure, it has to satisfy the following property:
\begin{itemize}\label{kko}
\item $M$ is non-increasing on $S(\mathcal{H}_A\otimes \mathcal{H}_S)$ under local operations.\label{c3}
\end{itemize}
This condition encapsulates the natural requirement that correlations cannot be created by local operations.

Note that since any product state can be prepared by local operations, all these states should give the same value of $M$, which also corresponds to the minimum of $M$ over all quantum states. We can impose this value to be equal to zero without loss of generality, having:
\begin{itemize}\label{kko}
\item  $M(\rho)\geq 0$ for any state $\rho$.\label{c1}
\item $M(\rho)=0$ if $\rho$ is a product state.\label{c2}
\end{itemize}

Now, consider a generic correlation measure {$M(\rho)$} that is continuously differentiable {on $S(\mathcal{H}_{AS})$}. If the intermediate map $\mathcal{I}_A\otimes V_{s,t}$ is differentiable with respect to $t$ the time derivative of {$M(\rho(s))\equiv M(\mathcal{I}_A\otimes V_{s,t}(\rho))$} at time $t$ can be expressed as
{ 
\begin{equation}\label{formul}
 \sum_i\frac{\partial M (\rho(s))}{\partial a_i(s)}\frac{d a_i(s)}{ds}\Big|_{s=t}=\sum_{ij}\frac{\partial M (\rho)}{\partial a_i} a_{j} \frac{d V_{ij}(s,t)}{ds}\Big|_{s=t} ,
 \end{equation}}
where $a_i$ ($a_i(s)$) are the coordinates of $\rho$ ($\rho(s)$).
We see that {$\frac{d}{ds}M(\rho(s))|_{s=t}$} depends only on {$a_i$},  and the time derivatives of the components $V_{ij}(s,t)$ of the intermediate map.
As long as {$\frac{\partial M (\rho)}{\partial a_i}a_{j}\neq 0$} for some $i,j$ and {$\rho$} there exist some intermediate map that will induce either a decrease or increase of {$M(\rho)$}. However, it is not always the case that a given measure $M$ satisfies {$\frac{\partial M (\rho)}{\partial a_i}a_{j}\neq 0$} for a {$\rho$} in the image of the dynamical map $\Lambda_t$. In Section \ref{IV'0} we show that entanglement measures and entanglement breaking evolutions \cite{horo} can provide examples of this situation.

%The first condition is a defining property of any measure and the second condition is the requirement that the measure should give a nonzero value only to correlated states.
%The third condition captures the property that correlation cannot be created by local operations. 
%Note that this directly implies that the measure must be invariant under reversible (unitary) local operations. 

\subsection{Correlation backflows as a signature of non-Markovianity}\label{III'C}

The monotonicity of correlation measures under CP maps implies that they are all non-increasing under Markovian, that is, CP-divisible, evolutions. In fact, consider the scenario of Fig.\ref{fig} left, where a system $S$, correlated with an additional particle $A$, is subjected to an evolution. If there is an increase in the correlations between these two particles during the evolution between times $t$ and $s$, there cannot be a CP intermediate map $V_{s,t}$. This is why correlation measures have been proposed to witness and quantify non-Markovian effects \cite{RHP,Luo}. {From an open-system perspective a decrease in correlations between system and ancilla during the evolution may be due to an irrevocable loss of system-ancilla correlations but it may also be that these correlations have been transformed into potentially recoverable correlations of the environment-system-ancilla. An increase can therefore be seen as a flow of correlations lost during the previous evolution back to the system and ancilla.}

%\add{The main task of this work is to understand how to witness the non-Markovian nature of a generic dynamical map $\Lambda_t$ with correlation measures for bipartite systems. In order to do so, we consider two different bipartition settings. Hence, being $S$ the system evolved by the dynamical map $\Lambda_t$, we consider (see Fig. \ref{fig}):
%\begin{itemize}
%\item First setting: we consider the bipartition between $S$ and $A$, where $A$ is an ancillary system ,
In what follows we also consider a slightly more complex setting with two additional systems $A$ and $A'$, while the evolution is again applied on system $S$, see also Fig.\ref{fig} right. The reasons to consider this extension will become clearer below, but it is straightforward to see that any measure of the correlations along the bipartition $A$-$A'S$ cannot increase either under Markovian dynamics. Of course it is trivial to recover the previous setting with only one additional {system} by taking an initial state which is product along the bipartition $A'$-$AS$.

\begin{figure}
\includegraphics[width=0.45\textwidth]{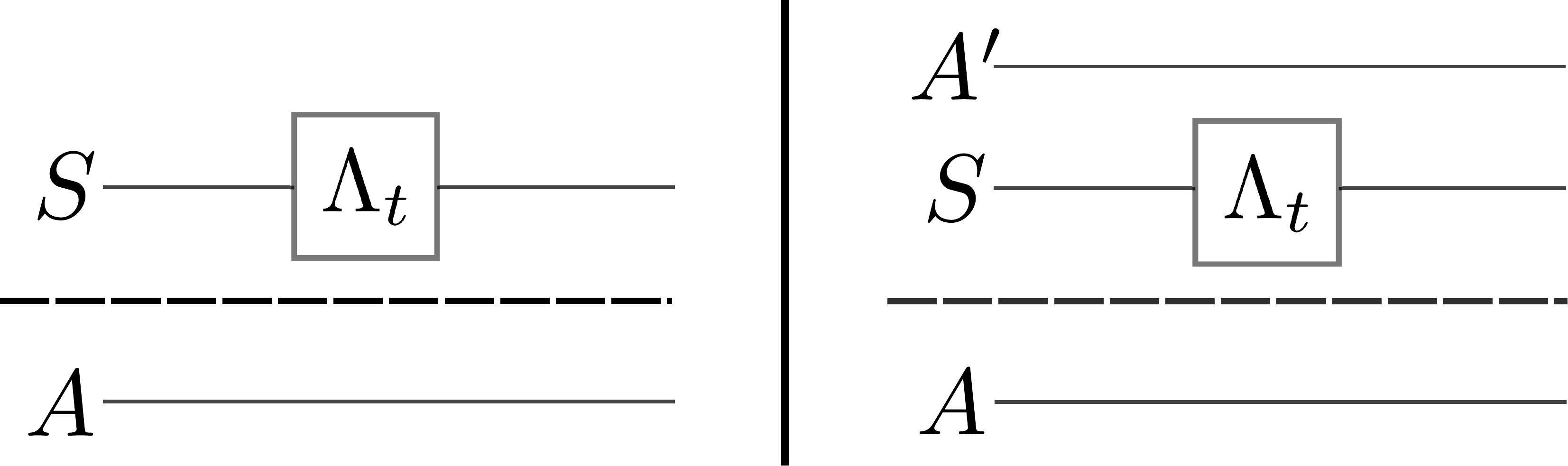}
\caption{Left: in the first setting, an initial state between system $S$ and ancilla $A$ is used. An increase of correlations between these two parts witnesses the presence of non-Markovian effects. Right: in our second extended setting, the whole setup consists of three parts, the systems $S$ and $A$ as before, plus an extra ancilla $A'$. An increase of the correlations over the bipartition $A$ versus $SA'$ is used to witness non-Markovian evolutions.}\label{fig}
\end{figure}

\subsection{Single parameter evolutions}\label{III'C2}

Our first result is for single-parameter evolutions and applies to any correlation measure that is continuously differentiable: for these evolutions, the sign of any non-zero derivative with time of a continuously differentiable correlation measure is determined by the sign of $\gamma(t)$. Therefore any non-Markovian effect leads to a correlation backflow, no matter which measure {of this kind is} used for the quantification {as long as it is not time independent on the whole image of $\Lambda_t$}. 

\begin{proposition}\label{prop1}
Let $M$ be a continuously differentiable correlation measure and let $\{\Lambda_t\}_t$ be a single parameter dynamical map. Then, 
\begin{equation}
\sign\left[\frac{d}{dt}M(\rho(t))\right]=-\sign[\gamma(t)] \, ,
\end{equation} 
for all {$\rho\in S(\mathcal{H}_{AS})$} such that {$\frac{d}{dt}M(\rho(t))\neq 0$}, where {$\frac{d}{dt}M(\rho(t)) \equiv \frac{d}{ds}M(\rho(s))|_{s=t}$ and $\rho(s)=\mathcal{I}_A\otimes V_{s,t}(\rho)$}.
\end{proposition}
\begin{proof}
For a continuously differentiable correlation measure $M$ the time derivative {$\frac{d}{dt}M(\rho(t))$} under an evolution of this type can be expressed on the form {$F(\rho)\gamma(t)$}, where {$F(\rho)\equiv \sum_{ij}\frac{\partial M(\rho)}{\partial a_i}a_j g_{ij}$} is a time independent function. Here we used that {$\sum_{ij}\frac{\partial M(\rho)}{\partial a_i}a_j h_{ij}(t)=0$} since $M$ is invariant under unitary evolution. Thus, {$\frac{d}{dt}M(\rho(t))$} is proportional to $\gamma(t)$.

If  $\gamma(t)>0$, so that $V_{\tau,t}$ is a non-unitary CP map, the time derivative 
{$\frac{d}{dt}M(\rho(t))$ is non-positive for all $\rho$. This implies that $F(\rho)$ is non-positive for all $\rho$.
Assume that $F(\rho)\neq 0$ for some $\rho$. Then 
it follows that $\sign{[\frac{d}{dt}M(\rho(t))]}=-\sign[\gamma(t)]$.}
\end{proof}
It follows from Proposition \ref{prop1} that a continuously differentiable correlation measure $M$ shows an increase in the not P-divisible intervals as long as {$\frac{d}{dt}M(\rho(t))\neq 0$} somewhere in the image of $\Lambda_t$.

%Note that the above discussion and the proof of Proposition \ref{prop1} do not use the condition $M(\rho)=0$ if $\rho$ is a product state. Thus, the above properties for single-parameter evolutions holds for all continuously differentiable measures that are non-increasing under local operations, not only for correlation measures.

\subsubsection{Example: Dephasing evolution}\label{ex1}
Dephasing in a fixed basis is an example of a {random unitary and} single-parameter evolution that satisfies the conditions in Propositions \ref{prop2} and \ref{prop1}.
The pure qubit dephasing dynamics is described by the dynamical maps
\begin{eqnarray}\label{dephasingEx}
\Lambda_t(\sigma_x)&=&e^{-\int_{0}^{t}\gamma(\tau)d\tau}\sigma_x,\nonumber\\
\Lambda_t(\sigma_y)&=&e^{-\int_{0}^{t}\gamma(\tau)d\tau}\sigma_y,\nonumber\\
\Lambda_t(\sigma_z)&=&\sigma_z,\nonumber\\
\Lambda_t(\mathbbm{1})&=&\mathbbm{1}.
\end{eqnarray}
The dynamical maps are bijective for all times and the intermediate maps $V_{t,s}$ are therefore given by $V_{s,t}=\Lambda^{-1}_t\Lambda_s$. 
Explicitly $V_{s,t}$ is given by
\begin{eqnarray}\label{cccc}
V_{s,t}(\sigma_x)&=&e^{-\int_{t}^{s}\gamma(\tau)d\tau}\sigma_x,\nonumber\\
V_{s,t}(\sigma_y)&=&e^{-\int_{t}^{s}\gamma(\tau)d\tau}\sigma_y,\nonumber\\
V_{s,t}(\sigma_z)&=&\sigma_z,\nonumber\\
V_{s,t}(\mathbbm{1})&=&\mathbbm{1}.
\end{eqnarray}
The generator $\mathcal{L}_t$ is given by
\begin{eqnarray}
\mathcal{L}_t(\rho)=\gamma(t)(\sigma_z\rho\sigma_z-\rho),
\end{eqnarray}
where $\gamma(t)$ is the time dependent dephasing rate. The dynamics is CP-divisible if and only if the dephasing rate $\gamma(t)\geq 0$. Furthermore, the dephasing dynamics is not P-divisible when $\gamma(t)< 0$. This can be seen directly from the two eigenvalues $1-e^{-\int_{t}^{s}\gamma(\tau)d\tau}$ and $1+e^{-\int_{t}^{s}\gamma(\tau)d\tau}$ of $V_{s,t}(\mathbbm{1}+\sigma_x)$.
%$\frac{1}{2}(1-e^{-\int_{t}^{s}\gamma(\tau)d\tau}),\frac{1}{2}(1-e^{-\int_{t}^{s}\gamma(\tau)d\tau}),0,0$ of the Choi matrix $C_{V_{s,t}}$,  and the eigenvalues 

%The corresponding dynamical map can be described by its action on the Pauli matrices $\sigma_i$ and the identity $I$.
%\begin{eqnarray}
%\Lambda_t(\sigma_x)=e^{-\int_{0}^{t}\gamma(\tau)d\tau}\sigma_x\nonumber\\
%\Lambda_t(\sigma_y)=e^{-\int_{0}^{t}\gamma(\tau)d\tau}\sigma_y\nonumber\\
%\Lambda_t(\sigma_z)=\sigma_z\nonumber\\
%\Lambda_t(I)=I
%\end{eqnarray}

We can see from Eq. (\ref{cccc}) that for all $i,j$ such that $V_{ij}(s,t)$ has a non-zero time derivatives it holds that $\frac{d V_{ij}(s,t)}{ds}\big{|}_{s=t}=-\gamma(t)$. Thus for any continuously differentiable correlation measure {$M(\rho(t))$ the dephasing rate $\gamma(t)$ determines the sign of $\frac{d}{dt}M(\rho(t))$  and $\frac{d}{dt}M(\rho(t))\geq 0$ for $\gamma(t)\leq 0$}.

%Question: Was it known that one always can find backflow wherever "outflow" occurs when the intermediate map is CP for dephasing?

\subsubsection{Example: Generalized amplitude damping evolution}\label{ex2}
Generalized amplitude damping evolution in a fixed basis is a second example that is single-parameter and satisfies the conditions in Propositions \ref{prop2} and \ref{prop1} except where it is non-differentiable.
The dynamics of generalized amplitude damping on a qubit is described by the dynamical maps
\begin{eqnarray}
\Lambda_t(\sigma_x)&=&G(t)\sigma_x,\nonumber\\
\Lambda_t(\sigma_y)&=&G(t)\sigma_y,\nonumber\\
\Lambda_t(\sigma_z)&=&G^2(t)\sigma_z,\nonumber\\
\Lambda_t(\mathbbm{1})&=&\mathbbm{1}+(2p-1)(1-G^2(t))\sigma_z,
\end{eqnarray}
where $0\leq G(t)\leq 1$ and $0\leq p\leq 1$. For $G(t)>0$ the dynamical maps are bijective and the
intermediate maps are given by
\begin{eqnarray}
 V_{s,t}(\sigma_x)&=&\frac{G(s)}{G(t)}\sigma_x,\nonumber\\
 V_{s,t}(\sigma_y)&=&\frac{G(s)}{G(t)}\sigma_y,\nonumber\\
 V_{s,t}(\sigma_z)&=&\left(\frac{G(s)}{G(t)}\right)^2\sigma_z,\nonumber\\
 V_{s,t}(\mathbbm{1})&=&\mathbbm{1}+(2p-1)\left(1-\left(\frac{G(s)}{G(t)}\right)^2\right)\sigma_z.
\end{eqnarray}
For $s$ and $t$ such that $G(t)=0$ the intermediate map only exists if $G(s)=0$ and can be defined as the identity map. If $G(t)=0$ and $G(s)\neq 0$ the intermediate map does not exist since the evolution is many-to-one.
For $t$ where $\mathcal{L}_t$ is well defined, it is given by
\begin{eqnarray}
\mathcal{L}_t(\rho)&=&p\gamma(t)(\sigma_-\rho\sigma_+-1/2\{\sigma_+\sigma_-,\rho\})\nonumber\\
&+&(1-p)\gamma(t)(\sigma_+\rho\sigma_--1/2\{\sigma_-\sigma_+,\rho\}),
\end{eqnarray}
where $\sigma_\pm=1/2(\sigma_x\pm i\sigma_y)$ and $\gamma(t)$ is given by
\begin{eqnarray}
\gamma(t)=-2\frac{d}{ds}\frac{G(s)}{G(t)}\bigg|_{s=t}=-\frac{2}{G(t)}\frac{d}{dt}G(t),
\end{eqnarray}
whenever $G(t)>0$ and differentiable. The dynamics is CP-divisible in a generic time interval $[t_1,t_2]$ when $\gamma(t)\geq 0$ for any $t\in [t_1,t_2]$. Furthermore, the amplitude damping dynamics is not P-divisible in  $[t_1,t_2]$  when $\gamma(t)< 0$ for any $t\in [t_1,t_2]$. 

We can easily see that when $G(t)>0$ and differentiable all non-zero time derivatives of the components of the intermediate map are proportional to $\gamma(t)$. Thus $\gamma(t)$ determines the sign of $\frac{d}{dt}M(\bar{a},t)$ for any continuously differentiable correlation measure $M(\bar{a},t)$.

\section{Entanglement measures}\label{IV'0}

After presenting the general definition of correlation measures, in the next sections we focus on different examples of these measures and study their behavior under non-Markovian effects. We start with entanglement measures, denoted by $M_E$, which are originally constructed to capture only non-classical correlations. This is why they do not only {not} increase under local operations, but even when these operations are assisted by classical communication (LOCC). This implies that  $M_E(\rho)=0$ if $\rho$ is a separable state. Ref. \cite{RHP} introduced the idea of using entanglement measures to witness non-Markovianity.

\begin{figure}
\includegraphics[width=0.31\textwidth]{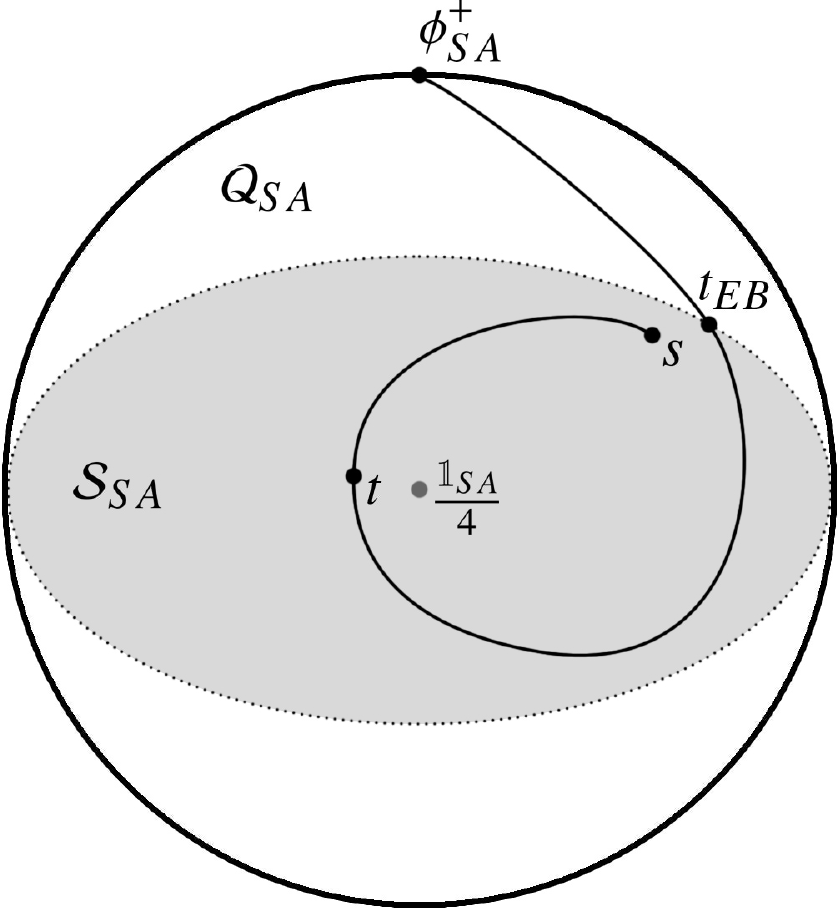}
\caption{Depiction of the trajectory of the evolution of a maximally entangled state $\phi^+_{SA}$ , where the system $S$ evolves under an entanglement breaking $\Lambda_t^{EB}$. Therefore, if $t>t_{EB}$, any initial state $\rho_{SA}(0)\in \mathcal{Q}_{SA}=S(\mathcal{H}_{SA})$ is evolved into a separable state $\rho_{SA}(t)\in \mathcal{S}_{SA}$. Suppose that $\Lambda_t^{EB}$ is non-Markovian but CP-divisible in $[0,t_{EB}]$ and with a non-CP intermediate map $V_{s,t}^{EB}$ for  $s>t>t_{EB}$. In this case, it is not possible to witness backflows of any entanglement measure. Indeed, entanglement is zero in the set of separable states $\mathcal{S}_{SA}$. }\label{EB}
\end{figure}

It is relatively straightforward to see that entanglement measures cannot detect all possible examples of non-Markovian dynamics, see also~\cite{previous}. In fact, consider the situation in which a dynamics $\Lambda_t^{EB}$ becomes and remains entanglement breaking (EB) after a given time $t_{EB}$, see Fig.~\ref{EB}. Any entanglement measure remains equal to zero for $t>t_{EB}$ and will therefore be unable to detect any non-Markovian effect taking place for $t>t_{EB}$. This is for instance the case for dynamics that are $P$-divisible. At $t=t_{EB}$ all the states in the image are separable, and remain separable after it {since} local positive maps do not create any entanglement when acting on separable states.

%An example (see Fig. \ref{EB}) of this is entanglement measures and entanglement breaking $\Lambda_t^{EB}$ \cite{horo}. A dynamical map $\Lambda_t^{EB}$ is entanglement breaking if it maps all input states to separable states.
%Therefore, let $t_{EB}(\Lambda_t^{EB})$ be the time such that, if $t\geq t_{EB}(\Lambda_t^{EB})$, the map $\mathcal{I}_A\otimes\Lambda_t^{EB}$ transforms any bipartite state $\rho_{SA}(0)$ into a separable state. More in specific, for any $\rho_{SA}(0) \in S(\mathcal{H}_S\otimes \mathcal{H}_A)$ and $t\geq t_{EB}$ exist a probability distribution $\{p_i\}$ and the two set of states $\{\rho_{A,i}\}_i$ and $\{\rho_{S,i}\}_i$, respectively of $S(\mathcal{H}_A)$ and $S(\mathcal{H}_S)$, such that $\rho_{SA}(t)=\mathcal{I}_A\otimes\Lambda_t^{EB}(\rho_{SA}(0))=\sum_i p_i \rho_{A,i} \otimes \rho_{S,i}$. 

%Since entanglement measures $M_E$ satisfy the condition that $M_E(\rho)=0$ if $\rho$ is a separable state, any entanglement measure is zero everywhere on the image of an entanglement breaking $\Lambda_t$.
%If the dynamics following an entanglement breaking $\Lambda_t$ is P-divisible no state will be mapped to an entangled state since a P intermediate map $V_{s,t}$ maps separable states to separable states.
%Therefore, there exist P-divisible but not CP-divisible evolutions where no entanglement measure can increase. 

For the sake of clarity, in what follows we provide an example of a non-Markovian qubit dynamics $\Lambda_t^{EB}$ which is CP-divisible in the time interval $[0,t_0]$, while it is only P-divisible for $t >t_0$. Showing that $t_{EB}(\Lambda_t^{EB}) < t_0$ we prove that this non-Markovian dynamics does not display any entanglement backflow.

\subsubsection{Example: the quasi-eternal non-Markovian model}

We consider a bipartite system, where $A$ and $S$ are qubits. A necessary and sufficient condition for a bipartite qubit state $\rho_{SA}\in S(\mathcal{H}_S\otimes \mathcal{H}_A)$ to be entangled is that its negativity, i.e., the entanglement measure
\begin{equation}\label{negativity}
N(\rho_{SA})\equiv \frac{||\rho_{SA}^{\Gamma_A}||_1-1}{2} \, ,
\end{equation}
is positive, where $\Gamma_A$ represents the operation of partial transposition on the system $A$.

 We consider the dynamical map $\Lambda_t^{(2/5,2)}$ that belongs to the class of quasi-eternal non-Markovian (P-divisible) evolutions introduced in Section \ref{quasieternal}, with $\alpha=2/5$, $t_0=2$ and the condition of physicality given by Eq. (\ref{alphat0}) is satisfied. Since $\gamma_x(t)=\gamma_y(t)=1/5$ and $\gamma_z(t)<0$ if and only if $t>t_0$, the corresponding intermediate map $V_{s,t}^{(2/5,2)}$ of this evolution is P- but not CP-divisible for any $s>t>t_0$, while it is CP when $t_0\geq s \geq t$.

This temporal evolution becomes entanglement breaking. Indeed, let $\phi^+_{SA}(t)= \mathcal{I}_A\otimes \Lambda_t^{(2/5,2)}(\phi^+_{SA})$ be the temporal evolution of the maximally entangled state $\phi^+_{SA}=\ketbra{\phi^+}{\phi^+}_{SA}$, where $\ket{\phi^+}_{SA}=(\ket{00}_{SA}+\ket{11}_{SA})/\sqrt{2}$. We obtain the separability of $\phi^+_{SA}(t)$, namely $N(\phi^+_{SA}(t))=0$, for any $t \geq t_{EB} (\Lambda_t^{(2/5,2)})\simeq 1.47$. Therefore, since $t_{EB}(\Lambda_t^{(2/5,2)})< t_0$, we conclude that it is not possible to observe a non-monotonic evolution of $N(\rho_{SA}(t))$, and more generally of any entanglement measure, for any initial state $\rho_{SA}(0)$. Indeed, the intermediate maps of this evolution are either P or CP for $ t_{EB}(\Lambda_t^{(2/5,2)})\leq t < s$ and CP for $t<s\leq t_{EB}(\Lambda_t^{(2/5,2)})$.

\section{Mutual Information}\label{V'0}

A commonly used correlation measure {that satisfy the conditions of Section~\ref{III0corrme}} is the %entropic correlation measures based on the von-Neumann entropy. 
the (quantum) mutual information $I(\rho_{AB})$~\cite{straton} %is such an entropic measure of bi-partite correlation of states $\rho_{AB}$ of a composite system with parts $A$ and $B$ 
defined as 
\begin{eqnarray}
I(\rho_{AB})\equiv S(\rho_A)+S(\rho_B)-S(\rho_{AB}),
\end{eqnarray}
where $S$ is the von Neumann entropy, $\rho_{AB}$ is the state of the system, and $\rho_A$ ($\rho_B$) is the reduced states of $A$ ($B$). The mutual information is a continuous function on the set of states and is analytic on the interior of the set of states, i.e., it is infinitely differentiable and equals its Taylor series in a neighbourhood of any point. {Differently from entanglement measures (see Section \ref{IV'0}), $I$ considers both classical and quantum correlations of bipartite systems.} The use of the mutual information to detect non-Markovian effects was proposed in Ref.~\cite{Luo}.

In the following, we present many different scenarios where the mutual information withnesses non-Markovian effects, but conclude by identifying a situation where it fails in this task.
In Section \ref{IVBunit} we show that whenever a qubit unital non-Markovian dynamics is not even P-divisible, an increase in the mutual information appears. In Section \ref{IVEnne} we consider non-Markovian  random unitary dynamics that are P-divisible to prove that maximally entangled states are not optimal to detect correlation backflow in terms of the mutual information. In Section \ref{IVAlabel} we move to a non-unital dynamics, in the form of a generalized amplitude damping channel, and provide a class of initial states that efficiently witness the non-Markovian nature of this evolution. Finally, in Sections \ref{IVCmeth} and  \ref{IVDrand}, we introduce the tools needed to show the existence of non-Markovian evolutions for which the mutual information does not provide backflows for any initial state.

\subsection{{Non-Markovian non-P-divisible} unital qubit dynamics}\label{IVBunit}

A number of commonly studied dynamics including dephasing and random unitary dynamics are unital.
For bijective unital dynamical maps acting on a qubit an increase of the mutual information can be observed  for any non-P intermediate map $V_{s,t}$.

\begin{theorem}
Let $\Lambda_{t}$ be a unital bijective qubit evolution. Furthermore assume that the intermediate map $V_{s,t}$ is analytic and non-$P$ at a given time $t$. Then there exist states $\rho_{AS}\in{B(\mathcal{H}_A\otimes \mathcal{H}_S)}$ in the image of $\Lambda_{t}$ for which $I(\mathcal{I}_A\otimes V_{s,t}(\rho_{AS}))>I(\rho_{AS})$.
\end{theorem}

\begin{proof}
Assume that $V_{s,t}$ is not $P$. Then there exists a pure state $\ket{\phi}_S$ on $\mathcal{H}_S$ such that $V_{s,t}(\ketbra{\phi}{\phi}_S)$ has a negative eigenvalue.
Let the eigenvalues of $V_{s,t}(\ketbra{\phi}{\phi}_S)$ be $1+\epsilon(s)$ and $-\epsilon(s)$, where $\epsilon(t)=0$ and $\epsilon(s)\geq{0}$.
Consider the state 
\begin{eqnarray}
\rho_{AS}&\equiv& \frac{1}{2}|0\rangle\langle 0|_A \otimes \left( p|\phi\rangle\langle\phi|_S+(1-p)\frac{\mathbbm{1}_S}{2}\right) \nonumber\\
&+&\frac{1}{2}|1\rangle\langle 1|_A \otimes \left( p|\phi^\bot\rangle\langle\phi^\bot|_S+(1-p)\frac{\mathbbm{1}_S}{2} \right) \, ,
\end{eqnarray}
where $\ket{\phi^\bot}$ is the pure state orthogonal to $\ket{\phi}$, and $|0\rangle\langle 0|_A$ and $|1\rangle\langle 1|_A$ are orthogonal states in $\mathcal{H}_A$.
Note that since the evolution is unital, the eigenvalues of $V_{s,t}(|\phi^\bot\rangle\langle\phi^\bot|_S)$ are $1+\epsilon(s)$ and $-\epsilon(s)$.
Note also that the reduced density matrices of both the system and the ancilla are maximally mixed.
Therefore, the reduced states are unchanged by a unital map $V_{s,t}$.
%The difference in mutual information between time $s$ and time $t$ for a unital map is thus $-S[I\otimes V_{s,t}(\rho)]+S(\rho)$.

The mutual information as a function of time is
\begin{eqnarray}%S(\rho_{ab})=-\frac{(1+p)}{2}\ln(\frac{(1+p)}{2})-\frac{(1-p)}{2}\ln(\frac{(1-p)}{2})\nonumber\\
I(\mathcal{I}_A&\otimes& V_{s,t}(\rho_{AS}))=\left(\frac{1+p}{2}+p\epsilon(s)\right)\ln\left(\frac{1+p}{2}+p\epsilon(s)\right)\nonumber\\&&\!\!\!\!\!\!+\left(\frac{1-p}{2}-p\epsilon(s)\right)\ln\left(\frac{1-p}{2}-p\epsilon(s)\right)+\ln 2. %\nonumber\\
\end{eqnarray} 
Its time derivative now reads %$\frac{d}{ds}I(s)|_{s=t}=-\frac{d}{ds}S[\mathcal{I}_A \otimes V_{s,t}(\rho)]|_{s=t}$ is
\begin{equation}-\frac{d}{ds}S (\mathcal{I}_A\otimes V_{s,t}(\rho_{AS}))|_{s=t}=
\frac{d\epsilon (t)}{dt}p\left( \ln(1+p)-\ln(1-p)\right).
\end{equation}
Note that $p(\ln(1+p)-\ln(1-p))>0$ for $0<p<1$. Therefore, for $\frac{d\epsilon(t) }{dt}>0$ the time derivative of the mutual information is positive for $\rho_{AS}$. Moreover, since $\epsilon(t)$ is assumed to be analytic $\frac{d\epsilon(t)}{dt}>0$ implies that the map $V_{t+\delta,t}$ is non-P for a sufficiently small $\delta$.
Since $\Lambda_{t}$ is bijective and unital, there always exists a sufficiently small $p$ such that $\rho_{AS}$ is in the image of $\Lambda_{t}$.
\end{proof}

\subsection{Non-Markovian random unitary dynamics and maximally entangled states}\label{IVEnne}

We provide a condition for the mutual information not to show backflows when a maximally entangled state is evolved by a random unitary dynamical map. {Thereafter}, we formulate a version of this condition that applies to qubits, where the P-divisibility of the dynamical map is implied.

We consider the bipartite scenario where the systems $S$ and $A$ are qudits ($d_S=d_A=d$), $S$ is evolved by a random unitary dynamical map $\Lambda_t$ (see Eq. (\ref{RU})) and  the ancillary system $A$ is left untouched. We study the evolution of a maximally entangled state $\phi=\ketbra{\phi}{\phi}$, where
\begin{equation}\label{mesSA}
\ket{\phi} = \frac{1}{\sqrt{d}} \sum_{i=1}^d \ket{a_i}_{A} \otimes\ket{ s_i}_S \, ,
\end{equation}
and $\{\ket{s_i}_S\}_i$ ($\{\ket{a_i}_A\}_i$) is an orthonormal basis of $\mathcal{H_S}$ ($\mathcal{H}_A$). First, we show that the evolved state is diagonal in a Bell basis with eigenvalues given by the the same probability distribution $\{p_k(t)\}_k$ that defines $\Lambda_t$ (see Eq. (\ref{RU})). Indeed
\begin{align}
\phi (t)  &= \sum_{k=1}^{N} p_k(t) (\mathbbm{1}_A \otimes\sigma_k )\ketbra{\phi}{\phi}  (\mathbbm{1}_A\otimes\sigma_k  ) 
\nonumber\\
&\equiv\sum_{k=0}^{N} p_k(t) \ketbra{\phi_k}{\phi_k}   ,
\end{align}
where $N=d^2-1$. The set of states $\{\ket{\phi_k}\}_k\equiv \{(\mathbbm{1}_A\otimes \sigma_k)\ket{\phi}\}_k$ define a Bell basis and are orthonormal since: $\braket{\phi_i}{\phi_j} = \tr{\phi (\mathbbm{1}_A\otimes \sigma_i \sigma_j) }= \trS{\trA{\phi}\sigma_i \sigma_j} = \frac{1}{d } \trS{ \sigma_i \sigma_j}=\delta_{ij}$.
It follows that the von Neumann entropy of $\phi (t)$ is defined by the distribution $p_k(t)$
\begin{equation}\label{Spk}
S(\phi (t))=- \sum_{k=0}^{N} p_k(t) \ln {p_k(t)} \, .
\end{equation}
The reduced states of $\phi (t)$ to the subsystems $S$ and $A$ are maximally mixed: $\rho_S(t)= \trA{\phi(t)}=\mathbbm{1}_S/d$ and $\rho_A(t)= \trS{\phi(t)}=\mathbbm{1}_A/d$. Thus, the only evolving component of the mutual information of $\phi (t)$ is given by (\ref{Spk}): $I(\phi (t))= 2 \log_2 d - S(\phi (t))$. The time derivative of this quantity is
\begin{align}
\frac{d}{dt} I (\phi (t))&=\sum_{k=0}^N \frac{dp_k(t)}{dt}  (\ln p_k(t) + 1 ) \nonumber\\
&= \frac{d p_0(t)}{dt} (\ln  p_0(t) + 1)  + \sum_{k=1}^N \frac{d p_k(t)}{dt} (\ln p_k(t) + 1)\nonumber\\ 
&= \sum_{k=1}^N \frac{d p_k(t)}{dt} \ln \frac{p_k(t)}{p_0(t)} \, .
\end{align}
It follows that, if  $\frac{d}{dt} p_k(t) \geq 0$  and $p_0(t)\geq p_k(t)$ for $k = 1,2,\dots,N$, we cannot witness any backflow of mutual information with a maximally entangled state since $\frac{d}{dt} I(\phi (t)) \leq 0$. 

Finally, we consider the qubit case, namely when $d=2$. From Eqs. (\ref{pkt}) and (\ref{Aij}) it follows that the conditions $p_0(t) \geq p_k(t)$, for $k=x,y,z$ and $t\geq 0$, are always satisfied. We conclude that, when $S$ and $A$ are qubits, if $\frac{d}{dt} p_k(t)\geq 0$ for any $k=x,y,z$ and $t\geq 0$, we cannot obtain any backflow of mutual information if the initial state is maximally entangled. 

We notice that in a time interval where $\frac{d}{dt} p_k(t)\geq 0$ the dynamics is P-divisible, but not necessarily CP-divisible. Thus, there are some cases of non-Markovian P-divisible qubit dynamics that cannot be witnessed by the the mutual information of an evolved maximally entangled state. In order to prove this result, we write the time derivative of $p_x(t)$
\begin{eqnarray}
\frac{d}{dt} p_x(t)&=& \frac{1}{2}\left(  (\gamma_x(t) + \gamma_y(t) ) A_{xy}(t)+(\gamma_z(t) + \gamma_x(t) ) A_{zx}(t) \right. \nonumber \\ 
& -&\left.  (\gamma_y(t) + \gamma_z(t) ) A_{yz}(t) \right) \, .
\end{eqnarray}
Similarly, we can write $\frac{d}{dt} p_y(t)$ and $\frac{d}{dt} p_z(t)$. We notice that $\frac{d}{dt} p_x(t) + \frac{d}{dt} p_y(t)=(\gamma_x(t)+\gamma_y(t) ) A_{xy}(t)$. Therefore, given the positivity of $A_{xy}(t)$, $A_{yz}(t)$ and $A_{zx}(t)$ (see Eq. (\ref{Aij})), if $\frac{d}{dt} p_k(t)\geq 0$ for $k=x,y,z$, the dynamics is P-divisible (in this case the conditions given in Eqs. (\ref{rattes}) are automatically satisfied). However, in general the converse is not true. Indeed, in Sec. \ref{Examples} we study two similar P-divisible evolutions for qubits where in the first the conditions $\frac{d}{dt}p_k(t) \geq 0$ for $k=x,y,z$ are not satisfied, while in the second they are.

In order to obtain an intuitive meaning of the conditions presented in this section, we look at the definition given in  Eq. (\ref{RU}) for random unitary evolutions. We notice that $p_0(t)$ represents the fraction of $\Lambda_t$ that acts as the identity map on $\rho_S(0)$. Therefore, if the value of $p_0(t)$ is increasing for some $t$, i.e., $\frac{d}{dt}p_0(t)>0$, it is reasonable to expect that at time $t$ the system $S$ is getting closer to its initial configuration $\rho_S(0)$ and therefore evolving under a non-Markovian evolution that can be witnessed. Conversely, since $\sum_{k=0}^N p_k(t)=1$, if $\frac{d}{dt} p_k(t) \geq 0$ for any $k\neq 0$, it follows that $\frac{d}{dt} p_0(t) \leq 0$. We expect that in this situation, where the overlap of $\rho_S(t)$ with its initial configuration decreases, $S$ undergoes an evolution that cannot be distinguished from a Markovian one.

\subsubsection{Example: the quasi-eternal non-Markovian model}\label{Examples}
In the last section we gave a set of conditions for non-Markovian random unitary dynamics such that, if satisfied, cannot be witnessed evolving maximally entangled states. 
In the case of qubits, these conditions are given in terms of the time derivative of the probability distribution $\{p_k(t)\}_k$ that defines Eq. (\ref{RU}). In this section we consider two examples. First, we consider a model that does not satisfy the conditions given in Section \ref{IVEnne} and we check if, apart from the maximally entangled states, the evolution of random pure states can provide any backflow. Secondly, we consider a model that satisfies these conditions, i.e., cannot be witnessed by any initially maximally entangled state, and we notice that the evolution of random pure initial states do not provide any backflow of mutual information.

First of all, we consider the qubit random unitary dynamics given in Section \ref{quasieternal}, where $\alpha=2/5$ and $t_0=1$. This example satisfies the condition of physicality given by Eq. (\ref{alphat0}). Indeed, $\lim_{t\rightarrow\infty} B^{(2/5,1)}_{xyz}(t)  \simeq 0.146 >0$. Moreover, the physicality of this model is shown by the positivity of the time-dependent distribution $\{p_k(t)\}_k$, for $k=0,x,y,z$, that defines the corresponding random unitary evolution 
$$
p_x(t)=p_y(t)=\frac{1}{4}\left(1- \expp{-4t/5} \right) \, ,
$$
$$
p_z(t)= \frac{1}{4}\left(1+\expp{-4t/5} -2\expp{-2t/5} \left(\frac{\cosh(t-1)}{\cosh(1)}\right)^{2/5}\right) \, ,
$$
where $p_0(t)=1-p_x(t) - p_y(t)-p_z(t) \geq 0$ and $p_0(0)=1$.  

 We define $t_{NM}(\rho_{SA}(0))$ the time when the backflow of mutual information starts if the initial state considered is $\rho_{SA}(0)$. 
We evaluated $t_{NM}(\rho_{SA}(0))$ for $2 \cdot 10^4$ pure random states $\rho_{SA}(0)=\ketbra{\psi_{SA}}{\psi_{SA}}$ of the form
$$
\ket{\psi_{SA}}= a_1 \ket{00}_{SA} + a_2 \ket{01}_{SA} +a_3 \ket{10}_{SA} + a_4\ket{11}_{SA} \, ,
$$
where the parameters $a_i$ are normalized complex random numbers. The minimum value of $t_{NM} (\rho_{SA}(0))$ obtained has been $t_{NM}(\overline \rho_{SA}(0))\simeq 2.404$, where the values of the parameters that generates $\overline\rho_{SA}(0)$ are characterized by: $|\overline a_1|\simeq |\overline a_3| \simeq |\overline a_4| \simeq 0$ and $|\overline a_2|\simeq 1$ up to local unitary operations on $A$. Our numerical analysis does not give any insight about the possible existence of a class of initial states for which the corresponding $t_{NM}(\rho_{SA}(0))$ is arbitrarily close to $t_0=1$, i.e., the earliest time for which the intermediate map $V_{t_1,t_2}$ is P but not CP. {If there exist pure states with $t_{NM}$ closer to $t_0=1$, they must belong to a small subset that we did not sample.}
We point out that for this model, while $\frac{d}{dt} p_{x}(t)\geq 0$ and $\frac{d}{dt} p_{y}(t)\geq 0$ for any $t\geq 0$, $\frac{d}{dt} p_z(t)<0$ for $t > 1.3254$. {Indeed, the evolution of a maximally entangled state $\ket{\phi^+}_{SA}=(\ket{00}_{SA}+\ket{11}_{SA})/\sqrt{2}$ shows a backflow of mutual information with $t_{NM}(\ketbra{\phi^+}{\phi^+})\simeq 2.741$.}

Finally, we have studied the eternal non-Markovian model given by $\alpha = 1 $ and $t_0=0$. Interestingly, in this case $\frac{d}{dt} p_{x}(t)\geq 0$, $\frac{d}{dt} p_{y}(t)\geq 0$ and $\frac{d}{dt} p_{z}(t)\geq 0$ for any $t\geq 0$. Therefore, this model satisfies the conditions given in Section \ref{IVEnne}, i.e., maximally entangled states that are evolved by this dynamics do not show backflows of mutual information. In our numerical simulations no backflow of mutual information has been observed evolving $10^3$ random pure initial states.

\subsection{Non-maximally entangled states help to detect non-Markovian dynamics}\label{IVAlabel}

The purpose of this section is to examine, through a concrete example, how the use of non-maximally entangled states can be of benefit for the detection of non-Markovian effects. This example also serves as an illustration of the witnessing potential of the mutual information for non-Markovian dynamics that are non-unital and not P-divisible. 

%A bipartite scenario is considered, where the system $S$ that undergoes the non-Markovian evolution is a qubit and $A$ is an ancillary qubit. The possibility to witness a backflow of the mutual  information of this bipartite system strictly depends on the initial state considered and evaluating all possible initial states can be a hard computational task. In general, it is not easy to determine if, at given a time $t$ for which the intermediate map $V_{t+dt,t}$ is not CP, there exists at least one initial state $\rho_{SA}(0)$ such that the time derivative of the mutual information is positive at time $t$, i.e., $\frac{d}{dt} I (\rho_{SA}(t)) >0$. To our knowledge, it is not even known is one can restrict the study to initial pure states.

The model we consider is a generalized amplitude damping channel (GADC) with two time dependent parameters, defined by the following set of Kraus operators
\begin{align}
K_1(t) &= \sqrt{s(t)}
\left( \begin{array}{cc}
1 & 0 \\
0 & \sqrt{r(t)}
\end{array} \right),  \nonumber\\ 
K_2(t) &= \sqrt{s(t)}
\left( \begin{array}{cc}
0 & \sqrt{1-r(t)} \\
0 & 0
\end{array} \right),  \nonumber\\ 
K_3(t) &= \sqrt{1-s(t)}
\left( \begin{array}{cc}
\sqrt{r(t)} & 0 \\
0 & 1
\end{array} \right),  \nonumber\\
K_4(t) &= \sqrt{1-s(t)}
\left( \begin{array}{cc}
0 & 0 \\
\sqrt{1-r(t)}  & 0
\end{array} \right),
\end{align}
where $s(t)=\cos^2(5t)$ and $r(t)=\expp{-t}$. The evolution induced by these operators is equivalent to that described by a generator $\mathcal{L}_t$ of the form given in Eq.(\ref{geno}) with $G_-=\sigma_-$ and $G_+=\sigma_+$ and the respective time-dependent rates
\begin{equation}\label{gammam}
\gamma_-(t)=  \cos^2(5t) -5 (1 -\expp{-t}) \sin(10t)   \, ,
\end{equation}
\begin{equation}\label{gammap}
\gamma_+(t)=\sin^2(5t)+5 (1 -\expp{-t}) \sin(10t) \, ,
\end{equation}
for which the following equality holds
\begin{equation}\label{1}
\gamma_-(t) + \gamma_+(t) =1 \, .
\end{equation}
In order to understand the non-Markovian behaviour of this model, it is possible to calculate the function $g(t)$ introduced in \cite{RHP}
$$
g(t)= \lim_{dt\rightarrow 0} \frac{||( \mathcal{I}_A\otimes V_{t+dt,t} )(\phi_{AS}) ||_1 - 1}{dt} \, ,
$$
where $\phi_{AS}$ is the maximally entangled state (\ref{mesSA}) and $|| \cdot ||_1$ represents the trace norm.
Indeed it is positive if and only if $V_{t+dt,t}$ is not CP. In our case $g(t)>0$ if and only if either $\gamma_-(t)$ or $\gamma_+(t)$ is negative
$$
g(t)= \frac{1}{2} \sum_{i=\pm} |\gamma_i(t)| - \gamma_i(t) =
 \left\{ \begin{array}{cc} 
-\gamma_-(t) &  t\in T^- \\
-\gamma_+(t) & t\in T^+ \\ 
 0 & \mbox{otherwise}
\end{array} \right. \, ,
$$  
where $T^\pm \equiv \{ t : \gamma_{\pm}(t) <0 \}$ are two non-overlapping sets of time intervals. We can define $T^-$ as the union of the time intervals $T_{k}^{-}\equiv (t_{in,k}^- , t_{fin,k}^-) $ when the rate $\gamma^-(t)$ is negative. If we do the same with $T^+$, we can write
$$
T^{\pm} \equiv  \bigcup_{k=1}^\infty T_{k}^{\pm} \equiv \bigcup_{k=1}^\infty \left(t_{in,k}^\pm , t_{fin,k}^\pm\right) \, .
$$
 In Ref. \cite{fanchini} the authors compare the ability of mutual information and entanglement of formation to witness non-Markovianity when a maximally entangled state is shared between $S$ and $A$ for the considered dynamics. They note that the mutual information does not show any backflow during the first time interval where dynamics is not CP-divisible, i.e., for $ t\in T^-_1$, while the entanglement of formation shows a backflow in a time interval that is a proper subset of $ T^-_1$. 
However, in order to fairly compare the witnessing potential of two different correlation measure, we must consider any possible initial state.

\begin{figure}
\includegraphics[width=0.49\textwidth]{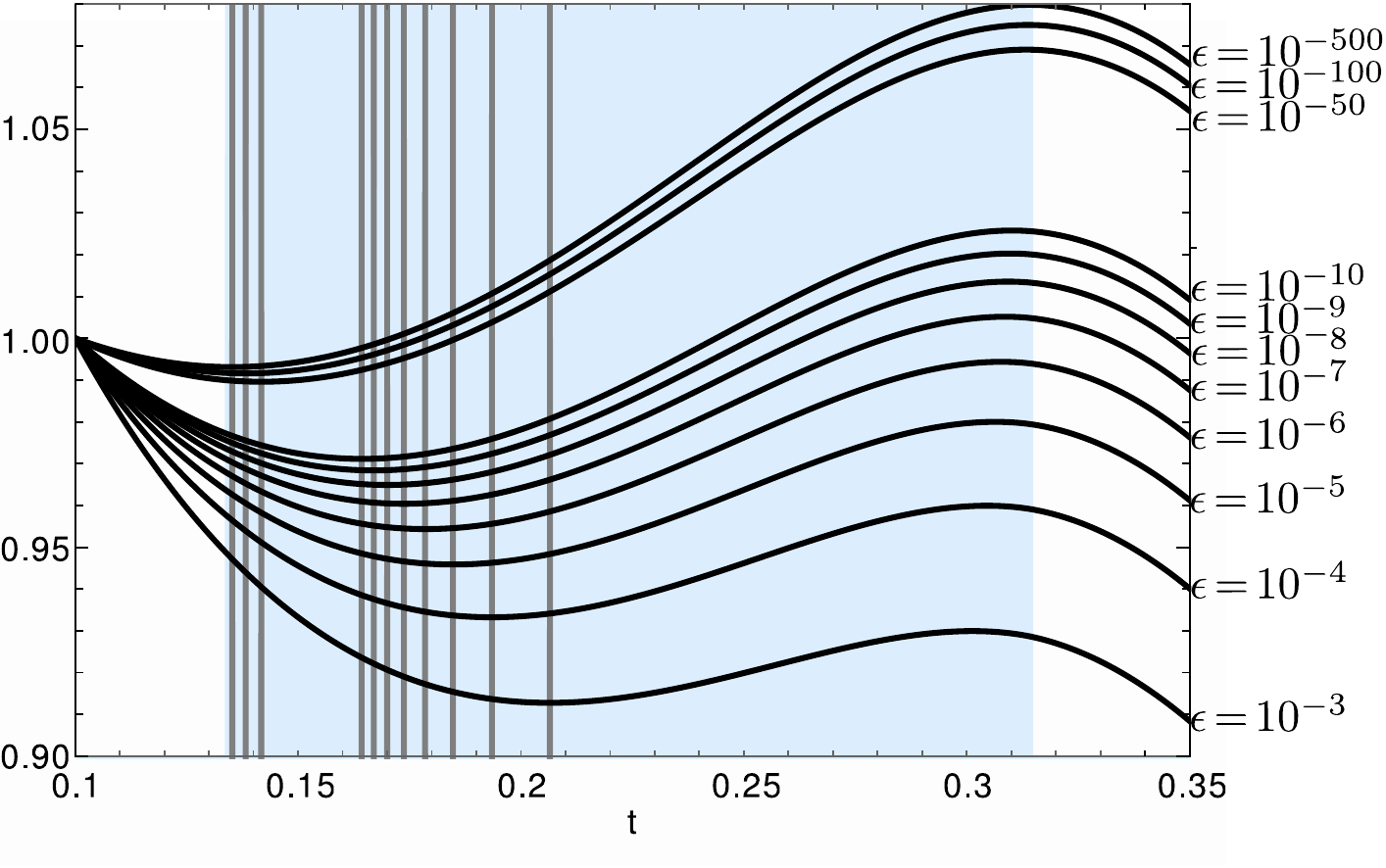}
\caption{The mutual information relative to its value at $t=0.1$, $I(\rho^-(t,\epsilon))/I(\rho^-(0.1,\epsilon))$ as a function of $t$ for values of $\epsilon$ between $10^{-3}$ and $10^{-500}$ (black curves). With successively smaller $\epsilon$ the mutual information increases in a larger part of the interval $0.13437\lesssim t\lesssim 0.31416$ where the dynamics is non CP-divisible (light blue area), and the $t$ where the mutual information begins to increase approaches the beginning of the interval (grey vertical lines). For $\epsilon=10^{-500}$ the increase in mutual information begins at $t\simeq 0.1352$.}\label{graphy}
\end{figure}

From now on we will focus on detecting backflows of mutual information for times $t\in T_1^- \simeq (0.13437, 0.31416)$: we give numerical results indicating that the mutual information can detect backflows for any $t\in T_1^-$. 
We consider as initial states the pure states 
\begin{equation}\label{psi-}
\ket{\psi^-(\epsilon)} \equiv \sqrt{1-\epsilon^2} \ket{00}_{SA} + \epsilon \ket{11}_{SA}  \, ,
\end{equation}
and provide strong evidence that these statesprovide backflows of mutual information for any time $t\in T_1^{-}$ for $\epsilon \rightarrow 0^+$. More precisely, we have observed backflows in the mutual information for $t\in \widetilde T_1^- \subset T_1^- $, where $\widetilde T_1^- \equiv ( t^-_{in,1}+\delta \tau, t^-_{fin,1}-\delta \tau )$ and $\delta\tau= 10^{-10}$.

We consider $\rho^{-}(0,\epsilon)=\ket{\psi^-(\epsilon)}\bra{\psi^-(\epsilon)}$ as the initial state of our complete system and we study its evolution $\rho^{-}(t,\epsilon)=\mathbbm{1}_A \otimes \Lambda_t(\rho^{-}(0,\epsilon))$, where $\Lambda_t$ represents the GADC described above. In Fig. \ref{graphy} we show the behaviour of $I(\rho^{-}(t,\epsilon))$ for several values of $\epsilon$. We notice that as $\epsilon$ approaches zero, the time interval where $I(\rho^- (t,\epsilon))$ is increasing widens and approaches $T^-_1$, while the amplitude of the mutual information decreases. The latter effect, and the increasingly small values of $\epsilon$, makes it difficult to numerically verify the possibility to witness a backflow of mutual information for  $t\in T^-_1$ arbitrarly close to $t_{in,1}^-$ and $t_{fin,1}^-$.

\begin{figure}
\includegraphics[width=0.45\textwidth]{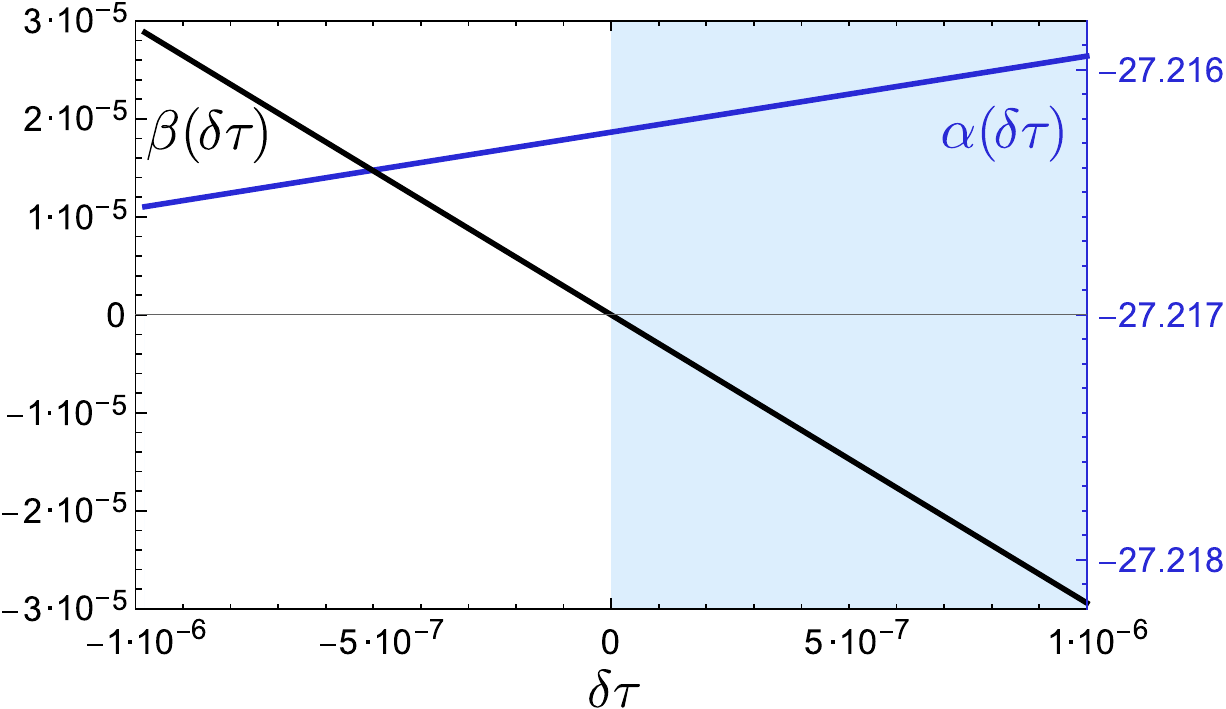}
\caption{The coefficients $\alpha(\delta \tau)$ (dark blue curve) and $\beta(\delta \tau)$ (black curve) of the leading order of the series expansion of $\frac{d}{dt}I(\rho^-(t,\epsilon))$ in $\epsilon$ as functions of $\delta \tau$ for $-10^{-6}\leq\delta\tau\leq 10^{-6}$. The dynamics is non-Markovian for $\delta \tau> 0$ (light blue area). For $\delta \tau\leq 0$ the coefficient $\beta(\delta \tau)$ is non-negative while $\alpha(\delta \tau)$ is negative and therefore the leading order term of the expansion is negative for any $\epsilon$.
For $\delta \tau >0$ both $\beta(\delta \tau)$ and $\alpha(\delta \tau)$ are negative and therefore, for sufficiently small $\epsilon$ the leading order term of the expansion is positive. For $\delta \tau=0$ the value of $\beta(\delta \tau)$ is zero to within numerical precision. 
}\label{agraph}
\end{figure}

To better understand the behaviour of $ I(\rho^-(t, \epsilon))$ when $t \simeq t_{in,1}^-$, we consider a series expansion  in $\epsilon$ of the time derivative of this quantity for times close to the beginning of $T_1^-$, e.g., for $|\delta\tau| \equiv| t-t_{in,1}^-| \leq 10^{-6}$. We find
$$
\frac{d}{dt} I(\rho^-(t_{in,1}^- + \delta\tau,\epsilon))=(\alpha(\delta\tau)+\beta(\delta\tau)\ln(\epsilon))\epsilon^2+\mathcal{O}(\epsilon^3) \, .
$$
This expansion (Fig. \ref{agraph}) is characterized by $\alpha (\delta\tau)<0$ and the relation sign$(\beta(\delta \tau))=-\,$sign$(\delta \tau)$, which has been verified up to $\delta\tau=\pm10^{-10}$.  An analogous result is obtained when $t \simeq t_{fin,1}^-$. For this case $\alpha (\delta\tau)$ is negative and sign$(\beta(\delta \tau))=\,$sign$(\delta \tau)$. In summary, the numerical analysis indicates that for any value of $t \in T_1^-$ there exists a positive number $\epsilon_t$ such that, if  $\rho_{SA}(0)=\rho^- (0,\epsilon)$, we have a backflow of mutual information at time $t$, i.e., $\frac{d}{dt} I(\rho^-(t,\epsilon)) >0$, for any $0<\epsilon < \epsilon_t$.

We focused just on the first time interval of non-Markovianity, i.e.,  $T_1^-$, because in this case mutual information does not show any backflow for the maximally entangled state. Indeed, given a value of $\epsilon>0$, the fraction of the time interval of $T_1^-$ for which $\frac{d}{dt} I (\rho^- (t,\epsilon))>0$ is smaller than the one that we have for $T_k^-$ when $k\geq 2$. Similar results are obtained for $T_1^+$ if we consider the initial state $\rho^+(0,\epsilon)\equiv \ketbra{\psi^{+} (\epsilon)}{\psi^{+} (\epsilon)}$ for $\epsilon>0$ that approaches to zero, where $\ket{\psi^+(\epsilon)} \equiv \epsilon \ket{00}_{SA} + \sqrt{1-\epsilon^2} \ket{11}_{SA}$. We underline that we cannot perform this numerical analysis for each $T_k^\pm$ since there is an infinite number of such intervals.

The results in this section demonstrate that it appears difficult to fully determine when a given non-Markovian dynamics experiences correlation backflow in terms of the mutual information, as one needs to consider all possible initial states. In fact, to our knowledge, it is not even clear if one can restrict the study to pure states. Despite all these difficulties, in the next sections we construct examples of non-Markovian dynamics where it can be proven that no backflow in the mutual information takes place.

\subsection{Non-Markovian dynamics with no backflows of mutual information}

We have analyzed several situations where correlations backflows as measured by the mutual information detect non-Markovianity, including explicit examples of non-P-divisible, P-divisible, unital and non-unital non-Markovian evolutions. In this section we construct, using the ideas introduced in~\cite{previous}, examples of non-Markovian evolutions for qubits that cannot be witnessed by backflows of mutual information. The key feature for the construction of these examples is that for large enough times all evolved states are contained in a neighborhood of the stationary state, that is, by choosing a large enough evolution time, the image of the dynamical map can be made as close as desired to its stationary state. Moreover, there are examples of evolutions in which it is possible to tune when non-Markovian effects take place. In particular, they can be made arbitrarily close to the stationary state. To understand the presence of correlation backflows for these maps it is enough to perform a perturbative study of the chosen correlation measure, the mutual information in this case, around the stationary state. This hugely simplifies the analysis and allows us to conclude the absence of correlation backflow in the mutual information for some non-Markovian dynamics. Examples of such maps are for instance shown in Sections \ref{IVCmeth} and \ref{IVDrand} for P-divisible evolutions of qubits.

\subsubsection{Example: the quasi-eternal non-Markovian model}\label{EE}

We start by presenting examples of channels in which it is possible to tune when non-Markovian effects take place.
In Section \ref{quasieternal} we introduced a class of non-Markovian random unitary evolutions \footnote{Since the rates of Eq. (\ref{gammaxyz}), for the values of $\alpha>0$ and $t_0\geq 0$ that satisfy the  condition of physicality (see Eq. (\ref{alphat0})), define a class of non-Markovian and unital dynamics for qubits which are P-divisible, the results given in Section \ref{IVBunit} do not apply. } where the rates that define the corresponding P-divisible dynamical maps, see Eqs. (\ref{randomunit}), are given by Eq. (\ref{gammaet}), i.e.,
\begin{equation}\label{gammaxyz}
\left\{ \gamma_x(t), \gamma_y(t), \gamma_z^{(t_0)}(t) \right\} = \frac{1}{5}\left\{ 1, 1, - \mbox{tanh}(t-t_0) \right\} \, .
\end{equation}
where we fixed $\alpha=2/5$.

We define $\Lambda_t^{(t_0)}$ to be the qubit dynamical map induced by the rates given in Eq. (\ref{gammaxyz}).
Moreover, let $V_{t_1,t_2}^{(t_0)}$ be the intermediate map of $\Lambda_t^{(t_0)}$ in the time interval $(t_1,t_2)$. 
 We define the image of the dynamical map $\Lambda_t^{(t_0)}$, i.e., the set of the accessible states of the evolution at time $t$, to be
\begin{equation}
\mbox{Im}(\Lambda_t^{(t_0)})\!=\! \{ \sigma\in S(\mathcal{H}_S) : \sigma= \Lambda_t^{(t_0)} ( \rho) \mbox{ for some } \rho\in S(\mathcal{H}_S) \}  .
\end{equation}
In Section \ref{Examples} we showed that $t_0=1$ and $\alpha=2/5$ define a physical evolution. Indeed, from Eq. (\ref{alphat0}), $\Lambda_t^{(t_0)}$ is a CPTP map for any $t\geq0$ and $t_0\geq T^{(2/5)}\simeq 0.769$.

The structure of this section is the following. First, we pick $t_0=t_0'$ to define $\Lambda_t^{(t_0')}$ and we consider an intermediate map $\overline V$ that occurs in the time interval $(t_1,t_2)$, where $t_1>t_0'$. Therefore, we consider a larger value of $t_0$, e.g., $t_0''>t_0'$, and we show that the same $\overline V$ occurs in a time interval which is delayed by the factor $\Delta {t_0}=t_0''-t_0'$, i.e., in $(t_1+\Delta {t_0},t_2+\Delta {t_0})$. Secondly, once we understand that varying the value of $t_0$ the same intermediate map $\overline V$ can be delayed as much as needed, we study the image of the dynamics for arbitrary large values of $t_0$. The final purpose is to show that, increasing enough the value of $t_0$, $\mbox{Im}(\Lambda_t^{(t_0)})$ is contained in an arbitrarily small neighbor of the stationary state.

\begin{figure}
\includegraphics[width=0.49\textwidth]{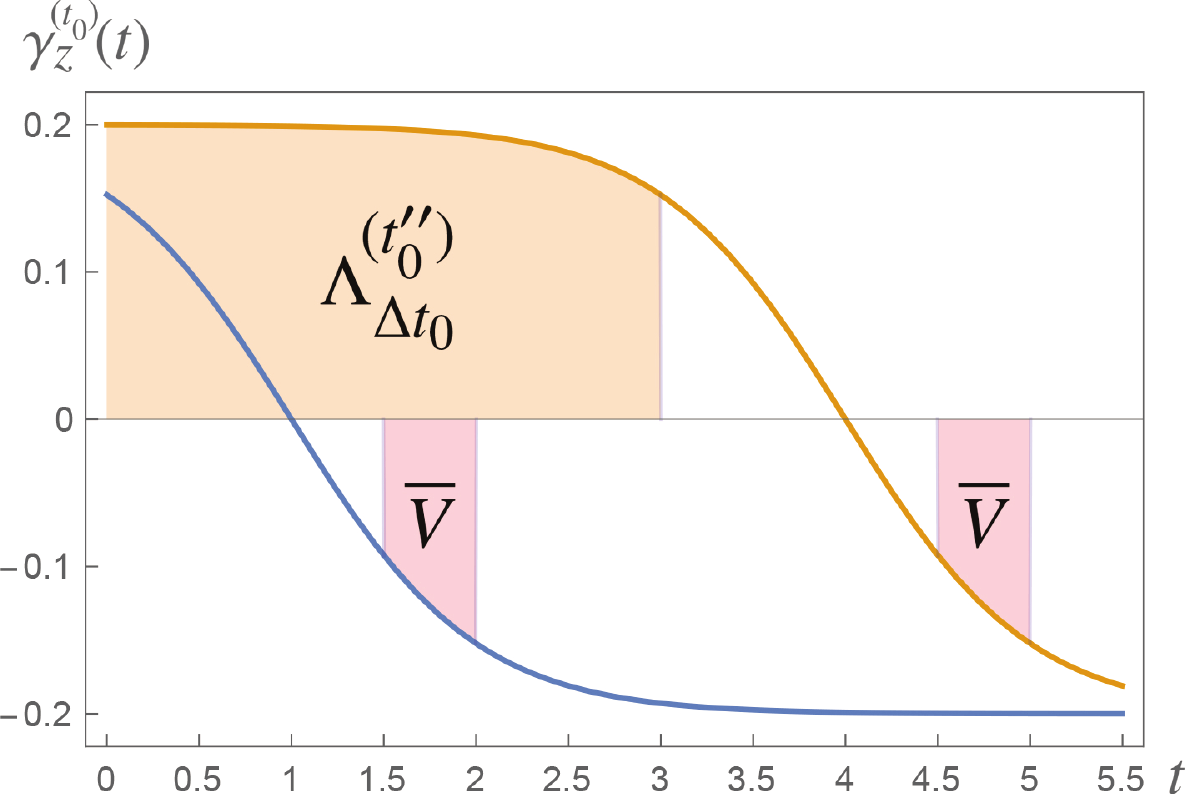}
\caption{The rate $\gamma_z^{(t_0)}(t)=-\tanh(t-t_0)/5$ of the quasi-eternal non-Markovian model for (blue) $t_0=t_0'=1$ and (orange) $t_0=t_0''=4$. The difference between $\Lambda_t^{(t_0')}$ and $\Lambda_t^{(t_0'')}$, since $\gamma_x(t) = \gamma_y(t)=1/5$ in both cases, is characterized by the different values of the integral of $\gamma_z^{(t_0)}(t)$ (see Eq. (\ref{randomunit})).
Let $\overline V$ be the P intermediate map of $\Lambda_t^{(t_0')}$ for the time interval $(t'_1,t'_2)=(1.5,2)$. Givent that $\Delta t_0=t_0''-t_0'=3$, $\overline V$ is also an intermediate map of $\Lambda_t^{(t_0'')}$, where now it occurs in a time interval shifted by the factor $\Delta t_0$, i.e., in $(t_1'',t_2'')=(4.5,5)$. The difference between the images of the two maps before the action of $\overline V$ is given by the contractive action of $\Lambda_{\Delta t_0}^{(t_0'')}$.}\label{Vbar}
\end{figure}

Following the steps we have described, we start considering the dynamical maps $\Lambda_t^{(t_0')}$ and $\Lambda_t^{(t_0'')}$, where $1 \leq t_0'<t_0''$ and $\Delta t_0\equiv t_0''-t_0'>0$. From Eqs. (\ref{randomunit}) and (\ref{ranu}), since the rates that defines the two evolutions differ by a simple time-shift (see Eq. (\ref{gammaxyz})), we can express the intermediate maps of $\Lambda_t^{(t_0'')}$ starting from time $\Delta t_0$ in terms of the dynamical map $\Lambda_t^{(t_0')}$. Indeed, for $\Delta t_0<t$
\begin{equation}\label{comp1}
\Lambda_{t-\Delta t_0}^{(t_0')}=V_{t,\Delta t_0}^{(t_0'')} \, ,
\end{equation}
As a consequence, if  $\Delta t_0<t$, the dynamical map $\Lambda_t^{(t_0'')}$ itself can be expressed as the composition of $\Lambda_{\Delta t_0}^{(t_0'')}$ and $\Lambda_{t-\Delta t_0}^{(t_0')}$
\begin{equation}\label{comp2}
\Lambda_t^{(t_0'')}=\Lambda_{t-\Delta t_0}^{(t_0')} \Lambda_{\Delta t_0}^{(t_0'')} \, .
\end{equation}
Indeed, composing Eqs. (\ref{comp1}) and (\ref{comp2}), we obtain the definition of the intermediate map $V_{t,\Delta t_0}^{(t_0'')}$, i.e., $\Lambda_t^{(t_0'')}=V_{t,\Delta t_0}^{(t_0'')} \Lambda_{\Delta t_0}^{(t_0'')}$. 

Now, we consider a time interval $(t'_1,t'_2)$ such that $t_0'<t'_1<t'_2$. In this time interval,  $\gamma_z^{(t_0')}(t)<0$  and the intermediate map $\overline{V}\equiv V_{t'_2,t'_1}^{(t_0')}$ is P but not CP  (see Section \ref{RUD}). From the results obtained above, it is clear that the action of $\overline V$ can also be obtained as the intermediate map of $\Lambda_t^{(t_0'')}$ that occurs in the shifted time interval $(t_1'',t_2'')\equiv (t'_1+\Delta t_0, t'_2 + \Delta t_0)$. Therefore, we have the identity
\begin{equation}\label{comp3}
\overline{V}= V_{t'_2,t'_1}^{(t_0')}= V_{t''_2,t''_1}^{(t_0'')} \, .
\end{equation}
Therefore, if we want to witness the non-Markovian effect of $\overline V$ while we consider $\Lambda_{t}^{(t_0')}$, its action occurs from time $t'_1$. Instead, considering $\Lambda_{t}^{(t_0'')}$,   $\overline V$ occurs from the delayed time $t_1''=t'_1+\Delta t_0$.

Having established that, for both $\Lambda_{t'_1}^{(t_0')}$ and $\Lambda_{t''_1}^{(t_0'')}$, the intermediate map of the evolution that follows for a time interval that lasts $t_2-t_1$ is $\overline V$, we proceed checking the images of the respective preceding evolutions, i.e., Im$(\Lambda_{t'_1}^{(t_0')})$ and Im$(\Lambda_{t''_1}^{(t_0'')})$.
%Therefore, since these two non-CP intermediate maps are equal, we can  check if the image of $\Lambda_{t_1+\Delta t_0}^{(t_0'')}$ is strictly smaller than the image of $\Lambda_{t_1}^{(t_0')}$.
The difference between the two images is given by the (contractive) CP map $\Lambda_{\Delta t_0}^{(t_0'')}$.
In order to understand the effect of this map, since $  \Delta t_0<t_0''$, the rates that define $\Lambda_{\Delta t_0}^{(t_0'')}$ through Eqs. (\ref{randomunit}) are positive in the time interval $(0,\Delta t_0)$. Therefore, the action of $\Lambda_{\Delta t_0}^{(t_0'')}$ is a "global" contraction, i.e., using Eq. (\ref{lambdaAij}) for $\Lambda_{ \Delta t_0}^{(t_0'')}$,
\begin{equation}\label{comp5}
\Lambda_{\Delta t_0}^{(t_0'')} (\sigma_i)= \lambda_i^{(t_0'')}(\Delta t_0) \sigma_i \, ,
\end{equation}
 where $\lambda_i^{(t_0'')}(\Delta t_0) < 1$ for  $i=x,y,z$ \footnote{In the example of Section \ref{ex1} we have studied a dynamics that, depending on the sign of $\gamma(t)$, can be Markovian or not.  Eqs. (\ref{dephasingEx}) shows that $\sigma_z$ is not contracted and therefore $\lambda_z=1$. If we consider the intermediate map of this evolution in a time interval where $\gamma(t)>0$, the action of this map is not a ``global'' contraction because $\lambda_z$ is never smaller than one.}. Moreover, since $\Lambda_{t'_1}^{(t_0')}$ is CP, we can write
\begin{equation}\label{comp6}
\Lambda_{t'_1}^{(t_0')} (\sigma_i)= \lambda_i^{(t_0')}(t'_1) \sigma_i \, ,
\end{equation}
where $\lambda_i^{(t_0')}(t_1') < 1$ for  $i=x,y,z$. Considering Eqs.  (\ref{comp5}) and (\ref{comp6}) in Eq. (\ref{comp2}), we obtain
\begin{equation}\label{comp6}
\Lambda_{t''_1}^{(t_0'')} (\sigma_i)=\lambda_i^{(t_0'')}(t''_1) \sigma_i=\lambda_i^{(t_0')}(t'_1)   \lambda_i^{(t_0'')}(\Delta t_0) \sigma_i \, ,
\end{equation}
where we remember that $t_1''=t'_1+\Delta t_0$.
Since $\lambda_i^{(t_0'')}(t_1'') < \min \{ \lambda_i^{(t_0')}(t'_1),\lambda_i^{(t_0'')}(\Delta t_0) \} $, we conclude that
\begin{equation}\label{comp4}
\I(\Lambda_{t''_1}^{(t_0'')})=\I (\Lambda_{t'_1}^{(t_0')} \Lambda_{\Delta t_0}^{(t_0'')} ) \subset \I (\Lambda_{t'_1}^{(t_0')} )  \, .
\end{equation}

The next logic step is to fix the value of $t_0'=1$ and increase $t_0''$ (and consequently $\Delta t_0=t_0''-t_0'$) in order to, at the same time, delay the occurrence of $\overline V$ for $\Lambda_t^{(t_0'')}$ and show that we can make the evolution that precedes its action more and more contractive (see Fig. \ref{Vbar}). 
Indeed, as showed above, if $t'_1$ is the time from which $\overline V$ occurs for $\Lambda_t^{(t_0')}$, the same intermediate map  (see Eq. (\ref{comp3})) occurs for $\Lambda_t^{(t_0'')}$ from time  $t_1''= t_1' + \Delta t_0$. Therefore, increasing $t_0''$, we delay more and more the time of occurrence $t_1''$ of the intermediate P map $\overline V= V_{t_2'',t_1''}^{(t_0'')}$.

Now, we prove that for any $\epsilon>0$ {there exists} a value of $t''_0$ such that the action of $\Lambda_{t_1''}^{(t''_0)}$ is characterized by  $\lambda_i^{(t''_0)}(t''_1)< \epsilon$ for each $i=x,y,z$. From Eqs. (\ref{randomunit}) and (\ref{gammaxyz}) it is easy to check that $\lambda_x^{(t''_0)}(t)=\lambda_y^{(t''_0)}(t)>\lambda_z^{(t''_0)}(t)$ for any $t>0$. Indeed,
\begin{equation}\label{lx}
\lambda_x^{(t''_0)}(t_1'')=\left( \expp{-{t_1''}} \frac{\cosh(t_1''-t_0'')}{\cosh(t_0)} \right)^{1/5} \!\!< \left( 2 \expp{-2t''_0}\right)^{1/5} \!\! \, ,
\end{equation}
\begin{equation}\label{lz}
 \lambda_z^{(t''_0)}(t_1'')=\left( 2 \expp{-2t''_0}\right)^{1/5} \left( \frac{\expp{-2(t'_1-1)}}{2} \right)^{1/5} \!\!\! \, .
\end{equation}
Therefore, for any $\epsilon>0$, if the following condition is satisfied
\begin{equation}\label{t0eps}
t_0''>\log \sqrt{2/\epsilon^5} \, ,
\end{equation} 
we have $\lambda_z^{(t_0'')}(t_1'')<\lambda_x^{(t_0'')}(t_1'')=\lambda_y^{(t_0'')}(t_1'')<\epsilon$. 

We want to understand the effects on the set of accessible states of $\Lambda_{t_1''}^{(t_0'')}$ that we obtain when $t_0''$ is increased over the bound given by Eq. (\ref{t0eps}). Therefore, we consider a generic initial state $\rho_S(0)=(\mathbbm{1}_S+\mathbf{v}(0)\cdot \bm{\sigma})/2$, represented by the Bloch vector $\mathbf{v}(0)=(v_x(0),v_y(0),v_z(0))$, where in the vector $\bm{\sigma}=(\sigma_x,\sigma_y,\sigma_z)$ we collect the Pauli matrices. We evolve this qubit state with the dynamical map $\Lambda_t^{(t_0'')}$, where the condition of Eq. (\ref{t0eps}) is satisfied for some $\epsilon>0$. At time $t=t_1''$ the Bloch vector is evolved to $\mathbf{v}(t_1'')=(v_x(t_1''),v_y(t_1''),v_z(t_1''))$. From Eq. (\ref{lx}) and (\ref{lz}), it is straightforward to show that $\max_i v_i(t_1'')<\epsilon \max v_i(0)$. In particular 
\begin{equation}
||\rho_S(t_1'') - {\mathbbm{1}_S}/{2}||_1= \frac{1}{2} \sqrt{ \sum_i v^2_i(t_1'') }<\frac{\epsilon }{2}  \, .
\end{equation}
In other words, Im$(\Lambda_{\tau(t_0)}^{(t_0)})$ is inside a neighbor of radius of order $\epsilon$ centered in the maximally mixed state $\mathbbm{1}_S/2$, namely the stationary state of the evolution. We can pick arbitrary small values of $\epsilon$ by considering arbitrary large values of $t_0''$ (and consequently $t_1''=t_1'+ t_0''-t_0'$), i.e., the time after which $\gamma_z^{(t_0'')}(t)<0$ and consequently $\Lambda_t^{(t_0'')}$ ceases to be CP-divisible.

Finally, we consider the mutual information $I$ of two qubits $A$ and $S$, where $A$ is an ancilla and $S$ is evolved by the dynamical map $\Lambda_t^{(t_0'')}$. In this scenario, we want to witness a backflow of $I$ in the time interval $(t_1'',t_2'')$, i.e., when the P intermediate map $\overline V$ evolves $S$.
Therefore, we increase the value of $t_0''$ until the image Im$(\mathcal{I}_A\otimes\Lambda_{t_1''}^{(t_0'')})$ 
of the evolution that precedes the action of $\mathcal{I}_A\otimes \overline V$ is in a neighbor of radius $\epsilon$ of the stationary states of the dynamics, namely $\rho_A\otimes \mathbbm{1}_S/2$, where $\rho_A$ is any state of $S(\mathcal{H}_A)$. At time $t_1''$ we have that
\begin{itemize}
\item The states in Im$(\mathcal{I}_A\otimes \Lambda_{t_1''}^{(t''_0)})$ are $\epsilon$-close from the stationary states of the evolution.
\item The evolution in the following time interval $(t_1'', t_2'')$ is described by the P, but not CP, intermediate map $\mathcal{I}_A\otimes \overline V$.
\item The previous point holds for arbitrary small time intervals $(t''_2-t''_1)$.
\end{itemize}
To sum up, there exist dynamics in which non-Markovian effects take place {only} arbitrarily close to the stationary state. For these cases, to understand the existence of correlation backflows in the mutual information it suffices to perform a perturbative analysis of this quantity around the stationary state. This is precisely what we do in what follows, Sections \ref{IVCmeth} and \ref{IVDrand}, for a random unitary dynamics for qubits. The calculations are lengthy and can be skipped by readers not interested in the technical details, as they are not needed for the rest of the article. However, this perturbative analysis allows us to conclude that the studied non-Markovian dynamics do not present any increase of the mutual information for any initial state in $S(\mathcal{H}_A\otimes \mathcal{H}_S)$.

\subsection{Taylor expansion of the time derivative of the mutual information}\label{IVCmeth}

{To study the time dependence of the mutual information perturbatively we here outline how the time derivative of the mutual information $\frac{d}{dt}I(\rho(t))\equiv\frac{d}{ds}I(\bar{a},V_{s,t})\big{|}_{s=t}$ in a neighbourhood of a given state $\rho$ can be described by a Taylor expansion in the coordinates $a_i$. In particular we consider Taylor expansions as a tool to investigate the neighborhoods of stationary states.}

{The mutual information as a functon on the set of states $S(\mathcal{H}_A\otimes \mathcal{H}_S)$ is analytic for all states $\rho$ of full rank, i.e., everywhere in the interior of the set of states, here denoted $\Int[S(\mathcal{H}_A\otimes \mathcal{H}_S)]$. Thus, in any open neighborhood $U\subset \Int[S(\mathcal{H}_A\otimes \mathcal{H}_S)]$ of a state $\rho$ the mutual information equals its Taylor series and we can use Taylor expansions to analyze it perturbatively around $\rho$. Moreover, if the dynamics is differentiable the time derivative of any analytic correlation measure is analytic as well.}

\begin{proposition}
Let $M$ be a correlation measure that is analytic at $\bar{a}$. If $\frac{d V_{s,t}}{ds}\big{|}_{s=t}$ is well defined it follows that $\frac{d}{dt}M(\bar{a},t)\equiv\frac{d}{ds}M[\bar{a},V_{s,t}]\big{|}_{s=t}$ is analytic at $\bar{a}$.
\end{proposition}

\begin{proof}
The time derivative $\frac{d}{dt}M(\bar{a},t)\equiv\frac{d}{ds}M[\bar{a},V_{ij}(s,t)]\big{|}_{s=t}$ can be expressed as $\frac{d}{dt}M(\bar{a},t)=\sum_{i,j}a_j\frac{d V_{ij}(s,t)}{ds}\big{|}_{s=t}\frac{\partial}{\partial a_i}M(\bar{a},t)$.
%depends only on $\bar{a}$ and the time derivatives $\frac{d V_{ji}(s,t)}{ds}\big{|}_{s=t}$ of the intermediate map $V_{s,t}$ as described in Eq. \ref{formul}.
Assume that $\frac{d V_{ij}(s,t)}{ds}\big{|}_{s=t}$ is well defined for each $i,j$. Then,
since products, linear combinations, and derivatives of analytic functions are analytic it follows that $\frac{d}{dt}M(\bar{a},t)$ is analytic as a function of $\bar{a}$. %Note that this argument applies to any analytic correlation measure.

\end{proof}

{Thus, in particular, if $V_{s,t}$ is differentiable $\frac{d}{dt}I(\rho(t))$ can be described perturbatively in any open neighborhood of a state in $\Int[S(\mathcal{H}_A\otimes \mathcal{H}_S)]$ by a Taylor expansion. On the other hand, for states of less than full rank, i.e., states on the boundary of $S(\mathcal{H}_A\otimes \mathcal{H}_S)$, the partial derivatives in the coordinates $a_i$ need not even be well defined to all orders.}

{Now consider a stationary state $\rho_0$ of a linear divisible dynamic described by $\Lambda_t$ which is of full rank, i.e. $\rho\in\Int[S(\mathcal{H}_A\otimes \mathcal{H}_S)]$.
Since $\frac{d}{dt}I(\rho_0(t))=0$ the sign of $\frac{d}{dt}I(\rho(t))$ in a neighborhood of $\rho_0$ is determined by the terms of order greater than zero in the Taylor expansion of $\frac{d}{dt}I(\rho(t))$.}

{In general $\frac{d}{dt}I(\rho(t))$ may take both positive and negative values in a neighbourhood of $\rho_0$. If a neighborhood of $\rho_0$ exists where $\frac{d}{dt}I(\rho(t))$ is everywhere non-negative, or alternatively everywhere non-positive, depends only on $\rho_0$ and on $\frac{d}{ds} V_{s,t}\big{|}_{s=t}$. 
In particular, the properties of the neighborhood is independent of the previous dynamic $\Lambda_t$ and the properties of the image of $\Lambda_t$ since we assumed linear divisibility of the dynamics.}

{This last observation allows us to make the following two statements about the change of the mutual information.
If there is some neighborhood of $\rho_0$ where $\frac{d}{dt}I(\rho(t))$ is somewhere positive, and this neighbourhood is contained in the image of $\Lambda_t$, we can observe an increase of the mutual information. 
If there is some neighborhood of $\rho_0$ where $\frac{d}{dt}I(\rho(t))$ is non-positive, and the image of $\Lambda_t$ is contained in this neighborhood, we can not observe any increase  of the mutual information.}

\subsubsection{Neighbourhoods of critical points}

{In the case that one or more first derivatives are zero one must consider higher order terms of the Taylor expansion to study how the sign of $\frac{d}{dt}I(\rho(t))$ behaves in a neighbourhood of a stationary state $\rho_0$. In particular this is true if all first derivatives with respect to the $a_i$ are zero, i.e., if $\rho_0$ is a critical point of $\frac{d}{dt}I(\rho(t))$.}
The relevance of considering critical points in relation to stationary states can be understood from the following two observations.
For any continuously differentiable dynamics a product state in the interior of the set of states is a critical point of $\frac{dM}{dt}$ if $M$ is analytic.

\begin{proposition}\label{prop3}
Let $M$ be a correlation measure that is analytic at a state $\rho$.
Assume that $V_{s,t}$ is continuously differentiable.
Then if $\rho\in\Int[ S(\mathcal{H}_A\otimes\mathcal{H}_S)]$ and is a product state it is a critical points of $\frac{dM}{dt}$.
\end{proposition}
\begin{proof}
See Appendix \ref{critic2}.
\end{proof}
{Thus, in particular, all product states in $\Int [S(\mathcal{H}_A\otimes \mathcal{H}_S)]$ are critical points of $\frac{d}{dt}I(\rho(t))$.}

For qubit evolutions a stationary state in the interior of the set of states is a critical point of $\frac{dM}{dt}$ if $M$ is analytic.
\begin{proposition}\label{prop4}
Let $M$ be a correlation measure that is analytic at $\rho$ and let $\{\Lambda_{t}\}_t$ be a continuously differentiable dynamical qubit evolution.
If $\rho$ is in $\Int [S(\mathcal{H}_A\otimes \mathcal{H}_S)]$ and is a stationary state of $\{\Lambda_{t}\}_t$ it is a critical point of $\frac{dM}{dt}$.
\end{proposition}
\begin{proof}
See Appendix \ref{critic3}.
\end{proof}
{Thus, for qubit evolutions all stationary states in $\Int [S(\mathcal{H}_A\otimes \mathcal{H}_S)]$ are critical points of $\frac{d}{dt}I(\rho(t))$.}

{The nature of a critical point $\rho_0$ can be analyzed by obtaining the eigenvalues of the Hessian matrix, i.e., the matrix of second derivatives $\textbf{H}_{i,j}=\frac{\partial^2}{\partial a_i \partial a_j}\frac{d}{dt}I(\rho(t))$. 
However, for any stationary state, the Hessian $\textbf{H}_{i,j}$ is of less than full rank since $\frac{d}{dt}I(\rho(t))=0$ on the set of stationary states of $V_{s,t}$, denoted $S_s$, and on the set of product states, denoted $S_p$.
From this follows that any eigenvector of the Hessian that is tangent to $S_s\cup S_p$ corresponds to a zero eigenvalue. 
Thus, the sign of $\frac{d}{dt}I(\rho(t))$ in the part of the neighborhood of $\rho_0$ that coincides with the zero-eigenspace $E_0$ of  $\textbf{H}_{i,j}$ cannot be determined from the Hessian matrix alone since it depends on higher order derivatives.}

{On the overlap of the neighborhood of $\rho_0$ with the complement of $E_0$, i.e., with $E_0^C\equiv B(\mathcal{H}_A\otimes \mathcal{H}_S)\backslash E_0$, the Hessian does describe the sign of $\frac{d}{dt}I(\rho(t))$ if the neighborhood is sufficiently small.
In particular, if all non-zero eigenvalues of the Hessian, which correspond to eigenvectors tangent to $E_0^C$, are negative there exists a neighborhood $U^-_{\rho_0}$ of $\rho_0$ where $\frac{d}{dt}I(\rho(t))$ is negative in $U^-_{\rho_0}\cap E_0^C$. If all the non-zero eigenvalues of the Hessian are positive there exists a neighborhood $U^+_{\rho_0}$ of $\rho_0$ where $\frac{d}{dt}I(\rho(t))$ is positive in $U^+_{\rho_0}\cap E_0^C$.}

\subsubsection{Calculating partial derivatives}\label{Taylor}

{Directly calculating the derivatives of $\frac{d}{dt}I(\rho(t))$ with respect to the coordinates $a_i$ can be demanding since the eigenvalues of $\rho$ are the roots of a polynomial with degree equal to $\dim (\mathcal{H}_A\otimes\mathcal{H}_S)$.
To circumvent this difficulty we calculate the derivatives and second derivatives at a state $\rho$ using a method  adapted from Ref. \cite{Tsing}. The method described in this reference is valid for real symmetric matrices but generalizing it to Hermitian complex matrices is very straightforward. 
We present the version of this method that works for Hermitian matrices in the following paragraphs.}

{Let $f$ be a spectral function defined on a set of $n\times n$ Hermitian matrices $A$ parametrized by real numbers $a_i$. By spectral function we refer to a function that only depends on the eigenvalues $\{\lambda_k\}_{k=1}^n$ of $A$ but not on the ordering of the eigenvalues. Furthermore, assume that $f$ is analytic in a given point $\bar{a}$ and let $\lambda_k(\bar{a})$ be the eigenvalue of $A(\bar{a})$ with corresponding normalized eigenvector $u_k(\bar{a})$. Then the first and second order partial derivatives of $f$ with respect to the parameters $a_i$ in the point $\bar{a}$ can be expressed as}
\begin{eqnarray}
\frac{\partial f(\bar{a})}{\partial a_i}=\sum_{k}\frac{\partial{f[\lambda(\bar{a})]}}{\partial \lambda_k}h_i^k(\bar{a}),
\end{eqnarray}
and
\begin{eqnarray}
\frac{\partial^2 f(\bar{a})}{\partial a_i \partial a_j}=&&\sum_{k,l}\frac{\partial^2{f[\lambda(\bar{a})]}}{\partial \lambda_k \partial \lambda_l}h_i^k(\bar{a})h_j^l(\bar{a})\nonumber\\&&+\sum_{k}\frac{\partial{f[\lambda(\bar{a})]}}{\partial \lambda_k}h_{ij}^k(\bar{a})+\eta_{ij}(\bar{a}),
\end{eqnarray}
respectively, where
\begin{eqnarray}
\label{inter}
h_i^k(\bar{a})=&&u_k^\dagger\frac{\partial A(\bar{a})}{\partial a_i}u_k,\nonumber\\
h_{ij}^k(\bar{a})=&&u_k^\dagger\frac{\partial^2 A(\bar{a})}{\partial a_i \partial a_j}u_k+\sum_{l|\lambda_k\neq \lambda_l}\frac{\alpha_{ij}^{kl}(\bar{a})}{\lambda_k(\bar{a})-\lambda_l(\bar{a})},\nonumber\\
\alpha_{ij}^{kl}(\bar{a})=&&\left(u_k^\dagger(\bar{a})\frac{\partial A(\bar{a})}{\partial a_i}u_l(\bar{a})\right)\left(u_l^\dagger(\bar{a})\frac{\partial A(\bar{a})}{\partial a_j}u_k(\bar{a})\right)\nonumber\\&&+\left(u_k^\dagger(\bar{a})\frac{\partial A(\bar{a})}{\partial a_j}u_l(\bar{a})\right)\left(u_l^\dagger(\bar{a})\frac{\partial A(\bar{a})}{\partial a_i}u_k(\bar{a})\right),\nonumber\\
\eta_{ij}(\bar{a})=&&\sum_{k,l|\lambda_k=\lambda_l,k<l}\alpha_{ij}^{kl}(\bar{a})\frac{\partial^2 f[\lambda(\bar{a})]}{\partial^2 \lambda_k}.
\end{eqnarray}
{When two or more eigenvalues coincide, the corresponding eigenvectors are not uniquely defined. Nevertheless, the method here can still be used since while the expressions given in Eq. \ref{inter}, e.g. $h_i^k$, may depend on the choice of eigenvectors, the partial derivatives themselves are independent and can be evaluated using any such choice.}

{If the diagonal form of $A$ in the point $\bar{a}$ and the corresponding eigenvectors $u_k(\bar{a})$ are known the method described here can greatly simplify the computation of the partial derivatives.}

\subsection{Random unitary dynamics that the mutual information cannot witness}\label{IVDrand}
{We now show that the mutual information is non-increasing for some cases of random unitary qubit dynamics that are P-divisible but not CP-divisible by studying a neighborhood of the stationary states using the methods described in Sect. \ref{IVCmeth}.}

We consider an ancilla that is also a qubit and explicitly introduce coordinates $a_i$ for $B(\mathcal{H}_A\otimes \mathcal{H}_S)$ with respect to an orthonormal basis $\{e_i\}_{i=0}^{15}$ defined by
\begin{align}
&e_0=\mathbbm{1}\otimes{\mathbbm{1}},& &e_{8}=\sigma_y\otimes{\mathbbm{1}},\nonumber\\
&e_{1}=\mathbbm{1}\otimes{\sigma_x},& &e_9=\sigma_y\otimes{\sigma_x},\nonumber\\
&e_{2}=\mathbbm{1}\otimes{\sigma_y},& &e_{10}=\sigma_y\otimes{\sigma_y},\nonumber\\
&e_{3}=\mathbbm{1}\otimes{\sigma_z},& &e_{11}=\sigma_y\otimes{\sigma_z},\nonumber\\
&e_4=\sigma_x\otimes{\mathbbm{1}},& &e_{12}=\sigma_z\otimes{\mathbbm{1}},\nonumber\\
&e_5=\sigma_x\otimes{\sigma_x},& &e_{13}=\sigma_z\otimes{\sigma_x},\nonumber\\
&e_6=\sigma_x\otimes{\sigma_y},& &e_{14}=\sigma_z\otimes{\sigma_y},\nonumber\\
&e_7=\sigma_x\otimes{\sigma_z},& &e_{15}=\sigma_z\otimes{\sigma_z},
\end{align}
where all operators are of the form $\chi_A\otimes{\chi_S}$ for $\chi_A\in B(\mathcal{H}_A)$  and $\chi_S \in B(\mathcal{H}_S)$. 
A state $\rho$ is represented as
\begin{equation}
\rho=\frac{1}{4} \mathbbm{1}\otimes{\mathbbm{1}}+\sum_{i=1}^{15} a_i e_i \, ,
\end{equation}
where $a_i=\frac{1}{4}\Tr(\rho e_i)$.

{The analysis of $\frac{d}{dt}I(\rho(t))$ in the neighborhood of the stationary states is done by first considering the states of full rank, i.e., the states in $\Int[S(\mathcal{H}_A\otimes\mathcal{H}_S)]$, where $\frac{d}{dt}I(\rho(t))$ is analytic.
There we calculate the second derivatives of $\frac{d}{dt}I(\rho(t))$ at the stationary states and find the eigenvalues of the Hessian matrix. On the subset of states that fall in the zero eigenspace of the Hessian we then directly evaluate  $\frac{d}{dt}I(\rho(t))$. Finally, we consider the states of less than full rank and describe the neighborhood of the intersection of the stationary states with the boundary of the set of states.}

The stationary states are of the form $\rho_A\otimes{\mathbbm{1}}/2$ for arbitrary $\rho_A$. 
{Since the stationary states in $\Int [S(\mathcal{H}_A\otimes \mathcal{H}_S)]$ are critical by Propositions \ref{prop3} and \ref{prop4} and such that $\frac{d}{dt}I(\rho(t))=0$, there exists some sufficiently small neighborhood of the set of stationary states where the second order terms of the Taylor expansion in the $a_{i}$ determine the sign of $\frac{d}{dt}I(\rho(t))$, in all directions where the second derivative is non-zero.
To simplify the calculation of these derivatives we note that unitary transformations on the ancilla do not change the mutual information and it is therefore sufficient to consider diagonal $\rho_A$. In other words, the purity of the state of the ancilla is the only parameter that is relevant for our analysis.
The diagonal stationary states are of the form $\frac{1}{4} \mathbbm{1}\otimes{\mathbbm{1}}+a_{12}\sigma_z\otimes{\mathbbm{1}}$ for $-1/4\leq a_{12}\leq{1/4}$. The states for which  $-1/4 < a_{12} < {1/4}$ are in $\Int [S(\mathcal{H}_A\otimes \mathcal{H}_S)]$ and the states with coordinates $a_{12}=\pm 1/4$ are at the boundary of the set of states.}

{The second derivatives at the diagonal stationary states of full rank were calculated using the method described in Sect. \ref{Taylor} and the eigenvalues of the Hessian matrix were obtained.}{The Hessian has six eigenvalues that are zero for all stationary states in $\Int[S(\mathcal{H}_A\otimes\mathcal{H}_S)]$, and for all values of the parameters $\gamma_k(t)$. The remaining nine eigenvalues are functions of the parameters $\gamma_k(t)$ and of $a_{12}$, and are given by}
\begin{eqnarray}
32[\gamma_y(t) + \gamma_z(t)]\left (\frac{16 a_{12}^2+1}{ 16 a_{12}^2-1 }\right),\nonumber\\ 
 32[\gamma_x(t) + \gamma_z(t)]\left (\frac{16 a_{12}^2+1}{ 16 a_{12}^2-1 }\right),\nonumber\\ 
 32[\gamma_x(t) + \gamma_y(t)]\left (\frac{16 a_{12}^2+1}{ 16 a_{12}^2-1 }\right),\nonumber\\
 -8 [\gamma_y(t) + \gamma_z(t)]\frac{\atanh(4 a_{12})}{a_{12}},\nonumber\\ 
 -8 [\gamma_y(t) + \gamma_z(t)]\frac{\atanh(4 a_{12})}{a_{12}},\nonumber\\ 
 -8 [\gamma_x(t) + \gamma_z(t)]\frac{\atanh(4 a_{12})}{a_{12}},\nonumber\\ 
-8 [\gamma_x(t) + \gamma_z(t)]\frac{\atanh(4 a_{12})}{a_{12}},\nonumber\\ 
 -8 [\gamma_x(t) + \gamma_y(t)]\frac{\atanh(4 a_{12})}{a_{12}},\nonumber\\ 
 -8 [\gamma_x(t) + \gamma_y(t)]\frac{\atanh(4 a_{12})}{a_{12}}.
\end{eqnarray}
These nine eigenvalues are all non-positive if and only if the conditions in Eq. (\ref{rattes}) are satisfied, i.e., if and only if the dynamics is P-divisible.
{In particular, they are all strictly negative if $\gamma_i(t)+\gamma_j(t)>{0}$ for all $i,j$.
In this case there thus exists a sufficiently small neighborhood of the stationary states where $\frac{d}{dt}I(\rho(t))$ is negative in the intersection of the neighbourhood with the complement of the zero eigenspace of the Hessian.}

{Next, we need to investigate $\frac{d}{dt}I(\rho(t))$ on the intersection of a neighbourhood around a stationary state with the zero eigenspace of the Hessian. Here we would need higher order terms in the Taylor expansion to determine the sign of $\frac{d}{dt}I(\rho(t))$, however on this eigenspace we can evaluate it directly.  
The zero eigenspace $E_0(a_{12})$ as a function of $a_{12}$, is spanned by the six eigenvectors
$(\mathbbm{1}+4a_{12}\sigma_z)\otimes{\sigma_i}$ and $\sigma_i\otimes{\mathbbm{1}}$ for $i=x,y,z$. These eigenvectors are tangent to the set of product states for all $a_{12}$, but the plane they span, i.e. $E_0(a_{12})$, also contains correlated states. 
For a given stationary state $\rho_0=\frac{1}{4} \mathbbm{1}\otimes{\mathbbm{1}}+a_0\sigma_z\otimes{\mathbbm{1}}$ we can parametrize $E_0(a_0)$. The states in the $E_0(a_0)$ are of the form}
\begin{eqnarray}
\frac{1}{4} \mathbbm{1}\otimes{\mathbbm{1}}+(\mathbbm{1}+4 a_{0}\sigma_z)\otimes{(a_{1}\sigma_x+a_{2}\sigma_y+a_{3}\sigma_z)}\nonumber\\
+ (a_4\sigma_x+ a_{8}\sigma_y+a_{12}\sigma_z)\otimes{\mathbbm{1}}.
\end{eqnarray}
{Note that $E_0(a_0)$ is an invariant subspace of $V_{s,t}$ for all $a_0$ since the dynamics is unital. Thus any state in $E_0(a_0)$ is mapped into a state also belonging to $E_0(a_0)$. Furthermore, $a_4,a_8$ and $a_{12}$ are time independent. Therefore the time derivative of the mutual information $I$ as a function on $E_0(a_0)$ depends only on the coordinates $a_1,a_2,a_3$. 
Since the mutual information is independent of unitary transformations on the system we can diagonalize $a_{1}\sigma_x+a_{2}\sigma_y+a_{3}\sigma_z$ without changing its value. Let $\pm\lambda(s)=\pm\sqrt{a_{1}^2(s)+a_{2}^2(s)+a_{3}^2(s)}$ be the corresponding eigenvalues as functions of time where}
\begin{eqnarray}
a_1(s)&&=a_{1} \, \exppp{-\int_{t}^{s}(\gamma_z(\tau)+\gamma_y(\tau))d\tau},\nonumber\\
a_2(s)&&=a_{2}\, \exppp{-\int_{t}^{s}(\gamma_z(\tau)+\gamma_x(\tau))d\tau},\nonumber\\
a_3(s)&&=a_{3}\, \exppp{-\int_{t}^{s}(\gamma_x(\tau)+\gamma_y(\tau))d\tau}.
\end{eqnarray}
{Since the only time dependence of $I$ is its dependence on $\lambda(t)$ we can express the time derivative of the mutual information as }  
\begin{eqnarray}\label{kk}
\frac{d I[\rho(t))]}{dt}=\frac{d I[\rho(s)]}{d\lambda(s)}\frac{d \lambda(s)}{ds}\big{|}_{s=t} \nonumber
\end{eqnarray}
{for any $\rho(t)\in E_0(a_0)$,}
{where $\frac{d \lambda(s)}{ds}\big{|}_{s=t}$ is given by}
\begin{eqnarray}\label{kk}
-\frac{a_{1}^2[\gamma_z(t)+\gamma_y(t)]+a_{2}^2[\gamma_x(t)
+\gamma_z(t)]+a_{3}^2[\gamma_x(t)+\gamma_y(t)]
}{\sqrt{a_{1}^2+a_{2}^2+a_{3}^2}}.
\nonumber
\end{eqnarray}
{When the conditions in Eq. (\ref{rattes}) are satisfied, i.e., when the dynamics is P-divisible, $\frac{d \lambda(s)}{ds}\big{|}_{s=t}$ is non-positive for all $a_{1},a_{2},a_{3}$. This is equivalent to stating that the length of the Bloch vector of the reduced state of the system does not increase when the dynamics is P-divisible. Moreover, we see that for all states in $E_0(a_0)$ except those of the form  $\rho_A\otimes{\mathbbm{1}}$, for which $a_1=a_2=a_3=0$, there exists some CP-divisible random unitary dynamics such that $\frac{d \lambda(s)}{ds}\big{|}_{s=t}<0$.
 Since we know that $\frac{d I[\rho(t)]}{dt}\leq 0$ for all $\rho(t)\in E_0(a_0)$ when the dynamics is CP-divisible it follows that $\frac{d I[\rho(s)]}{d\lambda(s)}\big{|}_{s=t}$ is non-negative for all $\rho(s)\in E_0(a_0)$ not of the form $\rho_A\otimes{\mathbbm{1}}$. Therefore we can conclude that $\frac{d I[\rho(t)]}{dt}\leq 0$ for all $\rho(t)\in E_0(a_0)$ when $V_{s,t}$ is P-divisible.}

{In the above analysis we have seen that there exists random unitary non-Markovian P-divisible dynamics for which no increase of the mutual information occurs in a sufficiently small neighborhood of the stationary states in $\Int[S(\mathcal{H}_A\otimes\mathcal{H}_S)]$.
It remains to consider the neighbourhood of the stationary states of full rank, i.e., of stationary states in the boundary of the set of states. For these stationary states $a_{12}=\pm 1/4$ and they are thus of the form $1/4(\mathbbm{1}+\sigma_z)\otimes{\mathbbm{1}}$ and $1/4( \mathbbm{1}-\sigma_z)\otimes{\mathbbm{1}}$.
It is sufficient to consider restricted neighborhoods of these states where the coordinate $a_{12}$ is held fixed at $1/4$ or $-1/4$ respectively. Any other point in their neighbourhoods either belongs to a neighbourhood of a stationary state in $\Int[S(\mathcal{H}_A\otimes\mathcal{H}_S)]$, or is unphysical.
The physical states in these restricted neighbourhoods for which $a_{12}=\pm 1/4$ are product states of the form $1/4( \mathbbm{1}\pm \sigma_z)\otimes\rho$, where $\rho\in B(\mathcal{H}_S)$.
This can be seen by noting that 
if $a_{12}=\pm 1/4$, one must choose $a_4=a_8=0$ to ensure non-negative eigenvalues of the reduced state on $\mathcal{H}_A$. Therefore, for all physical states in the restricted neighbourhoods the reduced state of the ancilla is pure and of the form $1/4( \mathbbm{1}\pm \sigma_z)$, which implies that all such states are product states.
Since any product state has zero mutual information and remains a product state during the evolution it follows that $\frac{d}{dt}I(\rho(t))$ is zero for all states in any neighborhood of $1/4( \mathbbm{1}\pm \sigma_z)\otimes{\mathbbm{1}}$ where $a_{12}=\pm 1/4$.}

{Finally we can conclude that for random unitary evolutions such as the quasi-eternal model described in Sect. \ref{EE} where the non-Markovian dynamics is P-divisible and all initial states have been mapped to a sufficiently small neighourhood of the stationary states by the preceding Markovian evolution no increase in the mutual information occurs. Moreover, the neighbourhood where no increase of the mutual information occurs only depends on $V_{s,t}$ and is independent of the preceding dynamics. Therefore, for any random unitary P-divisible evolution subsequent to time $t$ it is always possible to find a CP-divisible random unitary $\Lambda_t$ such that Im$(\mathcal{I}_A\otimes \Lambda_{t})$ is contained in this neighbourhood by appropriately choosing the rates $\gamma_k(\tau)>0$ for $0\leq\tau\leq t$.}

\section{Correlation measures detecting almost-all non-Markovian dynamics}\label{V0dario}

After seeing that the mutual information and entanglement measures are unable to detect all non-Markovian dynamics, in this section we review and expand the results in~\cite{previous}, where it was proven that correlation backflows can be observed for all non-Markovian dynamics that are at most point-wise non-bijective. To do so, we need to (i) move to the enlarged scenario in which two auxiliary particles are employed and (ii) consider a new correlation measure based on the distinguishability of quantum states. To present this correlation measure we need to introduce the notion of maximally entropic measurement and review the definition of guessing probability for an ensemble of quantum states.

\subsubsection{Maximally entropic measurements}\label{VA}

Consider a quantum state $\rho$. For this state, we say that a measurement, defined by a Positive Operator Valued Measure (POVM), is \textit{maximally entropic},  and denote by ME-POVM, if, when applied on the state $\rho$, each outcome has the same probability of occurrence $p_i=1/n$, where $n$ is the generic number of outputs of the POVM considered. Indeed,  if  $S(\left\{ p_i \right\}_{i=1}^n)=-\sum_i p_i \log_n p_i$ is the Shannon entropy of the resulting $n$-outcome probability distribution, where we take as the basis of the logarithm in the entropy the number of outputs, $S(\{p_i\}_{i=1}^n)=1$ if and only if $p_i = 1/n$. We define the set of ME-POVMs for $\rho$ as
\begin{equation}\label{MEPOVM}
\Pi \left( \rho \right) \equiv \left\{ \left\{ P_{i} \right\}_{i=1}^n : S(\left\{ p_i \right\}_{i=1}^n)=1  \right\} \, .
\end{equation}
For any state $\rho$, this collection is non-empty and contains measurements with any number of outputs (see Appendix \ref{neverempty}). It is trivial to adapt this notion to the case when the number of outputs is fixed to $n$, having the set of measurements $\Pi_n(\rho)$.

\subsubsection{Guessing probability of an ensemble}\label{VBsectionPg}

Consider an ensemble of states $\mathcal{E} = \{ p_i , \, \rho_{i} \}_{i}$ (with $i=1,\dots,n$) defined on a finite dimensional state space $S(\mathcal{H})$. Assume that we know the composition of the ensemble $\mathcal{E}$, and we want to answer the question:
What is the average probability to correctly identify a state extracted from $\mathcal{E}$, maximized over all possible measurements?
This quantity is called the \textit{guessing probability} of the ensemble
\begin{equation}\label{Pg}
P_g (\mathcal{E}) \equiv \max_{ \left\{ P_i \right\}_i }\sum_{i=1}^n  p_i \tr{ \rho_i  P_i  } \, ,
\end{equation}
where the maximization is performed over the space of the $n$-output POVMs. It is clear that $P_g( \mathcal{E}) \geq \overline{p}$, where $\overline{p}=\max_i \{p_i\}_i \geq 1/n$ and $P_g( \mathcal{E}) = \overline{p}$ if the states of the ensemble are identical: $\rho_i=\overline{\rho}$ for any $i=1,\dots, n$. It means that, if we fix the number of states of an ensemble to $n$, the minimum value of the guessing probability, i.e., $1/n$, is obtained if, and only if, the distribution is uniform and the states are identical. 

Note that when the ensemble is composed by two equiprobable states, i.e., $ \mathcal{E}^{eq}=\{\{p_{1,2}=1/2\},\{\rho_1, \, \rho_2\}\}$, $P_g(\mathcal{E}^{eq})$ can be expressed in terms of  the distinguishability between $\rho_1$ and $\rho_2$
\begin{equation}\label{PgD}
P_g(\mathcal{E}^{eq} ) = \frac{1}{4} \left( 2 + ||\rho_1 - \rho_2 ||_1 \right) \, ,
\end{equation}
where $|| \cdot ||_1$ is the trace norm. The distinguishability is defined as $D(\rho_1,\rho_2)=||\rho_1 - \rho_2 ||_1 /2$ and $ P_g(\mathcal E^{eq}) -1/2 = D(\rho_1,\rho_2)/2$.

\subsection{The 2-output ME-POVMs' scenario}\label{VI'A}

We now have all the ingredients needed to define a correlation measure that witnesses almost-all non-Markovian evolutions. 

Let $\rho_{AB}$ be a bipartite state acting on a finite dimensional state space of a composed system $S(\mathcal{H}_A \otimes \mathcal{H}_B)$.  Now, consider the effect of a measurement in one of the two systems, say $A$. This measurement has $n$ possible outcomes and is represented by a POVM, i.e., a collection of positive semi-definite operators  $\{ P_{A,i}\}_{i=1}^n$ where the condition $\sum_{i=1}^{n} P_{A,i} = \mathbbm{1}_A$ holds. Each $ P_{A,i}$ represents a possible outcome with the probability of occurrence $p_i$. After the measurement, depending on the outcome $i$, the state of the $B$ part is $\rho_{B,i}$. The output states and probabilities are
\begin{equation}\label{prhooutputs}
p_i= \tr{ \rho_{AB}  P_{A,i}\otimes   \mathbbm{1}_B } , \,\,\, \rho_{B,i} = \frac{\trA{ \rho_{AB} P_{A,i}\otimes   \mathbbm{1}_B}}{ p_i }\, .
\end{equation}
We define by $\mathcal{E}(\rho_{AB}, \{P_{A,i}\}_{i=1}^n)$ the output ensemble made by the probabilities $p_i$ and the states $\rho_{B,i}$ obtained by the procedure described by Eq. (\ref{prhooutputs}).

%Looking at the $B$ system, we can explain this scenario in the following equivalent way: we are randomly given one state from an ensemble of states $\mathcal{E} = \{ p_i, \, \rho_{B,i} \}_{i=1}^n$, where the $i$-th state $\rho_{B,i}$ and its frequency $p_i$ are given by Eq. (\ref{prhooutputs}). Given a state $\rho_{AB}$ and a measurement process represented by the POVM $\{ P_{A,i}\}_{i=1}^n$, we define the \textit{output ensemble} $\mathcal{E} (\rho_{AB}, \{ P_{A,i}\}_{i=1}^n )$ of this measurement by the states and the frequency rates given by Eq. (\ref{prhooutputs}).

Let's now restrict the previous analysis to ME-POVMs of $n$ outcomes for the reduced state on system $A$, $\Pi_n(\rho_A)$, where $\rho_A=\trB{\rho_{AB}}$. If $\{P_{A,i}\}_{i=1}^n \in \Pi_n(\rho_{A})$, from Eqs. (\ref{prhooutputs}) and (\ref{MEPOVM}) it follows that, 
\begin{equation}\label{ees}
\mathcal{E} (\rho_{AB}, \{ P_{A,i}\}_i )=\left\{p_i=\frac{1}{n} \, , \,\,\,
 \rho_{B,i} = n \, \trA{ \rho_{AB} P_{A,i}\otimes   \mathbbm{1}_B} \right\}_{i=1}^n  .
\end{equation}
Therefore, Alice measures $\rho_{AB}$ with  $\{P_{A,i}\}_{i=1}^n\in \Pi(\rho_{A})$, Bob obtains an \textit{equiprobable ensemble of states} (EES), namely an $n$-state output ensemble where the probability distribution of occurance of each state $\rho_{B,i}$ is uniform (see Fig. \ref{figees}).

\begin{figure}
\includegraphics[width=0.45\textwidth]{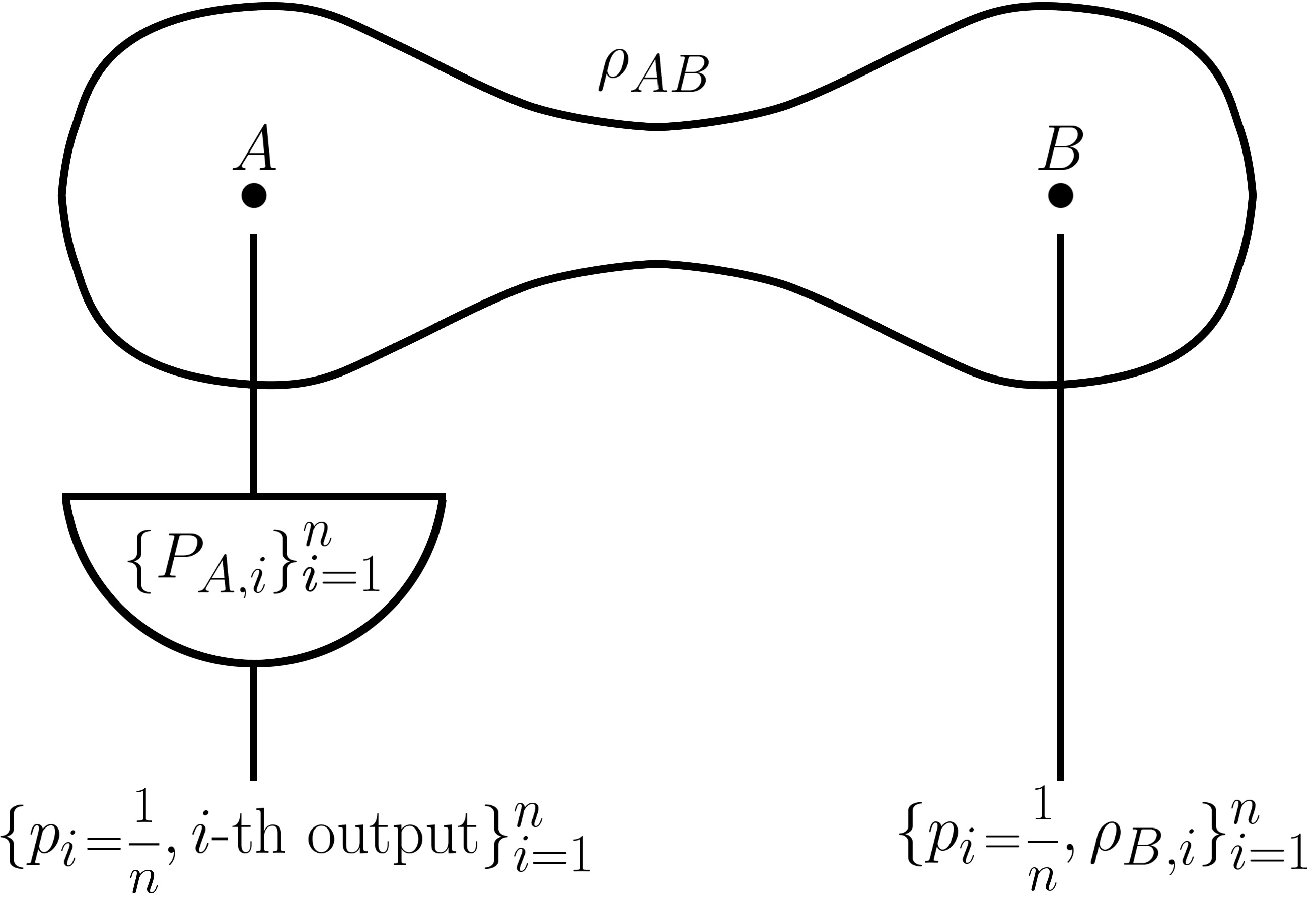}
\caption{The measurement scenario where Alice, measuring her side of $\rho_{AB}$ with an $n$-ouput ME-POVM $\{P_{A,i}\}_{i=1}^n$, produces on Bob's side the EES given by Eq. (\ref{ees}).}\label{figees}
\end{figure}

A family of correlation measures $C_A^{(n)}$ is obtained by maximizing the guessing probability of the EESs $\mathcal{E} (\rho_{AB}, \{ P_{A,i}\}_i )$ over the set of $n$-output ME-POVMs. But in fact, for most of our considerations, it suffices to consider the simplest case of 2 outputs, corresponding to a two-output ME-POVM $\{ P_{A,1}, P_{A,2}\} \in \Pi_2(\rho_A)$, having
\begin{equation}\label{CA}
 C_A^{(2)} (\rho_{AB}) \equiv \max_{ \left\{ P_{A,1}, P_{A,2}\right\} \in \Pi_2 \left( \rho_{A} \right) }  P_g \left( \mathcal{E} \left( \rho_{AB} ,\left\{ P_{A,1}, P_{A,2}\right\}  \right) \right)  - \frac{1}{2} \, ,
\end{equation}
where $\rho_A=\trB{\rho_{AB}}$ is the reduced state on $A$. In what follows, we mostly restrict our analysis to this case, although we come back to the case of an arbitrary number of outputs below. The scenario that reproduces the value of $C_A^{(2)} (\rho_{AB}) $ is described in Fig. \ref{figees}, where Alice chooses a 2-output ME-POVM that maximizes the guessing probability of the output ensemble generated on Bob's side. 
{Moreover,  we can use Eq. (\ref{PgD}) to rewrite $C_A^{(2)}(\rho_{AB})$ in the following way
\begin{equation}\label{CA}
 C_A^{(2)} (\rho_{AB}) = \max_{ \left\{ P_{A,1}, P_{A,2}\right\} \in \Pi_2 \left( \rho_{A} \right) } \frac{||\rho_{B,1} - \rho_{B,2}||_1}{2} \, ,
\end{equation}
where $\rho_{B,1}$ and $\rho_{B,2}$ are the two output states obtained when Alice applies the ME-POVM $\{P_{A,1},P_{A,2}\}$  (see Eq. (\ref{ees})).}

Alternatively, we could perform 2-output ME-POVMs on the system $B$ and obtain a measure
\begin{equation}\label{CB}
 C_B^{(2)} (\rho_{AB}) \equiv \max_{ \left\{ P_{B,1}, P_{B,2}\right\} \in \Pi_2 \left( \rho_{B} \right) }  P_g \left( \mathcal{E} \left( \rho_{AB} ,\left\{ P_{B,1}, P_{B,2}\right\}  \right) \right)  - \frac{1}{2} \, ,
\end{equation}
where $\rho_B=\trA{\rho_{AB}}$ is the reduced state on $B$.
We underline that the guessing probabilities that appear in Eq.~(\ref{CA}) and (\ref{CB}) can be evaluated using Eq. (\ref{PgD}).

A natural way to  construct a symmetric measure with respect to $A$ and $B$ is the following
\begin{equation}\label{C}
 C^{(2)} (\rho_{AB}) \equiv \max \, \left\{ C^{(2)}_A(\rho_{AB}), \, C^{(2)}_B(\rho_{AB}) \right\}  \, . 
\end{equation}
In Section \ref{C2SDP} we show that the computation of $C^{(2)} (\rho_{AB})$ can be done efficiently using SDP for any $\rho_{AB}$.
Operationally, $C^{(2)}_A (\rho_{AB})$ ($C^{(2)}_B (\rho_{AB})$) corresponds to the largest distinguishability between the pairs of equiprobable states of $B$ ($A$) that we can obtain from $\rho_{AB}$ by performing measurements on $A$ ($B$). 

We give two examples that provide an intuitive idea of the meaning of the correlation measures $C^{(2)}_A$, $C^{(2)}_B$ and $C^{(2)}$. First, we consider a generic product state $\rho_A \otimes \rho_B$. In this case, if Alice measures her side of $\rho_{AB}$ with a 2-output ME-POVMs $\{P_{A,i}\}_{i=1}^2$, the ensemble generated on Bob's side consists of the two states: $\rho_{B,i}=2\, \trA{\rho_{A}\otimes\rho_{B} \cdot P_{A,i}\otimes \mathbbm{1}_B}=\rho_B$, which are identical and equal to $\rho_B$. The corresponding guessing probability is $P_g=1/2$ and therefore $C^{(2)}_A(\rho_{A}\otimes\rho_{B})=0$. In fact, the $1/2$ factor is chosen just to make the value of the correlation measure equal to zero for product states. %It is straightforward to show $\rho_A\otimes\rho_B$ is uncorrelated also respect to $C_B^{(2)}$ and $C$. 

The second example is given by the maximally entangled state $\phi^+_{AB}=\ketbra{\phi^+_{AB}}{\phi^+_{AB}}$, where $\ket{\phi^+_{AB}}=(\ket{00}_{AB}+\ket{11}_{AB})/\sqrt{2}$. In order to evaluate $C^{(2)}_A(\phi^+_{AB})$, it is easy to realize that the projective measurement $\{P_{A,i}^{(\text{proj})}\}_{i=1}^2=\{\ketbra{0}{0}_A, \ketbra{1}{1}_A\}$ is the ME-POVM obtained by the maximization of Eq. (\ref{CA}). Indeed, in this case, Alice generates on Bob's side an orthogonal ensemble of two states: $\rho_{B,i}=2\, \trA{\phi^+_{AB} \cdot \ketbra{i}{i}_A\otimes \mathbbm{1}_B} = \ketbra{i}{i}_B$, which is perfectly distinguishable: $P_g(\{p_i=1/2, \ketbra{i}{i}_B\}_{i=1}^2)=1$. It follows that, since the guessing probability of an ensemble cannot be greater than 1, maximally entangled states are maximally correlated states with respect to $C^{(2)}_A$ and, as it is straightforward to prove, also to $C^{(2)}_B$ and $C^{(2)}$. Note however that the same maximum value can be obtained by a maximally correlated classical bit, defined by the equal mixture of states $\ket{00}$ and $\ket{11}$.

For these quantities to define a proper correlation measure, we need to demonstrate that they are non-increasing under local operations. This is proven in Appendix \ref{monolocal}, therefore showing that $C^{(2)}_A$, $C^{(2)}_B$ and $C^{(2)}$ are proper correlation measures. {Moreover, in Appendix \ref{conti} we show that the correlation measure $C^{(2)}$ is continuous on the set of states $S(\mathcal{H}_{A}\otimes\mathcal{H}_{B})$.}

\subsubsection{Computation of the correlation measure using semi-definite programming}\label{C2SDP}
A useful property of the introduced correlation measure is that its computation can be done by means of semi-definite programming (SDP). Without loss of generality, see Eqs. (\ref{CA}), (\ref{CB}) and (\ref{C}), we present this result for the case where Alice performs the ME-POVM.

The main reason why we were able to express $C_A^{(2)}$ using SDP is that for the case of two equiprobable states, $\rho_1$ and $\rho_2$, the guessing probability can be expressed as a function of the trace distance between the two states which, in turn, can be computed through SDP as follows:
\begin{eqnarray}
\label{trdist}
||\rho_1-\rho_2||_1 &=& \max \tr{P_+}+\tr{P_-}\, , \nonumber\\
\text{s.t.} &&\rho_1-\rho_2 = P_+ - P_- \, , \nonumber\\
&&P_+\,,P_-\geq 0 .
\end{eqnarray}
It now suffices to replace $\rho_1$ and $\rho_2$ by the two states prepared with probability $1/2$ on Bob's side by Alice's measurement on her part of $\rho_{AB}$. The final SDP to compute $C_A^{(2)}$ reads
\begin{eqnarray}
\label{C2}
C_A^{(2)}(\rho_{AB}) &=& \max \frac{1}{4}\left(\tr{P_+}+\tr{P_-}+2\right) \, ,\nonumber\\
\text{s.t.} &&\rho_{B,1}-\rho_{B,2} = P_+ - P_- \, ,\nonumber\\
&&\tr{ \rho_{A} \, P_{A,1} }= \tr{ \rho_{A} \, P_{A,2} } \, ,\nonumber\\
&&\rho_{B,1} = 2\trA{ \rho_{AB}  \, P_{A,1} \otimes \mathbbm{1}_B } \, ,\nonumber\\
&&\rho_{B,2} = 2\trA{ \rho_{AB}  \, P_{A,2} \otimes \mathbbm{1}_B } \, ,\nonumber\\
&&P_{A,1}+P_{A,2}=\mathbbm{1}_A \, ,\nonumber\\
&&P_+\,,P_-\,,P_{A,1}\,,P_{A,2}\geq 0 .
\end{eqnarray}
where again $\rho_A=\trA{\rho_{AB}}$ is the reduced state on $A$.

\subsubsection{Witnessing non-Markovian dynamics}\label{witnessing}

We now show how to use the correlation measure introduced above to detect non-Markovian evolutions. We prove that for any evolution that is at most point-wise non-bijective, we can find an initial state $\rho_{AB}^{(\tau)}(0)$ such that $C^{(2)} (\rho_{AB}^{(\tau)}(t))$ increases between time $t=\tau$ and $t=\tau+ \Delta t$ if and only if there is no CP intermediate map $V_{\tau+\Delta t,\tau}$.
 Although our method applies to any bijective or pointwise non-bijective evolution, at the moment we are unable to extend the proof to non-Markovian evolutions that are non-bijective in finite time intervals. Note however that the set of non-Markovian evolutions not covered by our result has zero measure in the space of evolutions. More precisely, if we take an evolution that is non-bijective in a finite time interval and add a perturbation chosen at random with respect to a Borel measure, this yields an at most point-wise non-bijective evolution with probability one \cite{ottyorke}.  

To take full advantage of this measure, we need to extend the standard setting and consider the scenario where $A$ is an ancillary qubit and $B$ is composed of the system $S$ undergoing the evolution and an additional ancilla $A'$, see Fig.\ref{fig}. This is because the measure exploits the increase of distinguishability of the states in the ensemble prepared by the measurement in $A$ under the action of a non-CP map. However, for $P$-divisible dynamics, this increase does not appear unless correlated states, for particles $S$ and $A'$, are considered \cite{koss2,ruskai}

In what follows we first construct the initial state $\rho_{AB}^{(\tau)}(t)$ to be used as a probe. Second, we show that for the class of non-Markovian dynamics specified above, $C^{(2)}( \rho_{AB}^{(\tau)}(t))$ provides a correlation backflow.

\subsubsection{The probe}
Let $\Lambda_t$ represent a bijective or pointwise non-bijective non-Markovian dynamical map that acts on the system $S$ and introduce an ancillary system $A'$.
As shown in Ref. \cite{bogna}, for any of these dynamics we can construct a class of pairs of initial states $\{ \rho_B'^{(\tau)}(0), \rho_B''^{(\tau)}(0)\}\in S(\mathcal{H}_B )= S (\mathcal{H}_{A'} \otimes \mathcal{H}_{S})$ that show an increase in distinguishability between time $t=\tau$ and $t=\tau +\Delta t$
\begin{equation}\label{rho12}
\big| \big|  \rho_B'^{(\tau)}(\tau +\Delta t) - \rho_B''^{(\tau)}(\tau +\Delta t)\big| \big|_1 >\big| \big|  \rho_B'^{(\tau)}(\tau ) - \rho_B''^{(\tau)}(\tau )\big| \big|_1  \, ,
\end{equation}
if and only if there is no CP intermediate map $V_{\tau+\Delta t, \tau}$, where the evolution of the system $B$ is given by the dynamical map $\mathcal{I}_{A'}\otimes \Lambda_{t}$, where $\mathcal{I}_{A'}$ is the identity map on $A'$.

The particular bipartite separable states $\rho_{AB}^{(\tau)}(t)$ for which we examine the correlation $C^{(2)}$ are classical-quantum states%\cite{wilde}
. Our ``probe'' state is
\begin{equation}\label{hatrho}
\rho_{AB}^{(\tau)}(t) \equiv \frac{1}{2} \left( \ketbra{0}{0}_A \otimes \rho_B'^{(\tau)} (t)+ \ketbra{1}{1}_A \otimes \rho_B''^{(\tau)} (t) \right) \, ,
\end{equation}
where and $ \rho_B'^{(\tau)}(t)$ and $ \rho_B''^{(\tau)}(t) $ are the states that appear in Eq. (\ref{rho12}) and $\mathcal{B}_A \equiv\{\ket{0}_A, \ket{1}_A \}$ is an orthonormal basis for $\mathcal{H}_A$. Since only the system $B$ is involved in the evolution, $\rho_{AB}^{(\tau)}(t)$ is given by Eq. (\ref{hatrho}) for any $t\geq 0$. Note that from Eq. (\ref{hatrho}) it follows that $\rho_{AB}^{(\tau)}(t)$ does not contain any entanglement. Moreover, the state can be chosen arbitrarily close to an uncorrelated state since, as shown in \cite{bogna}, one can always choose states $\rho_B'^{(\tau)}(0)$ and $\rho_B''^{(\tau)}(0)$ arbitrarily close to each other.

\subsubsection{Detecting the correlation backflow}\label{finalwitness}
 In this section we show how the correlation measure $C^{(2)}_A(\rho_{AB}^{(\tau)}(t))$, and later $C^{(2)}(\rho_{AB}^{(\tau)}(t))$, witnesses bijective {or pointwise non-bijective} non-Markovian dynamics. 

To evaluate $C^{(2)}_A(\rho_{AB}^{(\tau)}(t))$, given the maximization used in Eq. (\ref{CA}), we have to find a ME-POVM $\{P_{A,1},P_{A,2}\}$ that, applied on $\rho_{AB}^{(\tau)}(t)$, generates the output ensemble $\{ \{p_{1,2}=1/2\},\{\rho_{B,1}(t),\rho_{B,2}(t)\}\}$ with the largest value of $||\rho_{B,1}(t)- \rho_{B,2}(t) ||_1$. Let $\lambda \in[0,1]$  and $\eta\in[0,1]$  be the diagonal elements of $P_{A,1}$ in the basis $  \mathcal{B}_A=\{\ket{0}_A, \ket{1}_A \}$. It is easy to show that $\lambda+\eta=1$ for ME-POVMs. The corresponding output states are
\begin{eqnarray}
\rho_{B,1}(t)&=&{\lambda \rho_{B}'^{(\tau)}(t) + \eta \rho_B''^{(\tau)}(t)} \, , \\
\rho_{B,2}(t)&=&{(1-\lambda) \rho_{B}'^{(\tau)}(t) +(1-\eta) \rho_B''^{(\tau)}(t)} \, .
\end{eqnarray}
It follows that
\begin{equation}
||\rho_{B,1}(t)- \rho_{B,2} (t) ||_1 =  {|\lambda-\eta| \, \left| \left|\rho_{B}'^{(\tau)}(t) - \rho_B''^{(\tau)}(t) \right| \right|_1 }    \, . 
\end{equation}
Since $0\leq |\lambda-\eta| \leq 1$, the maximum is obtained when either $\lambda$ or $\eta$ is equal to 1. In both cases the output states are $\rho_B'^{(\tau)}(t)$ and $\rho_B''^{(\tau)}(t)$ and 
\begin{equation}\label{CAhatrho2}
C^{(2)}_A( \rho_{AB}^{(\tau)}(t)) =  \frac{\left| \left| \rho_B'^{(\tau)} (t) - \rho_B''^{(\tau)}(t) \right| \right|_1}{4}   \, .
\end{equation}
In Appendices \ref{CB2CA} and \ref{2enoughA} we prove that $C^{(2)} (\rho_{AB}^{(\tau)}(t)) = C^{(2)}_A(\rho_{AB}^{(\tau)} (t) ) \geq {C_B^{(2)} (\rho_{AB}^{(\tau)} (t) )}$.
Therefore, using Eqs. (\ref{rho12}) and (\ref{CAhatrho2}), we obtain a correlation backflow
\begin{equation}\label{ddtCA}
C^{(2)}\left(\rho_{AB}^{(\tau)}(\tau+\Delta t) \right) > C^{(2)}\left(\rho_{AB}^{(\tau)}(\tau) \right)  \, ,
\end{equation}
if and only if there is no CP intermediate map $V_{\tau+\Delta t,\tau}$.

\subsubsection{Example: the quasi-eternal non-Markovian model}\label{VI'C}

For the sake of clarity, we illustrate the previous general results through a specific dynamics. Let us consider the example introduced in Section \ref{quasieternal}, namely the class of dynamical maps $\Lambda_t^{(t_0,\alpha)}$ where $\alpha>0$ and $t_0\geq0$ satisfies the relation given in Eq. (\ref{alphat0}). Recall that in Section \ref{EE} we showed that the mutual information fails to detect some non-Markovian dynamics that belong to this class.

Now, we show how to build the states $\rho_B'^{(\tau)}(t_0)$ and  $\rho_B''^{(\tau)}(t_0)$ that appear in Eqs. (\ref{rho12}) and (\ref{hatrho}) and therefore $\rho^{(\tau)}_{AB}(t_0)$ itself, when the evolution is given by $\Lambda_t^{(t_0,\alpha)}$. Picking $\tau> t_0$,  the intermediate map $V_{\tau+\Delta t,\tau}^{(t_0,\alpha)}$ is not CP for any $\Delta t>0$ (see Section \ref{witnessing}). The constructive method given in \cite{bogna} suggests to consider, together with the qubit $S$ evolved by $\Lambda_t^{(t_0,\alpha)}$, an ancillary qutrit $A'$: $S(\mathcal{H}_B )= S (\mathcal{H}_{A'} \otimes \mathcal{H}_{S})$. Now, being $\{\ket{0}_{A'},\ket{1}_{A'},\ket{2}_{A'}\}$ and $\{\ket{0}_{S},\ket{1}_{S}\}$ orthonormal basis respectively for $\mathcal{H}_{A'}$ and $\mathcal{H}_S$, we have:
\begin{equation}\label{rhoB'}
\rho_{A'S}'^{(\tau)} (\tau)=(1-p) \sigma_{A'S} + p \phi^+_{A'S} \, ,
\end{equation}
\begin{equation}\label{rhoB''}
\rho_{A'S}''^{(\tau)} (\tau)= (1-p) \sigma_{A'S} + p  \ketbra{2}{2}_{A'} \otimes \rho_S \, .
\end{equation}
$\phi^+_{A'S} \equiv \ketbra{\phi^+}{\phi^+}_{A'S}$ is the maximally entangled state, where $\ket{\phi^+}_{A'S}\equiv (\ket{00}_{A'S} + \ket{11}_{A'S})/\sqrt{2}$ and $\sigma_{A'S} \in \Int [ \I (\Lambda_{\tau}^{(t_0,\alpha)}) ]$. 

In order to define completely $\rho_{A'S}'^{(\tau)} (\tau)$ and $\rho_{A'S}''^{(\tau)} (\tau)$, we fix their free components: $\sigma_{A'S} \equiv (\ketbra{0}{0}_{A'} + \ketbra{1}{1}_{A'} )/2 \otimes \mathbbm{1}_S/2 $ and  $\rho_S \equiv \mathbbm{1}_S/2$ and we get:
\begin{eqnarray}\label{rrhoB'}
\rho_{A'S}'^{(\tau)} (\tau)&=&\frac{ (\ketbra{0}{0}_{A'} + \ketbra{1}{1}_{A'} )\otimes \mathbbm{1}_S}{4}\nonumber\\
&+&p\,\frac{ \sigma_x\otimes \sigma_x-\sigma_y\otimes \sigma_y+\sigma_z\otimes \sigma_z}{4} \, ,
\end{eqnarray}
\begin{equation}\label{rrhoB''}
\rho_{A'S}''^{(\tau)} (\tau)= \left( (1-p) \frac{ \ketbra{0}{0}_{A'}+\ketbra{1}{1}_{A'}}{2}+p\ketbra{2}{2}_{A'}\right) \otimes \frac{\mathbbm{1}_S}{2} \, .
\end{equation}

Considering the rates given in Eq. (\ref{gammaxyz}), the evolution induced by the dynamical map $\Lambda_{\tau}^{(t_0,\alpha)}$ (see Eqs. (\ref{randomunit})) that precedes the action of $V^{(t_0,\alpha)}_{\tau+\Delta t,\tau}$, can be written as:
\begin{eqnarray}\label{ru}
\Lambda_{\tau}^{(t_0,\alpha)}(\sigma_x)&&=\left( \expp{-\tau}  \frac{\cosh(\tau-t_0)}{\cosh(t_0)} \right)^{\alpha/2}\sigma_x \equiv \lambda_{xy}^{(t_0,\alpha)}(\tau) \, \sigma_x \, ,\nonumber\\
\Lambda_{\tau}^{(t_0,\alpha)}(\sigma_y)&&=\left( \expp{-\tau}  \frac{\cosh(\tau-t_0)}{\cosh(t_0)} \right)^{\alpha/2} \sigma_y\equiv \lambda_{xy}^{(t_0,\alpha)}(\tau) \,  \sigma_y \, ,\nonumber\\
\Lambda_{\tau}^{(t_0,\alpha)}(\sigma_z)&&=\expp{-\alpha \tau} \sigma_z\equiv \lambda_{z}^{(\alpha)}(\tau) \, \sigma_z \, , \nonumber\\
\Lambda_{\tau}^{(t_0,\alpha)}(\mathbbm{1})&&=\mathbbm{1} \, ,
\end{eqnarray}
where, for $\tau>t_0$, we have $\lambda_{xy}^{(t_0,\alpha)}(\tau)>\lambda_{z}^{(\alpha)}(\tau)$.
The state $\rho_{A'S}''^{(\tau)} (\tau) $, since it assumes the form $\rho_{A'}\otimes \mathbbm{1}_S/2$ and the evolution is random unitary, it   is stationary for $\mathcal{I}_{A'}\otimes \Lambda_{t}^{(t_0,\alpha)}$. Therefore, $\rho_{A'S}''^{(\tau)} (0)=( \mathcal{I}_{A'}\otimes \Lambda_{t}^{(t_0,\alpha)})^{-1} (\rho_{A'S}''^{(\tau)} (\tau))=\rho_{A'S}''^{(\tau)} (\tau)$. Conversely, $\rho_{A'S}'^{(\tau)} (\tau)$  is not a stationary state and
$( \mathcal{I}_{A'}\otimes \Lambda_{\tau}^{(t_0,\alpha)})^{-1} (\rho_{A'S}'^{(\tau)} (\tau))$ is not physical for every $p\in[0,1]$. Indeed, we can write
\begin{eqnarray}\label{rrrhoB'}
&\rho_{A'S}'^{(\tau)}& (0)= \left(\Lambda_{\tau}^{(t_0,\alpha)}\right)^{-1}( \rho_{A'S}'^{(\tau)} (\tau) )=\frac{ (\ketbra{0}{0}_{A'} + \ketbra{1}{1}_{A'} )\otimes \mathbbm{1}_S}{4}\nonumber\\
&+&\frac{p}{\lambda_{xy}^{(t_0,\alpha)}(\tau)} \frac{ \sigma_x\otimes \sigma_x-\sigma_y\otimes \sigma_y}{4} + \frac{p}{\lambda_{z}^{(\alpha)}(\tau)} \frac{\sigma_z\otimes \sigma_z}{4} \, ,
\end{eqnarray}
which is a physical state for ${p}/{\lambda_{xy}^{(t_0,\alpha)}(\tau)}<{p}/{\lambda_{z}^{(\alpha)}(\tau)}\leq 1$ (see Section \ref{IVDrand}). Therefore, if $p$ satisfies the condition $p<\lambda_{z}^{(\alpha)}(\tau)$, Eq. (\ref{rrrhoB'}) is the physical initial state $\rho_{A'S}'^{(\tau)} (0)$ which at time $\tau$ is evolved to the state given in Eq. (\ref{rrhoB'}), and therefore, together with $\rho_{A'S}''^{(\tau)} (0)$, fulfills the requirements of the constructive method given in \cite{bogna}. 

The construction of the state $\rho_{AB}^{(\tau)}(0)$ is now straightforward (see Eq. (\ref{hatrho})). Being $A$ an ancillary qubit for which we adopt the orthonormal basis $\{\ketbra{0}{0}_A,\ketbra{1}{1}_A\}$, we have
\begin{eqnarray}\label{explicitprobe}
\rho_{AB}^{(\tau)}(0) =\frac{1}{2}\left(\ketbra{0}{0}_A \otimes \rho_{B}'^{(\tau)} (0) + \ketbra{1}{1}_A \otimes \rho_{B}''^{(\tau)} (0) \right) \, ,
\end{eqnarray}
where $ \rho_{B}'^{(\tau)} (0)$ is given in Eq. (\ref{rrrhoB'}) and $ \rho_{B}''^{(\tau)} (0)$ in Eq. (\ref{rrhoB''}).

\subsection{Scenarios where ME-POVMs have more than two outputs}\label{VI'B}

As mentioned, similar correlation measures $C^{(n)}$ can be obtained by fixing the number of outputs of the ME-POVMs to any integer $n\geq 3$. Indeed, similarly to Eqs. (\ref{CA}), (\ref{CB}) and (\ref{C}), we define:
\begin{equation}\label{CAn}
C_A^{(n)} (\rho_{AB}) \equiv \max_{ \{P_{A,i}  \}_{i=1}^n \in \Pi_n \left( \rho_{A} \right) }  P_g \left( \mathcal{E} \left( \rho_{AB} ,\left\{ P_{A,i}\right\}_{i=1}^n  \right) \right)  - \frac{1}{n} \, ,
\end{equation} 
\begin{equation}\label{CBn}
C_B^{(n)} (\rho_{AB}) \equiv \max_{ \{P_{B,i}  \}_{i=1}^n \in \Pi_n \left( \rho_{B} \right) }  P_g \left( \mathcal{E} \left( \rho_{AB} ,\left\{ P_{B,i}\right\}_{i=1}^n  \right) \right)  - \frac{1}{n} \, ,
\end{equation}
and
\begin{equation}\label{Cn}
C^{(n)} (\rho_{AB}) \equiv \max \left\{C_A^{(n)} (\rho_{AB}),C_B^{(n)} (\rho_{AB}) \right\} \, .
\end{equation}
The measuring scenario considered by $C_A^{(n)}(\rho_{AB})$, and similarly for $C_B^{(n)}(\rho_{AB})$, is described in Fig. \ref{figees}. Indeed, the value of $C_A^{(2)} (\rho_{AB}) $ is given when Alice chooses an $n$-output ME-POVM that maximizes the guessing probability of the output ensemble generated on Bob's side.

Moreover, we define $ C(\rho_{AB}) \equiv \max \, \{ C_A (\rho_{AB}), \, C_B (\rho_{AB}) \} $, where $C_A(\rho_{AB})$ ($C_B(\rho_{AB})$) is obtained without fixing the number of outputs of the ME-POVMs in $\Pi(\rho_A)$ ($\Pi(\rho_B)$), namely
\begin{equation}\label{C_Aorig}
C_A (\rho_{AB}) \equiv \max_{ \{P_{A,i}  \}_i \in \Pi \left( \rho_{A} \right) }  P_g \left( \mathcal{E} \left( \rho_{AB} ,\left\{ P_{A,i}\right\}_i  \right) \right)  - \frac{1}{2} \, . 
\end{equation} 
We define $C_B(\rho_{AB})$ similarly.

The details of the proof that shows that $C^{(n)}$ (for any $n\geq 3$) and $C$ are valid correlation measures, i.e., satisfy the conditions given in Sec. \ref{III0corrme}, can be found in Appendix \ref{monolocal}. {Moreover, in Appendix \ref{conti} we show that the correlation measures $C^{(n)}$ are continuous on the set of states  for any $n\geq 2$ and in Appendix \ref{bounded} we give a bound for the number of outcomes of the ME-POVMs that Alice has to consider in order to evaluate $C^{(n)}_A(\rho_{AB})$.}

In Section \ref{VI'A} we gave a physical interpretation of the value of the measure $C^{(2)}_A(\rho_{AB})$ in terms of the most distinguishable pair of output states that Alice, performing two-output ME-POVMs on her share of $\rho_{AB}$, can produce on Bob's side. Increasing to $n\geq 3$ the number of outputs of the ME-POVMs adopted by Alice, this interpretation can be easily adopted also for $C^{(n)}_A$. 

The meaning of the correlation measure $C_A$, and consequently of $C_B$ and $C$, is slightly different. Indeed, $C_A(\rho_{AB})$ represents the guessing probability of the most distinguishable ensemble of states that Alice can produce on Bob's side applying \textit{any} ME-POVM on $\rho_{AB}$, namely without fixing the number of states of the output ensemble obtained by Bob (although the number of outputs should be larger than one). Note that for product states the best measurement consists of only two outputs.

Making use of the probe state $\rho^{(\tau)}_{AB}(t)$ given in Eq. (\ref{hatrho}), also the correlation measure $C$ is able to witness almost-all non-Markovian evolutions. Indeed, in Appendices \ref{2enough} and \ref{2enoughA} we prove that for this state $C_B(\rho_{AB}^{(\tau)}( t))=C_B^{(2)}(\rho_{AB}^{(\tau)}( t))$ and $C_A(\rho_{AB}^{(\tau)}(t)) = C_A^{(2)}(\rho_{AB}^{(\tau)}(t))$. From these results it follows that for this initial probe state $C^{(2)} (\rho_{AB}^{(\tau)}( t)) = C (\rho_{AB}^{(\tau)}( t)) $ at any time $t\geq 0$. Therefore, for this state, $ C (\rho_{AB}^{(\tau)}( t))$ can {also} be calculated efficiently using SDP and shows the same correlation backflow given in Eq. (\ref{ddtCA}).

\subsection{Problems to witness any non-Markovian dynamics}\label{VFbusc}

The construction of the correlation measure is inspired by the results in~\cite{bogna}, valid for distinguishability measures for dynamics that at most point-wise non-bijective. However, in Ref. \cite{Buscemi&Datta} it was shown that for any non CP-divisible intermediate map $V_{s,t}$ an ensemble of states $\overline {\mathcal{E}}=\{p_i,\rho_i\}$ exists such that the guessing probability $P_g(\overline {\mathcal{E}})$ increases between time $t$ and $s$. How to construct such an ensembles is not known and there is no reason to believe that an equiprobable ensemble with this property exists for a general dynamics. A natural question is therefore if a correlation measure of the type described in Sect. \ref{V0dario} can be defined that shows an increase when $P_g(\mathcal{\bar{E}})$ increases.

A first attempt in this direction could go as follows. It is possible to generalize the correlation measure $C_A(\rho_{AB})$ by performing a maximization of $P_g$ over the set of POVMs $\Pi_A^p(\rho_{A})$ that produce some given fixed output probabilities $\{p_i\}$ when acting on $\rho_A$, not necessarily fully entropic. The resulting correlation measure  $C_A^p(\rho_{AB})$ reads
\begin{eqnarray}
C_A^P(\rho_{AB})\equiv\max_{\{P_A^i\}_i\in \Pi_A^p(\rho_{A})}P_g[\mathcal{E}(\rho_{AB},\{P_A^i\}_i)].\end{eqnarray}
The idea would be to mimic the previous result by using as initial probe a state consisting of the ensemble states of Ref.~\cite{Buscemi&Datta} for systems $S$ and $A'$ correlated with orthogonal states on $A$ according to the probabilities in the ensemble.

Unfortunately, this approach does not work. In fact, we show that for a classical-quantum state $\rho_{AB}$ of the type 
\begin{eqnarray}\label{ssst}
\rho_{AB}=\sum_i p_i |i\rangle\langle i|\otimes \rho_i.
\end{eqnarray}
the ensemble that maximizes $P_g$ is not defined by the projections on the states $\ket{i}$ in system $A$.
We show this by constructing a counterexample. Consider the state
\begin{eqnarray}
\rho_{AB}=p_1 |1\rangle\langle 1|\otimes \rho_1+p_2 |2\rangle\langle 2|\otimes \rho_2+p_3 |3\rangle\langle 3|\otimes \rho_3,
\end{eqnarray}
where $2p_3>p_1>2p_2$ and
\begin{eqnarray}
\rho_1=\left(\begin{array}{cc}
\frac{1}{2} & 0 \\
0 & \frac{1}{2}\end{array}\right),\rho_2=\left(\begin{array}{cc}
1 & 0 \\
0 & 0\end{array}\right),\rho_3= \left(\begin{array}{cc}
0 & 0 \\
0 & 1\end{array}\right).
\end{eqnarray}
For this example we can directly find $P_g$. To do so we use that $P_g$ is the solution to a convex optimization problem where strong duality holds%, i.e., $P_g$ is the solution to the dual optimization problem 
~\cite{yuen,eldar}. The dual formulation of $P_g$ is
\begin{eqnarray}
P_g=\min_{K} \Tr K\phantom{o}\textrm{s.t.}\phantom{o} K\geq p_i\rho_i \phantom{o}\forall i.
\end{eqnarray}

For the output ensemble $\mathcal{{E}}=\{\{p_1,p_2,p_3\}, \{\rho_1,\rho_2,\rho_3\}\}$, achieved by the p-POVM $\{|1\rangle\langle 1|,|2\rangle\langle 2|,|3\rangle\langle 3|\}$, it can be seen that $P_g=p_1/2+p_3$ by directly constructing the optimal $K$. However, for the output ensemble $\{\{p_1,p_2,p_3\}, \{(1-p_2/p_1)\rho_1+p_2/p_1\rho_2,\rho_1,\rho_3\}\}$ achieved by a non-projective p-POVM $\{(1-p_2/p_1)|1\rangle\langle 1|+|2\rangle\langle 2|,p_2/p_1|1\rangle\langle 1|,|3\rangle\langle 3|\}$ the construction of the optimal $K$ gives $P_g=p_1/2+p_2/2+p_3$. 
Thus, in general the maximization of $P_g$ over $\Pi_A^p(\rho_{A})$ does not produce the desired ensemble for systems $S$ and $A'$. Therefore, it is not clear to us how to adapt the results in~\cite{Buscemi&Datta} into a correlation measure that detects all non-Markovian evolutions. %It is therefore not clear to what extent a correlation measure of the type $C^p(\rho_{AB})=\max(C_A^p(\rho_{AB}),C_B^p(\rho_{AB}))$  can witness non CP-divisible dynamics since the required ensembles of Ref. \cite{Buscemi&Datta} may or may not be of a form that can be produced by a p-POVM acting on a classical-quantum state.

\section{Quasi-correlation measures}

We would like to conclude our study by discussing the use of what we call quasi-correlation measures in the context of non-Markovian detection. Note that while for a correlation measure it is demanded that it does not increase by local operations, for the detection of non-Markovianity it suffices to consider functions that do not increase under the action of operations by one of the parties (the one evolving through the dynamics). An increase on the value of these measures is enough to detect the failure of CP divisibility. We name quasi-correlation measures those functions of a bipartite state that do not increase when applying operations on only one half of the state.

An example of such measures is the version of the singlet fraction defined in Ref.~\cite{renner}. Given a bipartite state $\rho_{AB}$, it is defined as 
\begin{equation}
\label{singlfr}
\text{SF}(\rho_{AB})=\max_\epsilon \bra{\phi^+}(\mathbbm{1}_A\otimes\epsilon)(\rho_{AB})\ket{\phi^+} ,
\end{equation}
that is the maximum of the fidelity with a maximally entangled state optimised over local operations by one of the parties. Other definitions of the singlet fraction allow operations by both parties and include classical communication. It is easy to see that this measure cannot increase by local operations on the second system, as any further local processing can always be adsorbed in the optimization in~\eqref{singlfr}. 

It immediately follows from previous results that this measure detects all non-Markovian dynamics. In~\cite{renner}, it was shown that for classical-quantum correlated states of the form 
\begin{equation}
\rho_{\text{CQ}}=\sum_i p(i) \ket i \bra i \otimes \rho_i ,
\end{equation}
the singlet fraction is equal to the guessing probability of the ensemble $\{p(i),\rho_i\}$. Therefore, we can combine this with the results in~\cite{Buscemi&Datta}, proving the existence of an ensemble with increasing guessing probability for any non-Markovian dynamics, to conclude that this version of the singlet fidelity also detects all such dynamics.

\section{Discussion}\label{VI0}

Understanding the operational consequences of non-Markovian effects in terms of information backflow is a fundamental question. In this work, we focus on correlations and study how they can be used to detect the failure of CP-divisibility. We have identified strengths and weaknesses of several known correlation measures. In particular, we have shown that:
\begin{itemize}
\item {Non-Markovian effects in single-parameter dynamics, such as depolarization, dephasing or amplitude damping, always lead to correlation backflows for any continously differentiable measure that is not time-independent on the image of the the preceding evolution};
\item {It is possible to detect backflows in the mutual information for any qubit unital non-P-dividible dynamics};
\item Maximally entangled states are not necessarily optimal for observing backflows in the mutual information;
\item {There exist quasi-eternal non-Markovian dynamics with no backflow in the mutual information};
\item  {There exist quasi-correlation measures that an be used for non-Markovianity detection and always show a backflow.}
\end{itemize}

{In addition to this we have revisited and expanded upon or exemplified several points from Ref. \cite{previous}:}
\begin{itemize}
\item {Measures of entanglement between the system and an ancilla cannot provide any backflow in those cases where the non-Markovian dynamics is P-divisible and appear only after the dynamics has become entanglement breaking};
\item {Measures based on state distinguishability that detect almost all non-Markovian dynamics can be constructed in a setting with two additional particles}.
\end{itemize}

Our results clarify many issues but also point to several open questions. The most obvious one is to construct a correlation measures able to detect all dynamics that are not CP-divisible, either by adapting the results in~\cite{Buscemi&Datta} to our approach, or by considering a novel approach. A second open question is to understand if the use of the second additional particle can be of use for existing correlation measures, such as those based on entanglement or mutual information. 

{As mentioned, while completing this work, we became aware of the work \cite{Janek} where it is shown that negativity, a computable entanglement measure, is able to detect all the bijective non-Markovian dynamics, and even non-bijective in the case of qubits, in the enlarged scenario using two auxiliary particles. Note that this second particle is also crucial to extend the detection power of our correlation measure and detect all at most point-wise non-bijective evolutions. Whether similar effects can be observed for other correlation measures, such as the mutual information, deserves further investigation.}

\begin{acknowledgments}
Support from the ERC CoG QITBOX, the AXA Chair in Quantum Information Science,
Spanish MINECO (QIBEQI FIS2016-80773-P and Severo Ochoa SEV-2015-0522),
Fundaci\'o Privada Cellex, and the
Generalitat de Catalunya (SGR1381 and CERCA Program) is acknowledged.
% Spanish MINECO (FOQUS FIS2013-46768-P and SEV-2015-0522), Fundacion Cellex, and the Generalitat de Catalunya (SGR 875and CERCA Program) is acknowledged.
D.D.S. acknowledge support from the ICFOstepstone programme, funded by the Marie Sk\l odowska-Curie COFUND action (GA665884).
B.B. acknowledge support by the ICFO-MPQ
Fellowship.
Fellowship. N.K.B. thanks the Brazilian funding agency CAPES.
\end{acknowledgments}

\appendix

\section{Nonzero time derivatives of initially zero eigenvalues of rank one matrices for non-unitary single parameter maps}\label{dermu}

Let $\rho$ be a positive semidefinite Hermitian matrix.
Consider an eigenvalue $\lambda_k$ of $\rho$ and its corresponding normalized eigenvector $u_k$. By definition it holds that $\lambda_k=u_k^{\dagger}\rho u_k$. The time derivative of $\lambda_k$ is 

\begin{eqnarray}\frac{d \lambda_k}{dt}&=&\frac{d u_k^{\dagger}}{dt}\rho u_k+u_k^{\dagger}\rho \frac{du_k}{dt}+u_k^{\dagger}\frac{d\rho}{dt} u_k\nonumber\\
&=&\lambda_k\left(\frac{d u_k^{\dagger}}{dt}u_k+u_k^{\dagger} \frac{du_k}{dt}\right)+u_k^{\dagger}\frac{d\rho}{dt} u_k.\nonumber\\
\end{eqnarray}

 Since $u_k$ is normalized it follows that its derivative is orthogonal to $u_k$, i.e., $\frac{d u_k^{\dagger}}{dt}u_k=0$. If the evolution of $\rho$ is described by a continuously differentiable family of dynamical maps so that $\frac{d\rho}{dt}=\frac{d}{ds}V_{s,t}(\rho)|_{s=t}=\mathcal{L}_t(\rho)$ it follows that $\frac{d \lambda_k}{dt}=u_k^{\dagger}\mathcal{L}_t(\rho) u_k$.

Next consider the special case where $\rho$ is a a rank one positive semidefinite trace one $n\times n$ matrix and consider its block diagonal form.

\begin{eqnarray}
\rho=\left(\begin{array}{cc}
1 & \mathbf{0} \\
\mathbf{0} & \mathbf{0}\end{array}\right),\end{eqnarray} 
where $1$ represents the $1\times 1$ block corresponding to the nonzero eigenvalue $1$ and $\mathbf{0}$ represents the $n\times 1$, $1\times n$, and $n\times n$ zero blocks. We want to investigate $\frac{d}{ds}V_{s,t}(\rho)|_{s=t}=\mathcal{L}_t(\rho)$ for the case of single parameter evolution. In particular we want to study the projection of $\mathcal{L}_t(\rho)$ onto the zero-eigenspace of $\rho$.

First we consider the unitary part of $\mathcal{L}_t$. Let $P_0$ be the projector onto the zero eigenspace of $\rho$. One easily finds that $P_0[H,\rho]P_0=0$ for any $H$.

Then we consider $\left( G_k  \rho G_k^{\dagger}-\frac{1}{2}\left\{G_k^{\dagger} G_k,\rho\right\}\right)$ and express the matrix $G_k$ on the same block form as $\rho$, i.e.,
\begin{eqnarray}
G_k=\left(\begin{array}{cc}
A_k & B_k \\
C_k & D_k\end{array}\right),
\end{eqnarray}
where $A_k$ is the $1\times 1$ block. The projection of $\left( G_k  \rho G_k^{\dagger}-\frac{1}{2}\left\{G_k^{\dagger} G_k,\rho\right\}\right)$ onto the zero eigenspace of $\rho$ is then 
\begin{eqnarray}
P_0\left( G_k  \rho G_k^{\dagger}-\frac{1}{2}\left\{G_k^{\dagger} G_k,\rho\right\}\right)P_0= C_kC_k^{\dagger}.
\end{eqnarray}
The matrix $C_kC_k^{\dagger}$ is clearly Hermitian and positive semidefinite. It follows that 
\begin{eqnarray}
P_0\mathcal{L}_tP_0= \gamma(t)\sum_kC_kC_k^{\dagger},
\end{eqnarray}
is also a positive semidefinite matrix if $\gamma(t)>0$ and negative semidefinite if $\gamma(t)<0$. Moreover, $P_0\mathcal{L}_tP_0$ is zero if and only if $C_k$ is zero for every $k$.
Hence, if and only if for each $k$ the lower off-diagonal $n\times 1$ block of $G_k$ is zero in every basis will there be no rank one $\rho$ such that $P_0\mathcal{L}_tP_0$ is nonzero. In this case $G_k$ is proportional to the identity which implies $\left( G_k  \rho G_k^{\dagger}-\frac{1}{2}\left\{G_k^{\dagger} G_k,\rho\right\}\right)=0$ for every $\rho$.
Thus, for any $\mathcal{L}_t$ with non-zero dissipative part there exist at least one rank one $\rho$ such that the time derivative of the initial zero-eigenspace is nonzero.

To analyse the special case when $\rho=\phi^+_n$ where $\phi^+_n$ is the maximally entangled state on $\mathcal{H}\otimes\mathcal{H}$ and $\dim(\mathcal{H})=n$
we note that the condition $C_k=0$ can be formulated as $G_k\rho=\rho G_k \rho$.
We write $\phi^+=1/n\sum_{ij}E_{ij}\otimes E_{ij}$ where $E_{ij}$ is the matrix with the $ij$th element equal to $1$ and all other elements equal to zero.  We write $G_k=\mathbbm{1}\otimes G$ where $G$ is any matrix. Then

\begin{eqnarray}
G_k\rho&=&\frac{1}{n}\sum_{ij}E_{ij}\otimes G E_{ij}\nonumber\\
\rho G_k \rho&=&\frac{1}{n^2}\sum_{ijl}E_{li}E_{ij}\otimes E_{li} G E_{ij}=\frac{\Tr(G)}{n^2}\sum_{jl}E_{lj}\otimes E_{lj}.\nonumber\\
\end{eqnarray}
These two expressions are equal if and only if $G E_{ij}=\Tr(G)/nE_{ij}$ for each $ij$. Since the matrices $E_{ij}$ form a basis for the matrix space it follows that  $G_k\rho=\rho G_k \rho$ if and only if $G_k\propto \mathbbm{1}\otimes \mathbbm{1}$. As noted before this implies $\left( G_k  \rho G_k^{\dagger}-\frac{1}{2}\left\{G_k^{\dagger} G_k,\rho\right\}\right)=0$ for every $\rho$, and thus for any $\mathcal{L}_t$ with non-zero dissipative part there is an eigenvalue of $V_{s,t}(\phi^+_n)$ that is zero for $s=t$ but has a non-zero time derivative.
%Thus for every 

%there exist a rank one $\rho$ such that $P_0\mathcal{L}_tP_0$ is nonzero unless the lower off diagonal $1\times n$ block of $G_k$ is zero in every basis.

%Let $U_k$

\section{Proof of Proposition \ref{prop3}}\label{critic2}

Let $M(\bar{a})$ be a correlation measure that is an analytic function of the coordinates $a_i$ in a point $\bar{a}$ corresponding to a product state $\rho$ and let $V_{s,t}$ be a continuously differentiable intermediate map for all $s,t$.
Consider a family of states $\rho_{\epsilon}=\rho+\epsilon \chi$ where $\chi$ is Hermitian. The Taylor expansion of $M(\rho_{\epsilon})$ in $\epsilon$ around $\epsilon=0$ is

\begin{eqnarray}
M(\rho_{\epsilon})=\frac{\partial M(\rho_{\epsilon})}{\partial \epsilon}\Bigg{|}_{\epsilon=0}\epsilon+\frac{\partial^2 M(\rho_{\epsilon})}{2\partial \epsilon^2}\Bigg{|}_{\epsilon=0}\epsilon^2+\dots,\nonumber\\
\end{eqnarray}
where we have used that $M(\rho_\epsilon)|_{\epsilon=0}=0$. Since $M\geq 0$ on $S(\mathcal{H}_A\otimes \mathcal{H}_S)$ it follows that the first order term of the expansion must be zero if $\rho\in \Int[ S(\mathcal{H}_A\otimes \mathcal{H}_S)]$. Otherwise there would be a sufficiently small $\epsilon$ for which both $\rho_{\epsilon},\rho_{-\epsilon}\in \Int[ S(\mathcal{H}_A\otimes \mathcal{H}_S)]$ and  either $M(\rho_{\epsilon})$ or $M(\rho_{-\epsilon})$ was negative.

Assume that $\rho\in \Int[ S(\mathcal{H}_A\otimes \mathcal{H}_S)]$ and consider the Taylor expansion of $M[\mathcal{I}\otimes V_{s,t}(\rho_{\epsilon})]$ in $\epsilon$ around $\epsilon=0$

\begin{eqnarray}
M[\mathcal{I}\otimes V_{s,t}(\rho_{\epsilon})]=\frac{\partial M[\mathcal{I}\otimes V_{s,t}(\rho_{\epsilon})]}{\partial \epsilon}\Bigg{|}_{\epsilon=0}\epsilon\nonumber\\
+\frac{\partial^2 M[\mathcal{I}\otimes V_{s,t}(\rho_{\epsilon})]}{2\partial \epsilon^2}\Bigg{|}_{\epsilon=0}\epsilon^2+\dots\nonumber,\\
\end{eqnarray}
where we have used that $M(\mathcal{I}\otimes V_{s,t}(\rho_{\epsilon}))|_{\epsilon=0}=0$. Since $M\geq 0$ it follows again that the first order term of the expansion must be zero.
Thus $\frac{\partial M[(\mathcal{I}\otimes V_{s,t}(\rho_{\epsilon})]}{\partial \epsilon}\big{|}_{\epsilon=0}=0$ for all $s$. 

Next, consider the Taylor expansion of the derivative $\frac{d}{ds}M[\mathcal{I}\otimes V_{s,t}(\rho_{\epsilon})]|_{s=t}$ in $\epsilon$

\begin{eqnarray}
\frac{d}{ds}M[\mathcal{I}\otimes V_{s,t}(\rho_{\epsilon})]\Bigg{|}_{s=t}=\epsilon\frac{\partial}{\partial \epsilon}\frac{d}{ds}M[\mathcal{I}\otimes V_{s,t}(\rho_{\epsilon})]\Bigg{|}_{s=t}\Bigg{|}_{\epsilon=0}\nonumber\\
+\epsilon^2\frac{\partial^2 }{2\partial \epsilon^2}\frac{d}{ds}M[\mathcal{I}\otimes V_{s,t}(\rho_{\epsilon})]\Bigg{|}_{s=t}\Bigg{|}_{\epsilon=0}+\dots,\nonumber\\
\end{eqnarray}
where we have used that $\frac{d}{ds}M(\mathcal{I}\otimes V_{s,t}(\rho_{\epsilon}))\big{|}_{\epsilon=0}=0$.
Due to the continuous differentiability of $V_{s,t}$ it follows that $\frac{\partial}{\partial \epsilon}\frac{d}{ds}M[(\mathcal{I}\otimes V_{s,t}(\rho_{\epsilon})]$ is continuous, and therefore %$\frac{d}{ds}M(\rho_0)|_{s=t}=\sum_i\frac{\partial M (\bar{a},s)}{\partial a_i(s)}\frac{d a_i(s)}{ds}\Big|_{s=t}=\sum_{ij}\frac{\partial M (\bar{a})}{\partial a_i} a_{j}\frac{d V_{ji}(s,t)}{ds}\Big|_{s=t}$, \cite{rudin}
$\frac{\partial}{\partial \epsilon}\frac{d}{ds}M[(\mathcal{I}\otimes V_{s,t}(\rho_{\epsilon})]=\frac{d}{ds}\frac{\partial}{\partial \epsilon}M[(\mathcal{I}\otimes V_{s,t}(\rho_{\epsilon})]$ \cite{rudin}. It follows that the first order term in the Taylor expansion is zero.
Since this holds for every $\chi$ it follows that every product state in $\Int[ S(\mathcal{H}_A\otimes \mathcal{H}_S)]$ is a critical point of $M$.

\section{Proof of Proposition \ref{prop4}}\label{critic3}

We begin by considering the following two propositions.

\begin{proposition}\label{proood}
Let $\{\Lambda_{t}\}_t:B(\mathcal{H}_S)\to B(\mathcal{H}_S)$ be a family of CP qubit maps.
If the set of stationary states in $S(\mathcal{H}_S)$ is of non-zero dimension, the $\Lambda_{t}$ are unital. 
\end{proposition}
\begin{proof}

Assume that $\Lambda_{t}(\mathbbm{1})=\mathbbm{1}+\theta$ and $\Lambda_{t}(\rho_1)=\rho_1$, $\Lambda_{t}(\rho_2)=\rho_2$ where $\rho_1\neq\rho_2$. It follows that
$\Lambda_{t}(\mathbbm{1}+x(\rho_1-\rho_2))=\mathbbm{1}+x(\rho_1-\rho_2)+\theta$. Note that $x$ can be chosen such that $\mathbbm{1}\pm x(\rho_1-\rho_2)$ are rank one. Since these rank one qubit states are antipodal points on the Bloch ball $S(\mathcal{H}_S)$ it follows that unless $\theta=0$ at least one of 
$\mathbbm{1}+ x(\rho_1-\rho_2)+\theta$ and $\mathbbm{1}-x(\rho_1-\rho_2)+\theta$ is not positive semidefinite.
Thus if the set of stationary states has dimension greater than zero, it follows that $\Lambda_{t}$ is unital.
\end{proof}

\begin{proposition}\label{propend}
Let $\{\Lambda_{t}\}_t:B(\mathcal{H}_S)\to B(\mathcal{H}_S)$ be a family of CP qubit maps.
Then the set of stationary states of $\{\Lambda_{t}\}_t$ in $S(\mathcal{H}_S)$ has a dimension different from 2. 
\end{proposition}
\begin{proof}
Assume that the dimension of the set of stationary states in  $S(\mathcal{H}_S)$ is 2. From Prop. \ref{proood} follows that $\Lambda_{t}$ is unital. Without loss of generality we can assume that $\Lambda_{t}(\sigma_z)=\sigma_z$ and $\Lambda_{t}(\sigma_y)=\sigma_y$ and $\Lambda_{t}(\sigma_x)=a\sigma_x+b\sigma_y+c\sigma_z$. The Choi matrix of $\Lambda_{t}$ has eigenvalues $\pm\sqrt{1-2a+a^2+b^2+c^2}$ and $2\pm \sqrt{1+2a+a^2+b^2+c^2}$. Therefore $\Lambda_{t}$ is CP if and only if $a=1$ and $b=c=0$, i.e., if an only if $\Lambda_{t}=\mathcal{I}$. In this case the the set of stationary states has dimension 3, contradicting the assumption.
%Thus $\Lambda_{t}$ is CP only if it is the identity map.
\end{proof}

Next, consider a family of continuously differentiable dynamical qubit maps $\{\Lambda_{t}\}_t$ and a correlation measure $M$. If the set of stationary states in $S(\mathcal{H}_S)$ is non-empty its dimension is either zero or non-zero. If the dimension is zero, the set of stationary states in $S(\mathcal{H}_A\otimes \mathcal{H}_S)$ is a set of product states. Then it follows from Prop. \ref{prop3} that such a stationary state is a critical point if $M$ is analytical at the state and the state is in the interior of $S(\mathcal{H}_A\otimes \mathcal{H}_S)$.

If the dimension is 3 all states are stationary points and $\frac{d M}{dt}=0$ on all of $S(\mathcal{H}_A\otimes \mathcal{H}_S)$. Thus all states are critical points. By Prop. \ref{propend} the dimension is never 2.
The remaining case is a one-dimensional set of stationary states. Without loss of generality we can express any state in this set as $\rho_{AS}=\rho_A\otimes{\mathbbm{1}}+\chi_A\otimes{\sigma_z}$ for some $\rho_A,\chi_A\in B(\mathcal{H}_A)$.

Now, assume that $\rho_{AS}$ is in the interior of $S(\mathcal{H}_A\otimes \mathcal{H}_S)$ and that the correlation measure $M$ is an analytic function at $\rho_{AS}$. Then
consider the family of states $\rho_{\epsilon}=\rho_{AS}+\epsilon\chi_A\otimes\chi_S$ parametrized by $\epsilon$, where $\chi_A,\chi_S$ are Hermitian and $\Tr(\chi_A\otimes\chi_S)=0$. 
If $\chi_S=\sigma_z$ or $\chi_S=\mathbbm{1}$ it follows that $\rho_{\epsilon}$ is also a stationary state and thus $\frac{\partial}{\partial \epsilon}\frac{d}{dt}M(\rho_{\epsilon},t)\big{|}_{\epsilon=0}=0$.
If $\chi_S=\sigma_x$ or $\chi_S=\sigma_y$ there exist a local unitary operation, $I\otimes\sigma_z$, that commutes with $\rho_{AS}$ but anticommutes with  $\chi_A\otimes\chi_S$.
Since $M(\bar{a},t)$ is invariant under local unitary operations it follows that $M(\rho_{\epsilon})=M(\rho_{-\epsilon})$. Thus  $M(\rho_{\epsilon})$ is an even analytic function in $\epsilon$ and it follows that $\frac{\partial}{\partial \epsilon}M(\rho_{\epsilon},t)\big{|}_{\epsilon=0}=0$.
Since the $\sigma_x,\sigma_y,\sigma_z$ and $\mathbbm{1}$ span $B(\mathcal{H}_S)$ % and the $X^A_{kl},Y^A_{kl},Z^A_{k}$ and $I_A$ span $B(\mathcal{H}_A)$ 
we can conclude that $\frac{\partial}{\partial \epsilon}M(\rho_{\epsilon},t)\big{|}_{\epsilon=0}=0$ for every $\chi_A\otimes\chi_S$.
Moreover, $\frac{\partial}{\partial \epsilon}M(\rho_{\epsilon},t)\big{|}_{\epsilon=0}=0$ holds for any $t$. Therefore we can conclude that $\frac{d}{dt}\frac{\partial}{\partial \epsilon}M(\rho_{\epsilon},t)\big{|}_{\epsilon=0}=0$.
By the analyticity of $M$ and the continuous differentiability of $\Lambda_t$ it follows that 
$\frac{d}{dt}M$, $\frac{d}{d\epsilon}M$ and $\frac{d}{dt}\frac{d}{d\epsilon}M$ exist and that
$\frac{d}{dt}\frac{d}{d\epsilon}M$ is continuous as a function of $\epsilon$ and $t$. Therefore it follows that $\frac{d}{d\epsilon}\frac{d}{dt}M$ exist and $\frac{d}{d\epsilon}\frac{d}{dt}M=\frac{d}{dt}\frac{d}{d\epsilon}M$ \cite{rudin}.

We can conclude that all first derivatives of $\frac{d}{dt}M(\bar{a},t)$ with respect to $\bar{a}$ equal zero for states in the interior of $S(\mathcal{H}_A\otimes \mathcal{H}_S)$ that are stationary under a continuously differentiable $\{\Lambda_t\}_t$.

\section{Continuity of $C(\rho_{AB})$}\label{conti}

{We here show that the measures $C^{(n)}(\rho_{AB})$ for any $n$ as well as $C(\rho_{AB})$ are continuous functions on the set of states. Let $\Tr(\sum_i P_{A,i}\otimes P_{B,i} \rho_{AB})$ be an optimal solution to the optimization $\max_{\{P_{A,i}\}_{i=1}^n\in \Pi_n(\rho_{A})}P_g[\mathcal{E}(\rho_{AB},\{P_{A,i}\}_{i=1}^n)]$ (or alternatively $\max_{\{P_{A,i}\}_{i}\in \Pi(\rho_{A})}P_g[\mathcal{E}(\rho_{AB},\{P_{A,i}\}_{i})]$).}
Let $S_X$ be the set of traceless Hermitian operators $X$ satisfying $\Tr(XX^{\dagger})=1$ and consider a state $\tilde{\rho}_{AB}=\rho_{AB}+\epsilon X$ where $X\in S_X$. {Let
$\Tr(\sum_i \tilde{P}_{A,i}\otimes \tilde{P}_{B,i} \tilde{\rho}_{AB})=\Tr(\sum_i \tilde{P}_{A,i}\otimes \tilde{P}_{B,i} \rho_{AB})+\epsilon\Tr(\sum_i \tilde{P}_{A,i}\otimes \tilde{P}_{B,i} X)$ be an optimal solution to the optimization $\max_{\{P_{A,i}\}_{i=1}^n\in \Pi_A(\tilde{\rho}_{A})}P_g[\mathcal{E}(\tilde{\rho}_{AB},\{P_{A,i}\}_i)]$ (or alternatively $\max_{\{P_{A,i}\}_{i}\in \Pi(\rho_{A})}P_g[\mathcal{E}(\tilde{\rho_{AB}},\{P_{A,i}\}_{i})]$)}. Then we can construct the following upper and lower bounds 
\begin{eqnarray}\label{contin}
\Tr\left(\sum_i P_{A,i}\otimes P_{B,i} \rho_{AB}\right)-\epsilon \max_{X\in S_X,P_{A,i},P_{B,i} }\Bigg|\Tr\left(\sum_i P_{A,i}\otimes P_{B,i}  X\right)\Bigg|\nonumber\\
\leq\Tr\left(\sum_i P_{A,i}\otimes P_{B,i}  \rho_{AB}\right)+\epsilon\Tr\left(\sum_i P_{A,i}\otimes P_{B,i}  X\right)\nonumber\\\leq\Tr\left(\sum_i \tilde{P}_{A,i}\otimes \tilde{P}_{B,i}  \rho_{AB}\right)+\epsilon\Tr\left(\sum_i \tilde{P}_{A,i}\otimes \tilde{P}_{B,i}  X\right)\nonumber\\
\leq \Tr\left(\sum_i P_{A,i}\otimes P_{B,i}  \rho_{AB}\right)+\epsilon\Tr\left(\sum_i \tilde{P}_{A,i}\otimes \tilde{P}_{B,i}  X\right)\nonumber\\
\leq \Tr\left(\sum_i P_{A,i}\otimes P_{B,i} \rho_{AB}\right)+\epsilon \max_{X\in S_X,P_{A,i},P_{B,i}}\Bigg|\Tr\left(\sum_i P_{A,i}\otimes P_{B,i}  X\right)\Bigg|,\nonumber\\
\end{eqnarray}
{where we have used that $\Tr(\sum_i \tilde{P}_{A,i}\otimes \tilde{P}_{B,i} \rho_{AB})\leq\Tr(\sum_i P_{A,i}\otimes P_{B,i} \rho_{AB})$ and $\Tr(\sum_i P_{A,i}\otimes P_{B,i} \tilde{\rho}_{AB})\leq \Tr(\sum_i \tilde{P}_{A,i}\otimes \tilde{P}_{B,i} \tilde{\rho}_{AB})$.
Since any eigenvalues of an $X\in S_X$ is smaller or equal to one, the lower and upper bounds in Eq. (\ref{contin}) converge to $\Tr(\sum_i P_{A,i}\otimes P_{B,i} \rho_{AB})$ as $\epsilon$ converges to zero. An analogous argument can be made for the case where the roles of $A$ and $B$ are interchanged. Thus $C^{(n)}(\tilde{\rho}_{AB})$ converges to  $C^{(n)}(\rho_{AB})$ as $\tilde{\rho}_{AB}$ converges to $\rho_{AB}$ for any $\rho_{AB}$. Likewise $C(\tilde{\rho}_{AB})$ converges to  $C(\rho_{AB})$ as $\tilde{\rho}_{AB}$ converges to $\rho_{AB}$ for any $\rho_{AB}$.}

\section{Sufficient number of POVM elements}\label{bounded}

We can express the correlation measure $C (\rho_{AB})$ as
\begin{equation}\label{Capp}
C (\rho_{AB}) = \max_{\substack{\{ P_{A,i}^{(n)}\}_{i}\in \Pi_A(\rho_{AB}) \\\text{or}\\ \{ P_{B,i}^{(n)}\}_{i}\in \Pi_B(\rho_{AB})}}\max_{n:n\geq 2}\max_{\{ P_{A,i}^{(n)}\}_{i},\{ P_{B,i}^{(n)}\}_{i}}  \sum_i^{n}\Tr\left( \rho_{AB}  P^{(n)}_{A,i}\otimes{P^{(n)}_{B,i}} \right)  - \frac{1}{2},
\end{equation}
where the first maximization is over constraining either $\{ P_{A,i}^{(n)}\}_{i}$ or $\{ P_{B,i}^{(n)}\}_{i}$ to be a ME-POVM and $\max_{n:n\geq 2}$ denotes maximization over the number of POVM-elements running over all integers $n\geq 2$. 

While the number of POVM elements in the maximization is not constrained, here we show that it is sufficient to maximize over a finite range of integers.
We give a proof that it is sufficient to consider $n=2$ and odd $n$ in the range $2< n<\bar{n}$ where $\bar{n}$ is a function of $d_A=\dim(\mathcal{H}_A)$ and $d_B=\dim(\mathcal{H}_B)$ that satisfies $\bar{n}\leq \max(d_A,d_B)$. It is unclear what the necessary conditions for maximization over $n$ are.

Moreover, for classical-quantum states it is sufficient to consider maximization for $n=2$ and odd $n$ in the range $2< n\leq d_A$. For quantum-classical states it is sufficient to consider maximization for $n=2$ and odd $n$ in the range $2< n\leq d_B$. 
In particular this implies that $C_B(\rho_{AB}^{(\tau)}) = C^{(2)}_B (\rho_{AB}^{(\tau)})$ and that $C_A({\rho}_{AB}^{(\tau)}) = C_A^{(2)} ({\rho}_{AB}^{(\tau)}) $.

To find these sufficient conditions we introduce the following notation.
Denote by $C^{(n)} (\rho_{AB})$ the correlation measure where the maximization is limited to $n$-outcome POVMs
\begin{eqnarray}
C^{(n)} (\rho_{AB})=\max_{\substack{\{ P_{A,i}^{(n)}\}_{i}\in \Pi_A(\rho_{AB}) \\\text{or}\\ \{ P_{B,i}^{(n)}\}_{i}\in \Pi_B(\rho_{AB})}}\max_{\{ P_{A,i}^{(n)}\}_{i},\{ P_{B,i}^{(n)}\}_{i}}  \sum_i^{n}\Tr\left( \rho_{AB}  P^{(n)}_{A,i}\otimes{P^{(n)}_{B,i}} \right)  - \frac{1}{2}.\nonumber\\
\end{eqnarray}
Let $\{\overline{P}_{A,i}^{(n)}\}_{i}$ and $\{\overline{P}_{B,i}^{(n)}\}_{i}$ be two POVMs that maximize $C^{(n)}(\rho_{AB})$ where either $\{\overline{P}_{A,i}^{(n)}\}_{i}$ or $\{\overline{P}_{B,i}^{(n)}\}_{i}$ is a ME-POVM
\begin{eqnarray}\label{Pgneven}
C^{(n)} (\rho_{AB}) =\tr{{\rho}_{AB} \left( \sum_{i=1}^n \overline{P}_{A,i}^{(n)} \otimes \overline {P}_{B,i}^{(n)}  \right)   }  - \frac{1}{2} .
\end{eqnarray}

\subsection{The general case}

If $n$ is even we consider the following 2-output POVM:
\begin{eqnarray}\label{2outputeven}
P^{(2)}_{A,1}&=& \sum_{i\in E_1} \overline{P}_{A,i}^{(n)} \nonumber\\
P^{(2)}_{A,2} &=&\sum_{i\in E_2} \overline{P}_{A,i}^{(n)},
\end{eqnarray}
where $E_1$ and $E_2$ are any two sets of $n/2$ indices such that $E_1 \cup E_2 = \{ 1,2, \dots , n\}$. This structure guarantees that Eq. (\ref{2outputeven}) is a 2-output ME-POVM for $ \rho_{AB}$ if  $\{\overline{P}_{A,i}^{(n)}\}_{i}$ is. Similarly we define the 2 element POVM $\{P^{(2)}_{B,i}\}_i$ as
\begin{eqnarray}\label{PAiPBi}
P^{(2)}_{B,1} &=& \sum_{i\in E_1} \overline {P}_{B,i}^{(n)}\nonumber\\
P^{(2)}_{B,2} &=& \sum_{i\in E_2} \overline{P}_{B,i}^{(n)}.
\end{eqnarray}

We compare Eq. (\ref{Pgneven}) with the guessing probability of the output ensemble that we obtain applying Eq. (\ref{2outputeven}) on $\rho_{AB}$:
\begin{eqnarray}\label{Pg2x} 
C^{(2)} (\rho_{AB})&=&\max_{\substack{\{ P_{A,i}^{(n)}\}_{i}\in \Pi_A(\rho_{AB}) \\\text{or}\\ \{ P_{B,i}^{(n)}\}_{i}\in \Pi_B(\rho_{AB})}}\max_{\{P_{A,i}\}_{i},\{P_{B,i}\}_{i} }  \sum_{i=1}^2\tr{{\rho}_{AB}    P^{}_{A,i} \otimes P_{B,i}   } - \frac{1}{2} \nonumber\\
&\geq& \tr{{\rho}_{AB}  \left(  \sum_{i=1}^2  P^{(2)}_{A,i} \otimes {P}^{(2)}_{B,i}  \right)  }  - \frac{1}{2}\nonumber\\
&=& \tr{  {\rho}_{AB} \left( \sum_{i=1}^n \overline P_{A,i}^{(n)} \otimes  \overline{P}_{B,i}^{(n)} + P_{AB}^{mix} \right) } - \frac{1}{2} \nonumber\\
&=& C^{(n)} (\rho_{AB}) + \tr{ {\rho}_{AB}   P_{AB}^{mix}   }  \geq C^{(n)} (\rho_{AB}) \, ,
\end{eqnarray}
where $P_{AB}^{mix}$ is a sum of mixed terms of the form $\overline  P_{A,i}^{(n)} \otimes  \overline P_{B,j}^{(n)}$ with $i\neq j$, and provides a non-negative contribution. 
Thus for even $n$ $C^{(2)} (\rho_{AB})\geq C^{(n)} (\rho_{AB})$ for any $n\geq 2$.

If $n$ is odd, we define a 2-output POVM:
\begin{eqnarray}\label{PBodd}
P^{(2)}_{A,k} =\frac{1}{2}\overline P_{A,x}^{(n)}+ \sum_{i\in O_k^x}\overline P_{A,i}^{(n)} \hspace{0.75cm} (k=1,2) 
\end{eqnarray}
\begin{eqnarray}\label{PAodd}
P^{(2)}_{B,k} = \frac{1}{2} \overline P_{B,x}^{(n)}+ \sum_{i\in O_k^x} \overline P_{B,i}^{(n)}   \hspace{0.75cm} (k=1,2) 
\end{eqnarray}
where $O_1^x$ and $O_2^x$ are any two sets of $(n-1)/2$ indices such that $O_1^x \cup O_2^x = \{ 1,2, \dots , n\} \setminus \! x $. There are thus $n$ different ways to choose $x$.
This structure guarantees that $\{P^{(2)}_{A,k}\}$ is a 2-output ME-POVM for $ \rho_{AB}$ if  $\{\overline{P}_{A,i}^{(n)}\}_{i}$ is and likewise for $B$.

We can relate  $C^{(2)} (\rho_{AB})$ and ${C^{(n)} (\rho_{AB})}$ as
 
\begin{widetext}
\begin{eqnarray}\label{biggest}
C^{(2)} (\rho_{AB}) &\geq&  \tr{ {\rho}_{AB}  \left(  \left(\sum_{i\in O_1} \overline P_{A,i}^{(n)}+\frac{1}{2} \overline P_{A,x}^{(n)} \right)\otimes\left( \sum_{i\in O_1}\overline P_{B,i}^{(n)} + \frac{1}{2}\overline P_{B,x}^{(n)}\right)+ \left(\sum_{i\in O_2} \overline P_{A,i}^{(n)}+\frac{1}{2} \overline P_{A,x}^{(n)} \right)\otimes\left( \sum_{i\in O_2}\overline P_{B,i}^{(n)} + \frac{1}{2}\overline P_{B,x}^{(n)}\right)\right)   }-\frac{1}{2}\nonumber\\
&=&\tr{ {\rho}_{AB}  \left(  \sum_{i\neq x} \overline P_{A,i}^{(n)} \otimes \overline P_{B,i}^{(n)} \,+ \frac{1}{2} \overline P_{A,x}^{(n)} \otimes\overline P_{B,x}^{(n)}  + \frac{1}{2} \left(\sum_{i\neq x} \overline P_{A,i}^{(n)}  \right) \otimes\overline P_{B,x}^{(n)}  + \frac{1}{2}\overline P_{A,x}^{(n)}\otimes\left(\sum_{i\neq x} \overline P_{B,i}^{(n)} \right)   + P_{AB}^{mix} \right)   } -\frac{1}{2} \nonumber\\
&=&\tr{ {\rho}_{AB}   \left(  \sum_{i =1}^n\overline P_{A,i}^{(n)} \otimes\overline  P_{B,i}^{(n)} \,- \frac{3}{2} \overline P_{A,x}^{(n)} \otimes\overline P_{B,x}^{(n)}  + \frac{1}{2} \mathbbm{1}_A \otimes \overline P_{B,x}^{(n)}  + \frac{1}{2}\overline P_{A,x}^{(n)}\otimes \mathbbm{1}_B   +P_{AB}^{mix}\right)   } -\frac{1}{2} \nonumber\\
&=& C^{(n)} (\rho_{AB})  +  \tr{ \rho_{AB}  \left( - \frac{3}{2} \overline P_{A,x}^{(n)} \otimes \overline P_{B,x}^{(n)}  + \frac{1}{2} \mathbbm{1}_A \otimes \overline P_{B,x}^{(n)}  + \frac{1}{2}\overline P_{A,x}^{(n)}\otimes \mathbbm{1}_B   +P_{AB}^{mix}\right) } .
\end{eqnarray}
\end{widetext}
A sufficient condition for $C^{(2)} (\rho_{AB})\geq{C^{(n)} (\rho_{AB})}$ is that $\theta_x\equiv- {3} \overline P_{A,x}^{(n)} \otimes \overline P_{B,x}^{(n)}  +  \mathbbm{1}_A \otimes \overline P_{B,x}^{(n)}  +  \overline P_{A,x}^{(n)}\otimes \mathbbm{1}_B  $ is positive semidefinite. Let $\lambda_{Ak}$ and $\lambda_{Bk}$ be the eigenvalues of $P_{A,x}^{(n)}$ and $ P_{B,x}^{(n)}$, respectively. Then $\theta_x$ is positive-semidefinite if $-3\lambda_{Ak} \lambda_{Bl}+\lambda_{Ak}+\lambda_{Bl}\geq 0$ for every $k,l$. If $\lambda_{Ak}\leq 1/2$ or $\lambda_{Bl}\leq 1/2$ for every $k,l$ this inequality holds. If $\lambda_{Ak}\geq 1/2$ or $\lambda_{Bl}\geq 1/2$ the inequality holds if $\lambda_{Ak}\leq z$ and $\lambda_{Bl}\leq z/(3z-1)$ for some $1/2\leq z\leq 1$ and every $k,l$. 
A sufficient condition for the existence of a pair $\overline P_{A,i}^{(n)}$ and $\overline P_{B,j}^{(n)}$ satisfying this is $n\geq d_B(3z-1)/z$ and $n\geq d_A/z$ since $\Tr(\sum_i\overline P_{B,i}^{(n)})=d_B $ and $\Tr(\sum_i\overline P_{A,i}^{(n)})=d_A $.

Thus, a sufficient condition for $C^{(2)} (\rho_{AB})\geq{C^{(n)} (\rho_{AB})}$ is that $n\geq \bar{n}$ where $\bar{n}$ 

\begin{eqnarray}
\bar{n}\equiv \min_{z\in[\frac{1}{2},1]}\left[\max\left(\frac{d_A}{z},\frac{d_B(3z-1)}{z}\right)\right]=
\begin{cases}
    d_B       &  \text{if } \frac{d_A}{d_B}\leq \frac{1}{2}\\
    \frac{3d_Ad_B}{d_A+d_B}  &  \text{if } \frac{1}{2}<\frac{d_A}{d_B}< 2 \\
    d_A       &  \text{if } 2\leq\frac{d_A}{d_B}
  \end{cases}.\nonumber\\
\end{eqnarray}

In conclusion, for general states it is sufficient to consider maximization for $n=2$ and odd $n$ in the range $2< n<\bar{n}$.

\subsection{Classical-quantum and Quantum-classical states}

For the case of classical-quantum or quantum-classical states we can derive a different bound on $n$. To see this we derive an inequality similar to thart of Eq. (\ref{biggest}) 

\begin{widetext}
$$
P^{(2)}_g \geq  \tr{ {\rho}_{AB}  \left(  \sum_{i\neq x} \overline P_{A,i}^{(n)} \otimes  \overline P_{B,i}^{(n)} \,+ \frac{1}{2} \overline P_{A,x}^{(n)} \otimes\overline P_{B,x}^{(n)}  + \frac{1}{2} \left(\sum_{i\neq x} \overline P_{A,i}^{(n)}  \right) \otimes\overline P_{B,x}^{(n)}  + P_{AB}^{mix} \right)   } 
$$
$$
\geq \tr{{\rho}_{AB}   \left(  \sum_{i =1}^n\overline P_{A,i}^{(n)} \otimes\overline  P_{B,i}^{(n)} \,- \frac{1}{2} \overline P_{A,x}^{(n)} \otimes\overline P_{B,x}^{(n)}  + \frac{1}{2} \left(\sum_{i\neq x}\overline P_{A,i}^{(n)}  \right) \otimes\overline P_{B,x}^{(n)}   \right)   } 
$$
$$ 
= P_g^{(n)}  +  \tr{ \rho_{AB}  \left( \frac{ - \overline P_{A,x}^{(n)} }{2} \otimes \overline P_{B,x}^{(n)}  +  \frac{ \sum_{i\neq x}\overline P_{A,i}^{(n)}  }{2}\otimes\overline P_{B,x}^{(n)}   \right) }
= P_g^{(n)}  + \tr{ \rho_{AB}  \frac{ \mathbbm{1}_A - 2 \overline P_{A,x}^{(n)} }{2}  \otimes \overline P_{B,x}^{(n)} }  .
$$
\end{widetext}

If we restrict the set of states to the classical-quantum states
\begin{equation}
\rho_{AB}(t) = \sum_s p_s|s\rangle\langle s|_A \otimes \rho_{Bs} (t), ,
\end{equation}
it is clear that  $\tr{{\rho}_{AB}  \frac{ \mathbbm{1}_A - 2 \overline P_{A,x}^{(n)} }{2}  \otimes\overline P_{B,x}^{(n)} }$ is non-negative if
the diagonal elements of $(\mathbbm{1}_A- 2 \overline P_{A,x}^{(n)} )\otimes \overline P_{B,x}^{(n)}$ are non-negative, i.e., if the diagonal elements of $\overline P_{A,x}^{(n)}$ are all smaller or equal to $1/2$.   
An $x$ such that $\overline P_{A,x}^{(n)}$ satisfies this exists with certainty if $n>d_A$ since at most $d_A$ elements of $\{\overline{P}_{A,i}^{(n)}\}_{i}$ can have a diagonal element larger than $1/2$. 
In particular, for the probe state in Eq. (\ref{hatrho}) it is sufficient to consider $n=d_A=2$ regardless of $d_B$. Thus, $C({\rho}_{AB}^{(\tau)}) = C^{(2)} ({\rho}_{AB}^{(\tau)}) $.

If we restrict the set of states to the quantum-classical states
\begin{equation}
\rho_{AB}(t) = \sum_s  p_s\rho_{As} (t) \otimes|s\rangle\langle s|_B, ,
\end{equation}
a completely analogous argument gives the sufficient condition for the existence of a $\overline P_{B,x}^{(n)}$ with all diagonal elements smaller or equal to $1/2$. Such a $\overline P_{B,x}^{(n)}$ exists if $n>d_B$.

\section{$\Pi_A(\rho_{AB}) \neq \{ \mathbf{0}\}$ }\label{neverempty}
We  explicitly construct an element $\{ P_{A,i}\}_i$ of $\Pi_A(\rho_{AB})$ for an arbitrary state $\rho_{AB}$. The method that we use should convince the reader that there are innumerable other ways to construct a ME-POVM with any number of outputs.

By definition $\{ P_{A,i}\}_{i=1,\dots,n} \in\Pi_A(\rho_{AB})$   if the output ensemble $\mathcal{E}(\rho_{AB}, \{ P_{A,i}\}_i)=\{p_i, \rho_{B,i} \}_{i}$ is characterized by $p_i=1/n$. In general, we have that $p_{i} = \tr{ \rho_{AB} P_{A,i} } = \trA{ \rho_A P_{A,i} }  \, ,
$
where $\rho_A = \trB{\rho_{AB}}$. Using an orthogonal decomposition of $\rho_A$, we can always write it as: $\rho_A  = \sum_{i=1}^{d_A} \pi_i \ketbra{i}{i}_A$, where $\{\ket{i}_A\}_i$ is an orthonormal basis of $\mathcal{H}_A$. The condition $\sum_{i=1}^{d_A} \pi_i =1$ implies that there exist an $\overline{i}$, such that $S (\overline{i}) \equiv \sum_{i=1}^{\overline{i}} \pi_i > 1/2$ and $S(\overline{i}-1) \equiv \sum_{i=1}^{\overline{i}-1} \pi_i \leq 1/2$. We consider the following class of 2-output POVM that depends on a real parameter $\omega \in [0,1]$
$
P_{A,1} (\omega) = \sum_{i=1}^{\overline{i}-1} \ketbra{i}{i}_A + \omega \ketbra{\overline{i}}{\overline{i}}_A \, ,
$
$
P_{A,2}(\omega) = (1-\omega)\ketbra{\overline{i}}{\overline{i}}_A+ \sum_{i= \overline{i}+1}^{d_A}\ketbra{i}{i}_A \, .
$
We evaluate $p_{1}$ for a general value of $\omega$ and we obtain
$
p_{1}(\omega) =\sum_{i=1}^{\overline{i}-1} \pi_i + \omega \, \pi_{\overline{i} } = S(\overline{i}-1) + \omega \, \pi_{\overline{i} }  \, .
$
It is clear that, since $p_{1}(0) = S(\overline{i}-1)\leq1/2$ and $p_{1}(1) = S(\overline{i})>1/2$, the value $\omega =  \overline{\omega}\equiv (1/2 - S(\overline{i}-1) )/{\pi_{\overline{i}} }$, gives the uniform distribution  $p_{1,2}(\overline{\omega})=1/2$ and consequently $\{ P_{A,i}(\overline \omega)  \}_i \in \Pi_A(\rho_{AB})$, i.e a  ME-POVM for $\rho_{AB}$.

\section{Monotonic behaviour of $C$  {and $C^{(n)}$} under local operations}\label{monolocal}

Firstly, we prove that {$C_A$} is monotone under local operations of the form $\Lambda_A \otimes  \mathcal{I}_B$, and secondly we consider the case where the local operation is $\mathcal{I}_A \otimes  \Lambda_B$, where $\Lambda_A$ ($\Lambda_B$) is a CPTP map on $A$ ($B$) and $\mathcal{I}_A$ ($\mathcal{I}_B$) is the identity map on $A$ ($B$). The proof for {$C_A$} easily generalizes to {$C_B$} and {$C$}.
{Finally, we prove that the same monotonicity property holds for $C^{(n)}$ for any $n\geq 2$. We denote the set of ME-POVMs acting on $A$ for the state $\rho_{AB}$ by $\Pi_A(\rho_{AB})$ and similarly for $B$.}

In order to show the effect  of the application of a local operation of the form  $\Lambda_A \otimes  \mathcal{I}_B$ on $C_A(\rho_{AB})$, we look at $\Pi_A(\rho_{AB})$ in a different way. Each element of this collection is a ME-POVM for $\rho_{AB}$, i.e. they generate sets of { \it equiprobable ensembles of states} (EES) from $\rho_{AB}$. In fact 
\begin{equation}\label{CAapp}
C_A (\rho_{AB}) \equiv \max_{ \left\{ P_{A,i}\right\}_i  \in \Pi_A \left( \rho_{A} \right) }  P_g \left( \mathcal{E} \left( \rho_{AB} ,\left\{ P_{A,i} \right\}_i \right)  \right)  - \frac{1}{2} \, .
\end{equation}
 is a maximization over all the possible EES that we can generate from $\rho_{AB}$ with a measurement procedure on $A$.

The effect of the first local operation that we consider is: $\tilde{\rho}_{AB}=  \Lambda_A \otimes  \mathcal{I}_B \, ( \rho_{AB} ) = \sum_k \left( E_k \otimes \mathbbm{1}_B \right) \cdot \rho_{AB} \cdot \left( E_k \otimes \mathbbm{1}_B \right)^\dagger  \, ,$
where $\left\{ E_k \right\}_k$ is the set of the Kraus operators that defines $\Lambda_A$. What is the relation between $\Pi_A (\rho_{AB}) $ and $\Pi_A (\tilde{\rho}_{AB})$? Given an $n$-output ME-POVM for $\tilde \rho_{AB}$, i.e. $\{ P_{A,i}\}_i \in \Pi_A (\tilde{\rho}_{AB})$, the probabilities and the states of the output ensemble $\mathcal{E} \left( \tilde{\rho}_{AB}, \{ P_{A,i} \}_i\right)$ are $\tilde{p}_i =  \tr{ \tilde{\rho}_{AB} \cdot P_{A,i} }=1/n$ and $\tilde \rho_{B,i} = \trA{\tilde\rho_{AB} \cdot P_{A,i}}/\tilde{p}_i$. Now we look at the term
$$
\trA{\tilde{\rho}_{AB} \cdot P_{A,i}} = \tr{\Lambda_{A}\otimes \mathcal{I}_B\, ( \rho_{AB}) \cdot P_{A,i}} =
$$
$$
=\trA{\sum_k (E_k \otimes \mathbbm{1}_B)\cdot \rho_{AB} \cdot (E_k^\dagger \otimes \mathbbm{1}_B)\cdot P_{A,i} }=
$$
$$
=\trA{\rho_{AB} \sum_k (E_k^\dagger \otimes \mathbbm{1}_B) \cdot P_{A,i} \cdot (E_k \otimes \mathbbm{1}_B) } =
$$
$$
= \trA{\rho_{AB} \cdot \Lambda^*_{A} (P_{A,i} ) } = \trA{\rho_{AB} \cdot \tilde P_{A,i} } \, ,
$$
and we rewrite the probabilities and the output states as:
$\tilde p_i= \mbox{Tr}[ \rho_{AB} \cdot \tilde P_{A,i} ] = 1/n$ and $\rho_{B,i}=\mbox{Tr}_A[ \rho_{AB} \cdot \tilde P_{A,i} ] /\tilde p_i$. This ensemble is an EES. Next we show that:
$\{ \tilde{P}_{A,i} \}_i =\left\{ \Lambda^*_A \left( P_{A,i} \right) \right\}_i = \{ \sum_k E_k^\dagger\cdot P_{A,i} \cdot E_k \}_i \, ,
$
is a POVM. The elements of $\{ \tilde{P}_{A,i} \}_i$ sum up to the identity: 
$
\sum_i \tilde{P}_{A,i} =   \sum_{k,i} E_k^\dagger \, P_{A,i} \,  E_k = \sum_{k} E_k^\dagger \, \left( \sum_i P_{A,i} \right) \, E_k = \sum_{k} E_k^\dagger\,  E_k = \mathbbm{1}_B \, ,
$
and they are positive operators:
$
\tilde{P}_{A,i} = \sum_k E_k^\dagger\,  P_{A,i} \,  E_k = \sum_{k}  E_k^\dagger \,  M^\dagger_{A,i} \, M_{A,i} \,  E_k =  \tilde{M}_{A,i}^\dagger \tilde{M}_{A,i} \, ,
$
where the decomposition $P_{A,i} = M^\dagger_{A,i} M_{A,i}$ exists since $P_{A,i}$ is positive-semidefinite and  $\tilde{M}_{A,i} = \sum_k M_{A,i} \, E_k$. It follows that, $\{ \tilde{P}_{A,i} \}_i$ is a ME-POVM for $\rho_{AB}$, i.e. $\{ \tilde{P}_{A,i} \}_i\in \Pi_A(\rho_{AB})$.  Thus, for every ME-POVM $\{P_{A,i}\}_i \in \Pi_A(\tilde\rho_{AB})$ for $\tilde \rho_{AB}$, there is a ME-POVM $\{\tilde P_{A,i}\}_i \in \Pi_A(\rho_{AB})$ for $\rho_{AB}$, such that the output ensembles are identical: $\mathcal{E} (\tilde \rho_{AB}, \{P_{A,i}\}_i) = \mathcal{E} (\rho_{AB}, \{\tilde P_{A,i} \}_i )$.
Thus, any EES that can be generated from $\tilde{\rho}_{AB}$, is obtainable from $\rho_{AB}$ as well
\begin{equation}\label{EI}
\bigcup_{ \{P_{A,i} \}_i \in \Pi_A (\tilde{\rho}_{AB} ) }   \!\!\!\!\!\!\!\!  \mathcal{E} \left( \tilde{\rho}_{AB}, \, \{ P_{A,i} \}_i \right)\,  \subseteq \!\!\!\!\!\!\! \bigcup_{ \{P_{A,i} \}_i \in \Pi_A ({\rho}_{AB} ) }    \!\!\!\!\!\!\!\! \mathcal{E} \left( {\rho}_{AB}, \, \{ P_{A,i} \}_i \right) \, .
\end{equation}
Finally, because $C_A(\rho_{AB})$ could be thought as the maximum guessing probability of the EESs that can be generated from $\rho_{AB}$ (see Eq. (\ref{CAapp})), we conclude that 
\begin{equation}\label{monA2}
C_A \left( \rho_{AB} \right) \geq C_A \left(  \Lambda_A \otimes  \mathcal{I}_B \,  ( \rho_{AB}) \right) \, ,
\end{equation}
for any state $\rho_{AB}$ and CPTP map $\Lambda_A$.

{Fixing the number $n$ of outputs of the ME-POVMs considered in (\ref{CAapp}), Eq. (\ref{EI}) becomes:
\begin{equation}\label{EIn}
\bigcup_{ \{P_{A,i} \}_{i=1}^n \in \Pi_A (\tilde{\rho}_{AB} ) }   \!\!\!\!\!\!\!\!  \mathcal{E} \left( \tilde{\rho}_{AB}, \, \{ P_{A,i} \}_i \right)\,  \subseteq \!\!\!\!\!\!\! \bigcup_{  \{P_{A,i} \}_{i=1}^n \in \Pi_A ({\rho}_{AB} ) }    \!\!\!\!\!\!\!\! \mathcal{E} \left( {\rho}_{AB}, \, \{ P_{A,i} \}_i \right) \, .
\end{equation}
Therefore, it follows that:
\begin{equation}\label{monA2n}
C_A^{(n)} \left( \rho_{AB} \right) \geq C_A^{(n)} \left(  \Lambda_A \otimes  \mathcal{I}_B \,  ( \rho_{AB}) \right) \, ,
\end{equation}
for any integer $n\geq 2$, state $\rho_{AB}$ and  CPTP map $\Lambda_A$.
}

Next we show the property of monotonicity of $C_A(\rho_{AB})$  under  the action of local operations of the form $\mathcal{I}_A \otimes  \Lambda_B$. We find that the collection of the ME-POVMs for $\tilde{\rho}_{AB} =  \mathcal{I}_A \otimes  \Lambda_B \, ( {\rho}_{AB})$, i.e. $\Pi_A (\tilde{\rho}_{AB})$, coincides with $\Pi_A (\rho_{AB})$. 

In order to prove this, we apply a general POVM $\{P_{A,i}\}_i$ on both $\rho_{AB}$ and $\tilde \rho_{AB}$ and we show that the respective output ensembles are defined by the same probability distribution. We can write $p_i= \tr{\rho_{AB}\cdot P_{A,i}}$ and $\tilde p_i =\tr{ \mathcal{I}_A\otimes \Lambda_B \, ( \rho_{AB}) \cdot P_{A,i} } = \tr{\rho_{AB} \cdot P_{A,i} }$, where the last step uses the trace-preserving property of the superoperator $\mathcal{I}_A\otimes \Lambda_B$. Consequently, $p_i=1/n$ if and only if $\tilde p_i=1/n$ and $\{ P_{A,i}\}_i\in \Pi_A(\rho_{AB})$ if and only if $\{ P_{A,i}\}_i\in \Pi_A(\tilde \rho_{AB})$
\begin{equation}\label{uguale}
\Pi_A (\rho_{AB} ) = \Pi_A (\tilde{\rho}_{AB}) \, .
\end{equation}
Given a ME-POVM for both $\rho_{AB}$ and $\tilde{\rho}_{AB}$, we relate the output states
\begin{equation}\label{rhobicontratti}
\tilde{\rho}_{B,i} = \Lambda_B \cdot \trA{ \rho_{AB} P_{A,i}  }  / p_i = \Lambda_B (\rho_{B,i}) \, .
\end{equation}
From Eq. (\ref{rhobicontratti}) and the definition of the guessing probability, it follows that
\begin{equation}\label{monB0}
P_g\left(  \left\{ p_i, \, \rho_{B,i}  \right\}_i  \right) \geq P_g\left(  \left\{ p_i, \, \Lambda_B (\rho_{B,i})  \right\}_i  \right) \, ,
\end{equation}
and, considering Eq. (\ref{uguale}), Eq. (\ref{rhobicontratti}) and Eq. (\ref{monB0})
\begin{equation}\label{monB}
C_A \left( \rho_{AB} \right) \geq C_A \left(  \mathcal{I}_A \otimes \Lambda_B   \,  ( \rho_{AB}) \right) \, ,
\end{equation}
that is true for any state $\rho_{AB}$ and CPTP map $\Lambda_B$.

{From Eq. (\ref{uguale}) it follows the collection of the $n$-output ME-POVMs does not change if we apply a CPTP map $\Lambda_B$ on $\rho_{AB}$. Therefore, since Eq. (\ref{monB0}) is true for any number of outputs:
\begin{equation}\label{monB}
C_A^{(n)} \left( \rho_{AB} \right) \geq C_A^{(n)}  \left(  \mathcal{I}_A \otimes \Lambda_B   \,  ( \rho_{AB}) \right) \, ,
\end{equation}
for any integer $n\geq 2$, state $\rho_{AB}$ and  CPTP map $\Lambda_B$.
}

We underline that from this proof we automatically obtain the invariance under local unitary transformations of $C$ and $C^{(n)}$ for any $n\geq 2$.

\section{$C_A({\rho}_{AB}^{(\tau)}) \geq  C^{(2)}_B({\rho}_{AB}^{(\tau)})$ }\label{CB2CA}

In this section (where we omit the time dependence of $\rho_{AB}^{(\tau)}(t)$, $\rho_B'^{(\tau)}(t)$ and $\rho_B''^{(\tau)}(t)$) we show that $ C_A(\rho_{AB}^{(\tau)}) \geq C^{(2)}_B( \rho_{AB}^{(\tau)})$, where $C^{(2)}_B( \rho_{AB}^{(\tau)})$ is defined by
$$
C_B^{(2)} ({\rho}_{AB}^{(\tau)}) \equiv\!\!\!\!\!\!\!\! \max_{ \left\{P_{B,i}\right\}_i \in \Pi_B^{(2)} \left( \rho_{AB}^{(\tau)} \right) } \!\!\!\!P_g \left( \mathcal{E} \left( \rho_{AB}^{(\tau)} ,\left\{ P_{B,i} \right\}_i \right)  \right)-\frac{1}{2} \, ,
$$
where $\Pi_B^{(2)} ( \rho_{AB}^{(\tau)} )$ is the set of the 2-output ME-POVMs.
In Appendix  \ref{bounded} we showed that $C^{(2)}_B(\rho_{AB}^{(\tau)}) = C_B(\rho_{AB}^{(\tau)})$ and this completes the proof that $C_A({\rho}_{AB}^{(\tau)}) \geq  C_B({\rho}_{AB}^{(\tau)})$.

We  apply a general but fixed 2-output ME-POVM for ${\rho}_{AB}^{(\tau)}$, where now the measured system is  $B$: $\{P_{B,i}^{(2)}\}_i = \{ P_{B}, \, \overline{P}_B \} \in \Pi_B({\rho}_{AB}^{(\tau)})$, where $\overline{P}_B = \mathbbm{1}_B -P_{B}$. The output ensemble  $\mathcal{E} (\rho_{AB}^{(\tau)}, \{P_{B,i}^{(2)}\}_i)= \{ p_{A,i}, \, \rho_{A,i}\}_{i}$ is composed by an uniform distribution (by definition of ME-POVM) and states in the following form: 

\begin{equation}\label{p1}
p_{A,1}=   \frac{1}{2} \trB{\left( \rho_B'^{(\tau)} +\rho_B''^{(\tau)}\right) P_B  } =\frac{1}{2} \, ,
\end{equation}
\begin{equation}\label{p2}
p_{A,2} = \frac{1}{2} \trB{\left( \rho_B'^{(\tau)} +\rho_B''^{(\tau)}\right)\overline{P}_B  } =\frac{1}{2} \, ,
\end{equation}
\begin{equation}
\rho_{A,1}=   \ket{0}\bra{0}_A \trB{ \rho_B'^{(\tau)}  P_B  } +\ket{1}\bra{1}_A \trB{ \rho_B''^{(\tau)}  P_B  }  \, ,
\end{equation} 
\begin{equation}\label{rho2}
\rho_{A,2}=  \ket{0}\bra{0}_A \trB{ \rho_B'^{(\tau)}  \overline{P}_B  } +\ket{1}\bra{1}_A \trB{ \rho_B''^{(\tau)} \overline{P}_B }  \, .
\end{equation}
If we use Eq. (\ref{PgD}) to get $P_g( \mathcal{E} (\rho_{AB}^{(\tau)}, \{P_{B,i}^{(2)}\}_i))$, firstly we have to evaluate $|| \rho_{A,1} - \rho_{A,2} ||_1$. Hence, with Eqs. (\ref{p1})-(\ref{rho2}), we can write it as:
$$
||  \ketbra{0}{0}_A \trB{ \rho'^{(\tau)}_B  \Delta P_B} + \ketbra{1}{1}_A \!\trB{ \rho''^{(\tau)}_B \Delta P_B } ||_1 =
$$
$$  
= | \trB{ \rho'^{(\tau)}_B  \Delta P_B }| + | \trB{ \rho''^{(\tau)}_B \Delta P_B }| \, ,
$$
where $\Delta P_B = P_B - \overline{P}_B$. Hence:
$$
|| \rho_{A,1} - \rho_{A,2} ||_1 = \max_{\pm} |  \trB{ ( \rho'^{(\tau)}_B \pm \rho''^{(\tau)}_B)  \Delta P_B}|  \, .
$$
 Using Eq. (\ref{p1}) and Eq. (\ref{p2}) we see that
$
| \trB{ ( \rho'^{(\tau)}_B + \rho''^{(\tau)}_B) \Delta P_B }| = | \trB{ ( \rho'^{(\tau)}_B + \rho''^{(\tau)}_B)  P_B} - \trB{ ( \rho'^{(\tau)}_B + \rho''^{(\tau)}_B)  \overline{P}_B }| = 2| p_{A,1} - p_{A_2} | =0 \, .
$
Hence
$
|| \rho_{A,1} - \rho_{A,2} ||_1  =  | \, \trB{ ( \rho'^{(\tau)}_B - \rho''^{(\tau)}_B) (2P_B - \mathbbm{1}_B )}| = 2 |  \trB{ ( \rho'^{(\tau)}_B - \rho''^{(\tau)}_B) P_B}| \, ,
$
from which follows that $C_B^{(2)} ({\rho}_{AB}^{(\tau)} )$ is equal to
\begin{equation}\label{maximB}
\max_{ \{P_{B,i}^{(2)} \}_{i} \in \Pi_B( {\rho}_{AB}^{(\tau)} ) }  \frac{| \trB{ ( \rho'^{(\tau)}_B - \rho''^{(\tau)}_B)  P_B}| }{2}  .
\end{equation}

To compare with $C_A({\rho}_{AB}^{(\tau)})$, we write it as
$$
C_A({\rho}_{AB}^{(\tau)}) = P_g (  \{ \{ p_{A,1,2}=1/2 \}_i , \{ \rho'^{(\tau)}_B , \rho''^{(\tau)}_B \} \} ) -\frac{1}{2} =$$
$$=\max_{ \{P_{B,i} \}_{i} \in \Pi_B  } 	\frac{  \trB{ \rho_B'^{(\tau)}  P_B + \rho_B''^{(\tau)}  \overline{P}_B } }{2} -\frac{1}{2}= $$
$$ {= \max_{ \{P_{B,i} \}_{i} \in \Pi_B  }}    \frac{   \trB{(\rho'^{(\tau)}_B - \rho''^{(\tau)}_B) P_B } }{2}= $$
$$
= \max_{ \{P_{B,i} \}_{i} \in \Pi_B  }   \frac{ |\trB{(\rho'^{(\tau)}_B - \rho''^{(\tau)}_B) P_B }| }{2}  \, .
$$
 We have used the definition in Eq. (\ref{Pg}): $\Pi_B $ is the collection of all the POVMs that we can perform on $B$. It follows that the only difference between  $C_B^{(2)}(\rho_{AB}^{(\tau)}) $ and $C_A(\rho_{AB}^{(\tau)}) $ is in the maximization procedure: in the former we maximize only over the ME-POVMs $\Pi_B({\rho}_{AB}^{(\tau)})$, while in the latter we can pick any POVM: $C_A({\rho}_{AB}^{(\tau)}(t)) \geq C_B^{(2)}({\rho}_{AB}^{(\tau)}(t))$ follows as a natural consequence.

\section{Proof that $C_B(\rho_{AB}^{(\tau)}) = C^{(2)}_B (\rho_{AB}^{(\tau)})$}\label{2enough}

In this Appendix, in contrast to Appendix \ref{CB2CA}, we consider the action of any ME-POVM over $B$ for $\rho_{AB}^{(\tau)}$.
We want to show  that for each  ME-POVM $\{P_{B,i}^{(n)}\}_{i}$ that we can consider in $C_B(\rho_{AB}^{(\tau)})$, where $i$ runs from 1 to $n>2$, we can  always find at least one 2-output ME-POVM acting on $B$, i.e. $\{P_{B,1}, P_{B,2}  \} \in \Pi_B(\rho_{AB}^{(\tau)}) $, that provides an ensemble with a higher value of $P_g (\cdot) $. We recall that, if $\mathcal{E}=\{ p_i,\rho_i\}_i$ is a generic ensemble of $n$ states defined on $S(\mathcal{H})$, where $\mathcal{H}$ is a generic finite dimensional Hilbert space, the guessing probability of $\mathcal{E}$ is
\begin{equation}\label{Pg}
P_g (\mathcal{E})\equiv \max_{ \left\{ P_i \right\}_i }\sum_{i=1}^n  p_i \tr{ \rho_i \cdot P_i  } \, ,
\end{equation}
where the maximization is performed over the space of the $n$-output POVMs $\{P_i\}_i$ on $S(\mathcal{H})$.
 Starting from a general $n$-output ME-POVM $\{P_{B,i}^{(n)}\}_{i}$, we construct the corresponding 2-output ME-POVM $\{P_{B,1}, P_{B,2}  \} \in \Pi_B(\rho_{AB}^{(\tau)})$ that accomplishes this task.

For every given $n$-output ME-POVM $\{ P_{B,i}^{(n)} \}_i$ for ${\rho}_{AB}^{(\tau)}$, we can generate an equiprobable ensemble of states (EES) of the form $\mathcal{E}( {\rho}_{AB}^{(\tau)} , \, \{P_{B,i}^{(n)} \}_i) = \{\{p_i=1/n\},\{\rho_{A,i}\} \}_i$. The guessing probability of this ensemble, which we denote by  $P_g^{(n)} = P_g(\mathcal{E} (\rho_{AB}^{(\tau)},\{ P_{B,i}^{(n)} \}_i))$, is
\begin{equation}\label{Pgneven}
P_g^{(n)} =\tr{{\rho}_{AB}^{(\tau)}\cdot \left( \sum_{i=1}^n \overline{P}_{A,i}^{(n)} \otimes P_{B,i}^{(n)}  \right)   } \, ,
\end{equation}
where  $\{ \overline{P}_{A,i}^{(n)}\}_i$ is a POVM that provides the maximum in  Eq. (\ref{Pg}). If $n$ is even we consider the following 2-output POVM
\begin{equation}\label{2outputeven}
P^{(2)}_{B,1} = \sum_{i\in E_1} P_{B,i}^{(n)} \, , \,\,\,
P^{(2)}_{B,2} = \sum_{i\in E_2} P_{B,i}^{(n)} \, ,
\end{equation}
where $E_1$ and $E_2$ are any two sets of $n/2$ indices such that $E_1 \cup E_2 = \{ 1,2, \dots , n\}$. This structure guarantees that Eq. (\ref{2outputeven}) is a 2-output ME-POVM for $\rho_{AB}^{(\tau)}$. We compare Eq. (\ref{Pgneven}) with the guessing probability of the output ensemble that we obtain applying Eq. (\ref{2outputeven}) on $\rho_{AB}^{(\tau)}$
$$
P_g^{(2)}=\max_{\{P_{A,i}\}_{i=1,2} }  \tr{{\rho}_{AB}^{(\tau)} \cdot \left(  \sum_{i=1}^2 P_{A,i} \otimes P^{(2)}_{B,i}  \right)  } \geq $$
\begin{equation}\label{Pg2x} \geq \tr{{\rho}_{AB}^{(\tau)} \cdot \left(  \sum_{i=1}^2  P^{(2)}_{A,i} \otimes {P}^{(2)}_{B,i}  \right)  } \, ,
\end{equation}
where the POVM $\{ P^{(2)}_{A,i}\}_i $ is defined by 
\begin{equation}\label{PAiPBi}
P^{(2)}_{A,1} = \sum_{i\in E_1} \overline P_{A,i}^{(n)} \, , \, \,\,
P^{(2)}_{A,2} = \sum_{i\in E_2} \overline{P}_{A,i}^{(n)} \, .
\end{equation}

$$
P_g^{(2)}\geq \tr{ {\rho}_{AB}^{(\tau)} \cdot \left(P^{(2)}_{A,1}  \otimes P^{(2)}_{B,1} + P^{(2)}_{A,2}  \otimes P^{(2)}_{B,2}  \right) } 
= $$
$$
= \tr{  {\rho}_{AB}^{(\tau)} \cdot\left( \sum_{i=1}^n \overline P_{A,i}^{(n)} \otimes  P_{B,i}^{(n)} + P_{AB}^{mix} \right) } = $$
\begin{equation}\label{Pgeven}
= P_g^{(n)} + \tr{ {\rho}_{AB}^{(\tau)}  \cdot P_{AB}^{mix}   }  \geq P_g^{(n)}  \, ,
\end{equation}
where $P_{AB}^{mix}$ is a sum of mixed terms of the form $\overline  P_{A,i}^{(n)} \otimes  P_{B,j}^{(n)}$ with $i\neq j$, and it provides a non-negative contribution. 

On the other hand, if $n$ is odd, we define
\begin{equation}\label{PBodd}
P^{(2)}_{B,k} =\frac{1}{2} P_{B,x}^{(n)}+ \sum_{i\in O_k^x} P_{B,i}^{(n)} \hspace{0.75cm} (k=1,2) 
\end{equation}
\begin{equation}\label{PAodd}
P^{(2)}_{A,k} = \frac{1}{2} \overline P_{A,x}^{(n)}+ \sum_{i\in O_k^x} \overline P_{A,i}^{(n)}   \hspace{0.75cm} (k=1,2) 
\end{equation}
where $O_1^x$ and $O_2^x$ are any two sets of $(n-1)/2$ indices such that $O_1^x \cup O_2^x = \{ 1,2, \dots , n\} \setminus \! x $ (the value of $x$ will be fixed later). 
We consider again Eq. (\ref{Pg2x}), where $\{P^{(2)}_{B,i}\}_i$ is now given by Eq. (\ref{PBodd}) and and $P^{(2)}_{A,i}$ is now given by Eq. (\ref{PAodd}). Since $P^{(2)}_{A,i}$ is not necessarily a POVM that maximizes Eq. (\ref{Pg}) we have the following inequality for $P_g^{(2)}$
\begin{widetext}
$$
P^{(2)}_g \geq  \tr{ {\rho}_{AB}^{(\tau)}\cdot  \left(  \sum_{i\neq x} \overline P_{A,i}^{(n)} \otimes  P_{B,i}^{(n)} \,+ \frac{1}{2} \overline P_{A,x}^{(n)} \otimes P_{B,x}^{(n)}  + \frac{1}{2} \left(\sum_{i\neq x} \overline P_{A,i}^{(n)}  \right) \otimes P_{B,x}^{(n)}  + P_{AB}^{mix} \right)   } \geq 
$$
$$
=\tr{ {\rho}_{AB}^{(\tau)}  \cdot \left(  \sum_{i =1}^n P_{A,i}^{(n)} \otimes  P_{B,i}^{(n)} \,- \frac{1}{2} \overline P_{A,x}^{(n)} \otimes P_{B,x}^{(n)}  + \frac{1}{2} \left(\sum_{i\neq x} P_{A,i}^{(n)}  \right) \otimes P_{B,x}^{(n)}   \right)   } = 
$$
$$ 
= P_g^{(n)}  +  \tr{ \rho_{AB}^{(\tau)} \cdot \left( \frac{ - \overline P_{A,x}^{(n)} }{2} \otimes P_{B,x}^{(n)}  +  \frac{ \sum_{i\neq x} P_{A,i}^{(n)}  }{2}\otimes P_{B,x}^{(n)}   \right) }
= P_g^{(n)}  + \tr{\rho_{AB}^{(\tau)} \cdot \frac{ \mathbbm{1}_A - 2 \overline P_{A,x}^{(n)} }{2}  \otimes P_{B,x}^{(n)} }  ,
$$
\end{widetext}
where $P_{AB}^{mix}$ represents terms that provide positive contributions  to $P_g^{(2)}$. We have to find a value of $x$ that makes the second term of the last relation positive. Let $a_x$ and $b_x$ be the diagonal elements of $\overline P_{A,x}^{(n)}$ in the orthonormal basis $\{\ket{0}_A,\ket{1}_A\}$. We recall that $\rho_{AB}^{(\tau)} = ( \ketbra{0}{0}_A \otimes \rho_B'^{(\tau)}+ \ketbra{1}{1}_A \otimes \rho_B''^{(\tau)} )/2$
and we obtain
\begin{equation}\label{Pgodd}
P_g^{(2)} \!\geq \! P_g^{(n)}+  \mbox{Tr}_B{\left[ \!\left( \!\frac{1-2 a_x}{4} {\rho_B'}^{(\tau)} \!+\! \frac{1-2b_x}{4} \rho_B''^{(\tau)} \! \right)\cdot P^{(n)}_{B,x} \right]\!,}
\end{equation}
where the second term on the right-hand side of the inequality is definitely positive when $a_x,\, b_x \leq 1/2$. From  $\sum_i \overline P_{A,i}^{(n)}  = \mathbbm{1}_A $ follows that $\sum_{i=1}^n a_i = 1$ and $\sum_{i=1}^n b_i=1$. Therefore, if $a_x > 1/2 $ ($b_x > 1/2$), then $a_y \leq 1/2$ ($b_y \leq 1/2$) for any $y \neq x$. In order to fix the value of $x$, we must consider that $a_x$ and $b_x$ could be bigger than $1/2$ for two different values of $x$: let's say $x_a$ and $x_b$. Even in this ``worst-case'' scenario we still have $n-2$ other possible choices for $x$ such that $(1-2a_x),\, (1-2b_x) \geq 0$. We pick one of these values, and we call it $\overline x \in \{1,\dots,n \} \setminus \{x_a,x_b\}$. Finally, if we use ${\overline x}$ in the definition of the POVMs $\{ P^{(2)}_{A,i} \}_i$ and $\{ P^{(2)}_{B,i}\}_i$, from Eq. (\ref{Pgodd}) we obtain
\begin{equation}\label{Pgodd2}  
P_g^{(2)} \geq  P_g^{(n)}   \, .
\end{equation}

Equations (\ref{Pgeven}) and (\ref{Pgodd2}) show that, when we evaluate $C_B( \rho_{AB}^{(\tau)})$, the guessing probability of the ensembles generated by the $n$-output ME-POVMs is never bigger than the one that we obtain if we only consider the 2-output ME-POVMs:
$
C^{(2)}_B(\rho_{AB}^{(\tau)})= C_B(\rho_{AB}^{(\tau)}) \, .
$
Thanks to this result we can finally say that $C_A(\rho_{AB}^{(\tau)}(t)) \geq C_B(\rho_{AB}^{(\tau)}(t))$ and $C (\rho_{AB}^{(\tau)}(t))= C_A(\rho_{AB}^{(\tau)}(t))$. This result is valid if we consider $ \rho_{AB}^{(\tau)}$, but in general it is not true.

\section{Proof that $C_A({\rho}_{AB}^{(\tau)}) = C_A^{(2)} ({\rho}_{AB}^{(\tau)}) $}\label{2enoughA}

When we considered $C_A(\rho_{AB}^{(\tau)})$, we have seen that if the maximization over the ME-POVMs is considered only over the 2-output ones, the maximum is obtained for $\{P^{proj}_{A,i}\}_i=\{ \ketbra{0}{0}_A, \ketbra{1}{1}_A \}$. 
In order to complete the proof, we need to show that even if we consider general $n$-output ME-POVMs  (as in the definition (\ref{CAapp})), we don't get higher guessing probabilities of the corresponding output ensembles. In other words, if we use the definition
$$
C_A^{(2)} ({\rho}_{AB}^{(\tau)}) =\!\!\!\!\!\!\!\! \max_{ \left\{P_{A,i}\right\}_i \in \Pi_A^{(2)} \left( \rho_{AB}^{(\tau)} \right) } \!\!\!\!P_g \left( \mathcal{E} \left( \rho_{AB}^{(\tau)} ,\left\{ P_{A,i} \right\}_i \right)  \right)-\frac{1}{2} \, ,
$$
where $\Pi_A^{(2)}(\rho_{AB}^{(\tau)})$ contains only the 2-output ME-POVMs of $\rho_{AB}^{(\tau)}$, then $C_A(\rho_{AB}^{(\tau)}) = C_A^{(2)} ({\rho}_{AB}^{(\tau)}) $. 

To see this we can make the same analysis as done in Appendix \ref{2enough} for $C_B(\rho_{AB}^{(\tau)})$ but we switch the role of $A$ and $B$ in Eq. (\ref{2outputeven}) and Eq. (\ref{PAiPBi}) when $n$ is even and Eq. (\ref{PBodd}) and Eq. (\ref{PAodd}) when $n$ is odd. The definitions for $P_g^{(n)}$, $P_g^{(2)}$, $E_{1,2}$ and $O_{1,2}^x$ are preserved.

The guessing probability of an EES generated by a ME-POVM $ \{ P_{A,i}^{(n)}\}_i$ with an even number of outputs is
$$
P_g^{(n)} =\tr{{\rho}_{AB}^{(\tau)}\cdot \left( \sum_{i=1}^n {P}_{A,i}^{(n)} \otimes \overline P_{B,i}^{(n)}  \right)   } \, ,
$$
where $\{ \overline P_{B,i}^{(n)}\}_i$ is a POVM that maximizes the guessing probability in Eq. (\ref{Pg}).
The 2-output ME-POVM that provides a higher guessing probability is
\begin{equation}\label{PAiPBiA}
P^{(2)}_{A,1} = \sum_{i\in E_1}  P_{A,i}^{(n)} \, , \, \,\,
P^{(2)}_{A,2} = \sum_{i\in E_2} {P}_{A,i}^{(n)} \, .
\end{equation}
We define the following POVM on the system $B$
\begin{equation}\label{2outputevenA}
P^{(2)}_{B,1} = \sum_{i\in E_1} \overline P_{B,i}^{(n)} \, , \,\,\,
P^{(2)}_{B,2} = \sum_{i\in E_2} \overline P_{B,i}^{(n)} \, .
\end{equation}
Consequently, we consider the following inequality
\begin{eqnarray*}
&P_g^{(2)} \geq \tr{ {\rho}^{(\tau)}_{AB} \cdot \sum_{i=1,2} P^{(2)}_{A,i}  \otimes P^{(2)}_{B,i} } =  \\ 
&= P_g^{(n)}+ \tr{ {\rho}^{(\tau)}_{AB}  \cdot \sum_{k=1}^2 \sum_{i \neq j}^{  i,j\in E_k}    P_{A,i}^{(n)} \otimes  \overline P_{B,j}^{(n)}   } ,& 
\end{eqnarray*}
which shows that $P_g^{(2)} \geq  P_g^{(n)}$. 
If $n$ is odd, we use again the technique from Appendix \ref{2enough}, where we switch the role of $A$ and $B$, to obtain the inequality
$$
P_g^{(2)}\geq  P_g^{(n)}  + \tr{ \rho_{AB}^{(\tau)} \cdot \frac{ \mathbbm{1}_A - 2 P_{A,x}^{(n)} }{2}  \otimes \overline P_{B,x}^{(n)} }  ,
$$
where the right-hand side is greater than $P_g^{(n)}$ if $x$ is suitably chosen.
 
{
We underline that the results given in this section and Appendix \ref{CB2CA} suffice to state that $C^{(2)} (\rho^{(\tau)}_{AB})=C^{(2)}_A (\rho^{(\tau)}_{AB})\geq C^{(2)}_B (\rho^{(\tau)}_{AB})$.
}

\end{document}